\documentclass[%
 reprint,
 superscriptaddress,
 amsmath,amssymb,
prb,
]{revtex4-2}
\usepackage{comment}
\usepackage{amsfonts}
\usepackage{graphicx}
\usepackage{physics}
\usepackage{color}
\usepackage{amsthm}
\usepackage{tikz}
\usepackage{tikz-cd}
\usepackage{relsize}
\usepackage{hyperref}
\usepackage{caption}
\usepackage{subcaption}
\usepackage{dutchcal}
\usepackage{algpseudocode,algorithm}

\newcommand{\beq}{\begin{equation}\begin{aligned}}
\newcommand{\eeq}{\end{aligned}\end{equation}}

\newtheorem{claim}{Claim}[section]
\newtheorem{definition}{Definition}
\newtheorem{proposition}{Proposition}[section]

\preprint{}

\begin{document}

\title{Symmetry-enriched topological order from partially gauging symmetry-protected topologically ordered states assisted by measurements}

\author{Yabo Li}
\affiliation{C. N. Yang Institute for Theoretical Physics, State University of New York at Stony Brook, New York 11794-3840, USA}
\affiliation{Department of Physics and Astronomy, State University of New York at Stony Brook, New York 11794-3840, USA}
\author{Hiroki Sukeno}
\affiliation{C. N. Yang Institute for Theoretical Physics, State University of New York at Stony Brook, New York 11794-3840, USA}
\affiliation{Department of Physics and Astronomy, State University of New York at Stony Brook, New York 11794-3840, USA}
\author{Aswin Parayil Mana}
\affiliation{C. N. Yang Institute for Theoretical Physics, State University of New York at Stony Brook, New York 11794-3840, USA}
\affiliation{Department of Physics and Astronomy, State University of New York at Stony Brook, New York 11794-3840, USA}

\author{Hendrik Poulsen Nautrup}
\affiliation{Institute for Theoretical Physics, University of Innsbruck, Technikerstr, 21a, A-6020 Innsbruck, Austria}
\author{Tzu-Chieh Wei}

\affiliation{C. N. Yang Institute for Theoretical Physics, State University of New York at Stony Brook, New York 11794-3840, USA}
\affiliation{Department of Physics and Astronomy, State University of New York at Stony Brook, New York 11794-3840, USA}

\date{\today}

\begin{abstract}
Symmetry protected topological (SPT) phases exhibit nontrivial short-ranged entanglement protected by symmetry and cannot be adiabatically connected to  trivial product states while preserving the symmetry. In contrast, intrinsic topological phases do not need ordinary symmetry to stabilize them and their ground states exhibit long-range entanglement. It is known that for a given symmetry group $G$, the 2D SPT phase protected by $G$ is dual to the 2D topological phase exemplified by the twisted quantum double model $D^{\omega}(G)$ via gauging the global symmetry $G$. Recently it was realized that such a general gauging map can be implemented by some local unitaries and local measurements when $G$ is a finite, solvable group. Here, we review the general approach to gauging a $G$-SPT starting from a fixed-point ground-state wave function and applying a $N$-step gauging procedure. We provide an in-depth analysis of the intermediate states emerging during the N-step gauging and provide tools to measure and identify the emerging symmetry-enriched topological order (SET) of these states. We construct the generic lattice parent Hamiltonians for these intermediate states, and show that they form an entangled superposition of a twisted quantum double (TQD) with an SPT ordered state. Notably, we show that they can be connected to the TQD through a finite-depth, local quantum circuit which does not respect the global symmetry of the SET order. We introduce the so-called symmetry branch line operators and show that they can be used to extract the symmetry fractionalization classes (SFC) and symmetry defectification classes (SDC) of the SET phases with the input data $G$ and $[\omega]\in H^3(G,U(1))$ of the pre-gauged SPT ordered state. We illustrate the procedure of preparing and characterizing the emerging SET ordered states for some Abelian and non-Abelian examples such as dihedral groups $D_n$ and the quaternion group $Q_8$.
 
\end{abstract}

\maketitle
\tableofcontents

\section{Introduction}

Topological order first originated from the study of the fractional quantum Hall effect~\cite{tsui1982two,Laughlin1983}. It cannot be described by local order parameters and is beyond Landau's classification of matter. It exhibits ground-state degeneracy dependent on the topology of the underlying manifold and the excitations, displaying anyonic statistics~\cite{Wen1990,wen1990ground}.  More recently, it was recognized as possessing some kind of long-range quantum entanglement~\cite{ChenGuLiuWen2013}  and having nonzero topological entanglement~\cite{hamma2005bipartite,kitaev2006topological,levin2006detecting}.
In addition to fractional quantum Hall systems and certain spin liquids~\cite{balents2010spin}, there are models that manifest topological order, such as Kitaev's toric code and quantum double (QD) models~\cite{Kitaev1997}, their twisted versions~\cite{PhysRevB.87.125114}, Levin-Wen string-nets~\cite{LevinWen2005} and more recently fractons~\cite{nankishore2019fractons,pretko2020fracton}. Topological features that characterize such a phase of matter are robust to local perturbations, which is a property highly desirable in quantum memories~\cite{Terhal2015}. Some of the topological models also offer the capability of topological quantum computation (TQC) by exploiting the braiding of anyons~\cite{Kitaev1997,NayakSimonSternFreedmanDasSarma2008, CuiHongWang2014, Mochon2003}, which has emerged as one of the schemes for fault-tolerant quantum computation due to its inherent robustness.

Interestingly, from the perspective of adiabatic connection and quantum circuits, ground states of different phases at zero temperature cannot be connected by either adiabatic evolution or a finite-depth quantum circuits~\cite{ChenGuWen2010}. Intrinsic topologically  ordered states therefore cannot be created from a trivial ground state, such as product states, with a quantum circuit of finite depth. When the circuits are required to respect a certain global symmetry, trivial gapped phases can be further fine-grained into distinct classes: those that can be created from product states with symmetric, finite-depth circuits and those that cannot. The latter classes are referred to as nontrivial symmetry-protected topological (SPT) phases~\cite{gu2009tensor,pollman2010entanglement,chen2012symmetry}, and most of them can be classified by cohomology~\cite{ChenGuLiuWen2013,else2014classifying}.

Quantum technology has been constantly improving and evolving.  Several medium-scale quantum computers are available.
Recently, certain topologically ordered states, such as those of the toric model, were created by quantum circuits~\cite{satzinger2021realizing}, and furthermore, some braiding statistics has also been observed in experiments~\cite{2018Natur.559..205B, 2018Natur.559..227K, https://doi.org/10.48550/arxiv.2210.10255,xu2022digital, iqbal2023creation}.
Yet, preparation of high-fidelity ground states and precise manipulation of excitations in topological systems still remain  challenging in the current era of noisy intermediate-scale quantum (NISQ) devices~\cite{Preskill_2018}.

For the family of  the QD models and their twisted versions, i.e., twisted quantum double (TQD) models~\cite{PhysRevB.87.125114}, there is a well-known correspondence to models of SPT phases~\cite{ChenGuLiuWen2013, PerezGarciaWolfSanzVerstraeteCirac2008, JiangRan2017, Williamson-Verstraete2016} via a procedure called gauging~\cite{LevinGu2012,HaegemanVanAcoleyenSchuchCiracVerstraete2015,shirley2019foliated,  Kogut1997}.
When two quantum states are topologically distinct with respect to a symmetry $G$, then they cannot be transformed to each other with a finite-depth, piecewise local unitary transformation that preserves the symmetry. 
The classification and characterization of SPT phases with global symmetry $G$ in two dimensions is facilitated by a function of $g_i\in G$, namely a 3-cocycle $\omega_3(g_1,g_2,g_3)$, which is a representative element of the 3rd cohomology group of $G$, denoted by $H^3(G,U(1))$~\cite{ChenGuLiuWen2013}.
At the same time, the TQD model is also characterized by the 3-cocycle: inequivalent choices of the 3-cocycle give rise to distinct intrinsic topological phases~\cite{PhysRevB.87.125114}.
Indeed, the wave function of the TQD model with the gauge group $G$ is obtained by gauging the global symmetry $G$ in the corresponding SPT wave function. 
We note that gauging by itself is not a unitary operation; an interesting question is, therefore, whether such a map is physically possible.

It has been known that the ground state of the toric code can be efficiently prepared by measuring half of qubits in a 2D cluster state~\cite{raussendorf2005long-range}. 
The use of measurements thus can provide a route to creating long-range entangled states with finite-depth operations~\cite{tantivasadakarn2021long, PhysRevLett.127.220503, ashkenazi2022duality}.
Ref.~\cite{verresen2021schorodinger} demonstrated that ground states of QD models with $S_3$ and $D_4$ groups can be prepared through finite-depth local unitary operations supplemented with on-site measurements. 
It was also argued in Ref.~\cite{tantivasadakarn2021long} that a similar procedure should work for QD models with any solvable group $G$, which was later elaborated in Ref.~\cite{Bravyi2022a} through repeated rounds of finite-depth operations, where each round incorporates unitaries, measurements, feedforward, and  corrections. This scheme was further generalized to the general TQD models with solvable groups in~\cite{Hierarchy}, where the number of measurement rounds was classified for various topological orders, leading to a conjecture of a new hierarchy of topological orders when one includes measurements as an ingredient.
It is worth mentioning that further improvement is possible for the QD models with $D_4$ and $Q_8$ groups, which can be prepared with a single round of measurements,  feedforward, and corrections~\cite{tantivasadakarn2022shortest}.
Experimentally, measurement-based gauging is a promising method for realizing nontrivial topological orders in small-scale systems requiring only local unitary operations, mid-circuit measurements, and feedforward corrections~\cite{iqbal2023topological,foss-feig2023experimental, iqbal2023creation}.

The present work re-examines the measurement-based gauging from the perspective of group representation theory and provides a characterization of the transformation and emergence of SPT, SET, and intrinsic topological order during gauging.  In general, for a solvable group $G$, the corresponding TQD model can be prepared from a $G$-SPT through a multi-step gauging procedure. In this work, we provide two approaches that realize such an $N$-step gauging which reduces to a one-step gauging when $G$ is abelian or to a two-step gauging when $G$ is dihedral. For non-solvable groups, it is argued that the measurement-assisted gauging procedure cannot be implemented by a finite-depth circuit~\cite{Hierarchy}. 

Interestingly, we find that the intermediate states, that emerge during the multi-step gauging, can be naturally described as symmetry-enriched topological (SET) orders~\cite{2013PhRvB..87j4406E, MesarosRan2013, 2013arXiv1302.2634L, Barkeshli_2019}. We also show that, without respecting global symmetry, there is a finite-depth quantum circuit that takes the SET ground state to a ground state of a corresponding twisted quantum double model (TQD).

The essential data of an SET order, besides the intrinsic anyon theory $\mathcal{C}$, include the symmetry action as an automorphism on $\mathcal{C}$, the symmetry fractionalization class, and the defectification class~\cite{Barkeshli_2019}. A key result of our work is to characterize the resulting SET order given the 3-cocycle that describes the initial SPT wave function. 
If the emergent SET order has a global symmetry that \emph{does not} change the anyon type, we develop a general formalism based on symmetry branch line operators for the braiding phases between any abelian anyon in the theory and the anyons obtained from fusing point defects, exactly characterizing the symmetry fractionalization patterns. If the SET order we enter has a global symmetry that \emph{does} change anyon types, we conjecture the form and algebra of non-abelian symmetry branch line operators that can create the corresponding symmetry  defects. Then, by calculating the tensor product of such operators, one can derive the fusion rules of these symmetry defects, which we believe is sufficient to characterize the symmetry fractionalization patterns. We consider the dihedral SPT states as an example to illustrate this case.

The remainder of this paper is organized as follows. In Sec.~\ref{sec:GaugingMap}, we review the duality between SPT states with global symmetry group $G$ and ground states of a twisted quantum double model with a gauge group $G$ in two dimensions. This duality is given by a formal gauging map, which turns the global symmetry $G$ into a gauge symmetry. In Sec.~\ref{sec:N-stepGauging}, we describe the general procedure of $N$-step gauging $G$-SPT ordered states when $G$ is a solvable group in terms of an algorithm (see Algorithm~\ref{alg:gauging} below). In Sec.~\ref{sec:GaugingAbelian} and Sec.~\ref{sec:GaugingDihedral}, we discuss 1-step and 2-step gauging respectively, and consider Abelian and dihedral groups as illustrative examples. For the latter, we find that after the first gauging step, the system remains in a SET state where the remaining quotient group describes the global symmetry. Sec.~\ref{sec:SymmetryDefect} contains the discussion on symmetry properties of the emergent SET phases from the perspective of symmetry defects. Using the framework of symmetry branch lines, we relate the transformation of symmetry defects under gauging to properties of the SET phase. We give several examples to illustrate our formalism. In Sec.~\ref{sec:conclusion}.  we make some concluding remarks.
The Appendix provides materials that support the results in the main text. For example,
we provide a constant-depth unitary circuit to map an SET state to a TQD state in Appendix~\ref{sec:local unitary}.
 In Appendix~\ref{N-step}, We also give an alternative gauging prescription based on a different presentation of solvable groups which is alternative but equivalent to the standard one, as proven in Appendix~\ref{proof:solvale=seqnormal}.

\section{Fixed-point SPTs, twisted quantum doubles and gauging}\label{sec:GaugingMap}

On an oriented triangulated lattice $\Lambda$, given a finite group $G$, we assign a Hilbert space $\mathcal{H}_v=\{\sum_{g\in G}c_g\ket{g}_v|c_g\in \mathbb{C}\}$ to each vertex $v$. Then we can write a fixed-point $G$-SPT wave function. To do this, we first assign a group cocycle to each simplex, where $\omega$ is a representative in $H^3(G,U(1))$ respecting the cocycle condition,
\beq
    \frac{\omega(h,k,l)\omega(g,hk,l)\omega(g,h,k)}{\omega(gh,k,l)\omega(g,h,kl)}=1,
\eeq
for any $g,h,k,l\in G$.

\begin{figure}[h]
    \centering
    \includegraphics[width=0.8\linewidth]{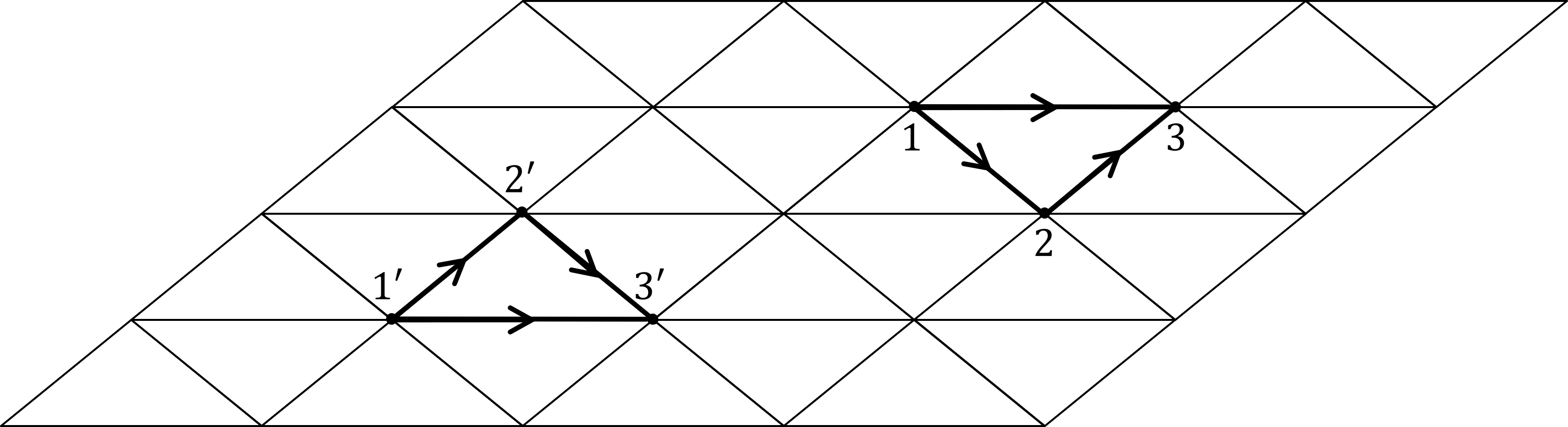}
    \caption{On an oriented triangulated plane, two typical simplexes with opposite orientations are shown. Their corresponding cocycles are $\omega(g_3g_2^{-1},g_2g_1^{-1},g_1)$ and $\omega(g_{3'}g_{2'}^{-1},g_{2'}g_{1'}^{-1},g_{1'})^{-1}$.  }
    \label{fig:simplex}
\end{figure}

The fixed-point SPT wave function is given by taking a product over all such cocycles,  
\beq
\label{eq:PsiSPT}
    \ket{\Psi_{\text{SPT}}}=\sum_{\{g_v\}}\prod_{\text{simplex}\,\Delta_{123}}\omega(g_3g_2^{-1},g_2g_1^{-1},g_1)^{s(\Delta_{123})}\bigotimes_{v} \ket{g_v}_v,
\eeq
where $s(\Delta)=\pm1$ indicates the orientation of a simplex $\Delta_{123}$ (with a given branching structure, and $\{1,2,3\}$ labels the vertices on the simplex), the tensor product runs over all vertices $v$ on the lattice, and all the configurations $\{g_v\}$ are summed over. Note that we use a convention from Ref.~\cite{williamson2016matrix}, which is slightly different from Ref.~\cite{ChenGuLiuWen2013}, for the sake of convenience in later discussions. This state can be obtained by the action of a unitary operator $U_{\omega}$  on the product state $\bigotimes_v\sum_{g}\ket{g}_v$,
\beq
    U_{\omega}=\sum_{\{g_v\}}\prod_{\Delta_{123}}\omega(g_3g_2^{-1},g_2g_1^{-1},g_1)^{s(\Delta_{123})}\bigotimes_{v} \ket{g_v}_v\bra{g_v}.
    \label{eq:cocycleunitary}
\eeq

We define the left/right action of $x$ on $\mathcal{H}_v$ as
\beq
    L^x_{+v}\ket{g}_v=\ket{xg}_v,~L^x_{-v}\ket{g}_v=\ket{gx^{-1}}_v.
\eeq
Then the global symmetry action (in our convention) $U^x\equiv\prod_v L^x_{-v}$ on SPT state yields
\beq
    U^x \ket{\Psi_{\text{SPT}}}&=\sum_{\{g_v\}}\prod_{\Delta} \omega^{s(\Delta)}(\{g_v\})\bigotimes_v\ket{g_v x^{-1}}_v
    \\
    &=\sum_{\{g_v\}}\prod_{\Delta} \omega^{s(\Delta)}(\{g_v x\})\bigotimes_v\ket{g_v}_v
    \\
    &\equiv\sum_{\{g_v\}}\prod_{\Delta}Amp(\{g_v\},x)\omega^{s(\Delta)}(\{g_v \})\bigotimes_v\ket{g_v}_v,
\eeq
where in the second line we used a change of variables and we have defined a phase factor $Amp$ in the fourth line. 

Suppose $M$ is the two-dimensional spatial manifold on which the Hilbert space is defined, and $I=\{x_3 | 0\leq x_3 \leq 1\}$ is an interval in the (Euclidean) time direction.
The manifold $M\times I$ is now three-dimensional.
We triangulate the $M\times I$ by 3-simplexes (tetrahedrons) with the constraint that each time slice at $x_3=0$ and $x_3=1$ matches the original two-dimensional lattice; 
see Fig.~\ref{fig:amplitude}.
The amplitude $\text{Amp}$ is computed once the triangulation of $M\times I$ is specified.
We now give more details.

\begin{figure}[h]
    \centering
    \includegraphics[width=0.8\linewidth]{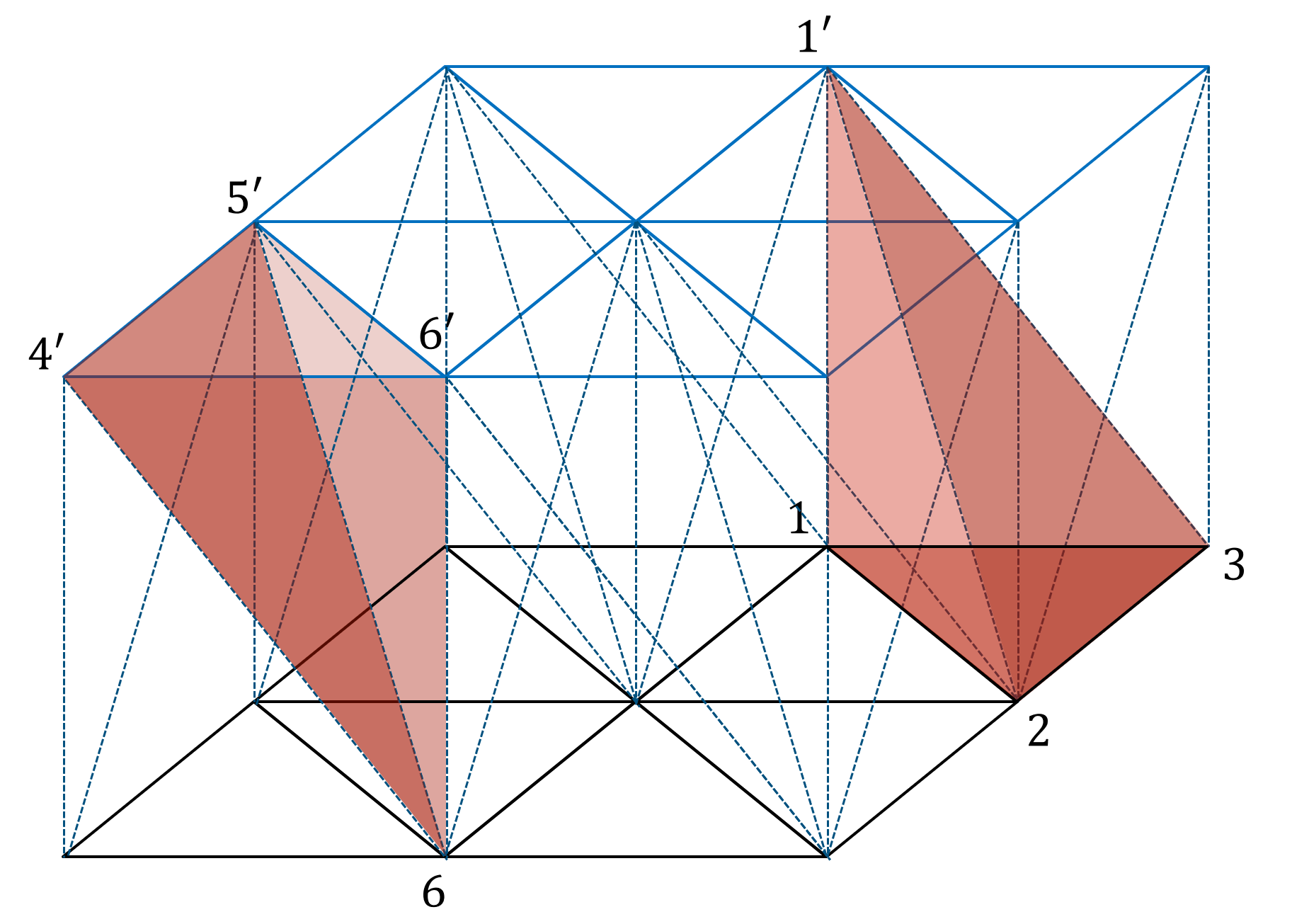}
    \caption{The phase factor $Amp(\{g_v\},x)$ can be given by the triangulation of such prisms, where $g_{v'}=g_v x$.   }
    \label{fig:amplitude}
\end{figure}

We assign a 3-cocycle to each tetrahedron as in Fig.~\ref{fig:tetrahedron}. The phase factor can be seen to be~\cite{ChenGuLiuWen2013}
\beq
    Amp(\{g_v\},x)=\prod_{\text{tetrahedron}}\omega(\text{tetra})^{s(\text{tetra})}.
\eeq
\begin{figure}[h]
    \centering
    \includegraphics[width=0.8\linewidth]{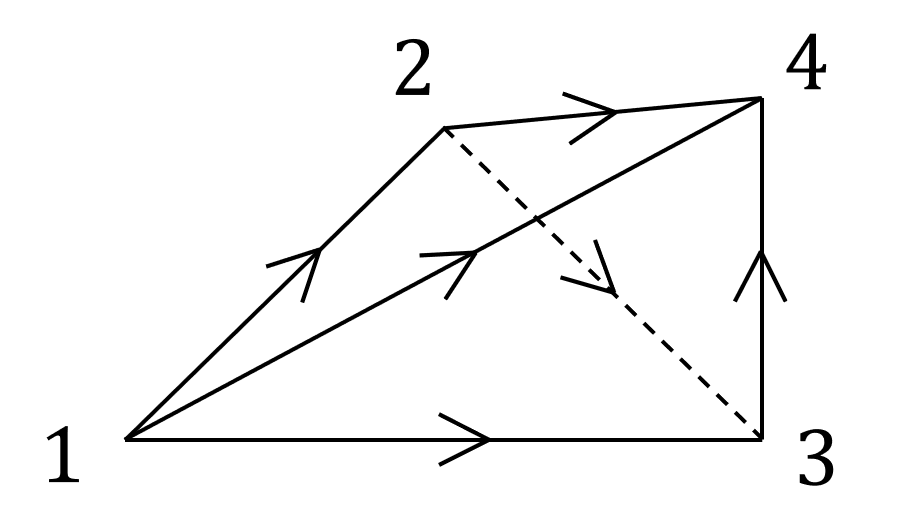}
    \caption{A positively oriented tetrahedron. The 3-cocycle assigned to it is $\omega(g_4g_3^{-1},g_3g_2^{-1},g_2g_1^{-1})$.}
    \label{fig:tetrahedron}
\end{figure}
When the spatial manifold is closed, using cocycle conditions, one can show that $Amp\equiv 1$. Therefore, the SPT state is invariant under the global symmetry transformation,
\beq
    U^x\ket{\Psi_{\text{SPT}}}=\ket{\Psi_{\text{SPT}}}.
\eeq
Due to this symmetry, this state can be written schematically as  (where we have suppressed the indices in $\Delta$ for simplicity),
\begin{align}
    |\Psi_{\text{SPT}}\rangle =&\sum_{\{g_v\}}\prod_{\text{simplex}\Delta}\omega(g_3g_2^{-1},g_2g_1^{-1},g_1)^{s(\Delta)}\bigotimes_{v} \ket{g_v}_v\nonumber\\
    =&\sum_{\{g_v\}} \Omega(\{g_vg_{v'}^{-1}\}) 
    \bigotimes_{v} |g_v\rangle_v, 
    \label{eq:SPT-wave-function}
\end{align}
where $\Omega$
denotes the product of cocycles, and $v$ and $v'$ are vertices connected by an edge, $\langle v, v' \rangle \in E$.

\subsection{Gauging global symmetry $G$}

\begin{figure}[h]
    \centering
    \includegraphics[width=0.8\linewidth]{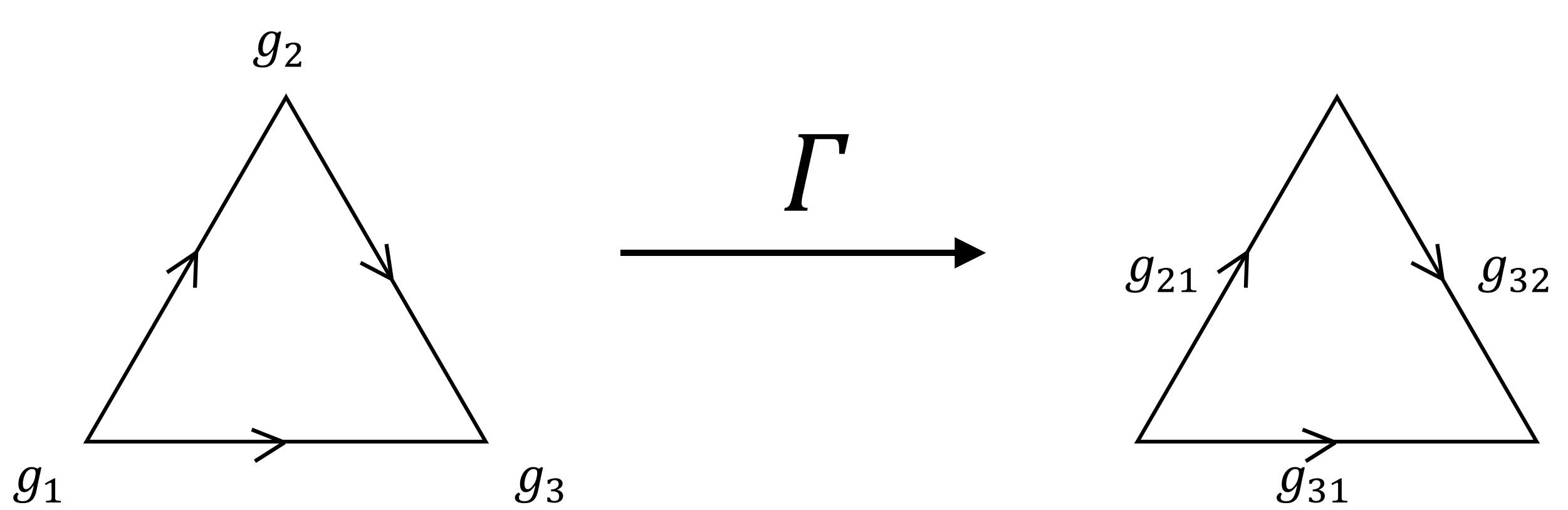}
    \caption{The gauging map $\Gamma$ maps from vertex DOFs to edge DOFs}
    \label{fig:gaugingmap}
\end{figure}

Under a gauging map as shown in Fig.~\ref{fig:gaugingmap}, the vertex degrees of freedom (DOFs) are mapped to the edge DOFs, 
\beq
    \Gamma:\ \  \ket{\{g_i\}}_v\rightarrow \ket{\{g_ig_j^{-1}\}}_e.
\eeq
This in turn maps the SPT state to an intrinsic topologically ordered state~\cite{YOSHIDA2017387}
 \beq
    \ket{\Psi_{\text{TQD}}}=\sum_{\{g_e\}}\Omega(\{g_e\})\bigotimes_e\ket{g_e},
    \label{TQDwfn}
 \eeq
 which is a ground state of the twisted quantum double (i.e., described by the Dijkgraaf-Witten theory) $D^{\omega}(G)$~\cite{dijkgraaf1990topological}.
The twisted quantum double can be formulated on a triangulated lattice with a Kitaev's Quantum Double-like Hamiltonian~\cite{PhysRevB.87.125114},
\beq
H=-\sum_{v}A_v-\sum_{p}B_p,
\eeq
where $v$ and $p$ stand for the vertices and plaquettes, respectively, on the lattice. The vertex operator  
\beq
A_v=\frac{1}{|G|}\sum_{g\in G}\big(\prod_{e\supset v}L^g_{\pm e}\big) \Tilde{W}_v^g,
\label{eq:vertexoperator}
\eeq
is hermitian and is a projector (see appendix~\ref{TQDprop}),
where $L^g_{+e}$ and $L^g_{-e}$ are left and right action of the group element $g$ on the edge $e$. 
When $e$ emanates from the vertex $v$ to another vertex, we apply $L^g_{-e}$ in Eq.~(\ref{eq:vertexoperator}), 
When $e$ flows to the vertex $v$, we apply $L^g_{+e}$.

\begin{figure}[h]
    \centering
    \includegraphics[width=0.8\linewidth]{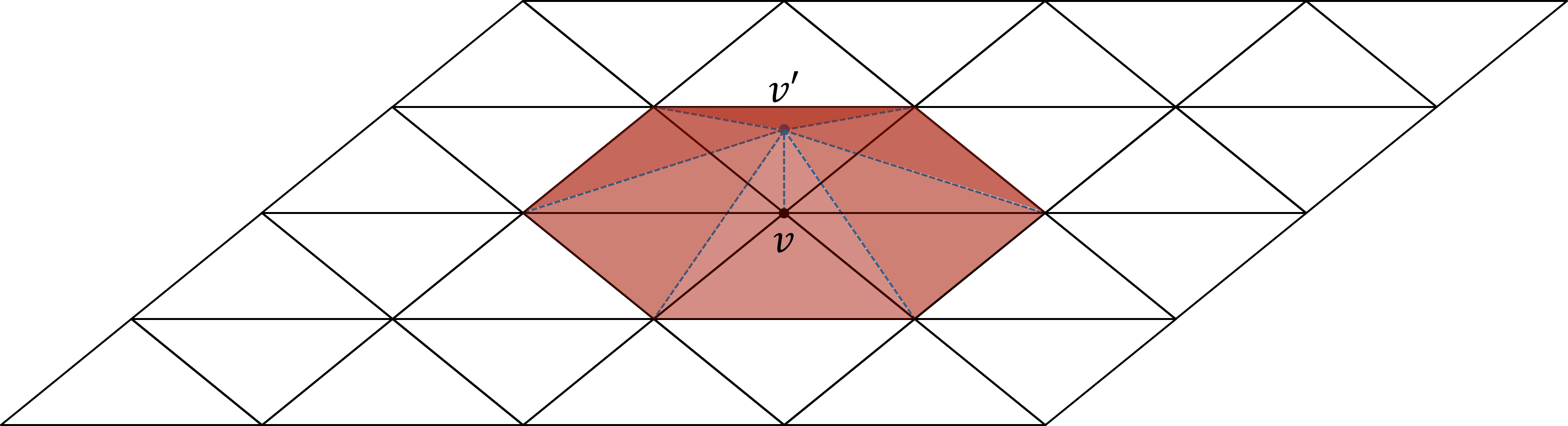}
    \caption{The phase $\Tilde{W}^g_v$ is defined as the multiplication of the phases corresponding to the tetrahedrons. Here, $g_{v'v}=g$.}
    \label{fig:W-phase}
\end{figure}

\begin{figure}[h]
    \centering
    \includegraphics[width=0.8\linewidth]{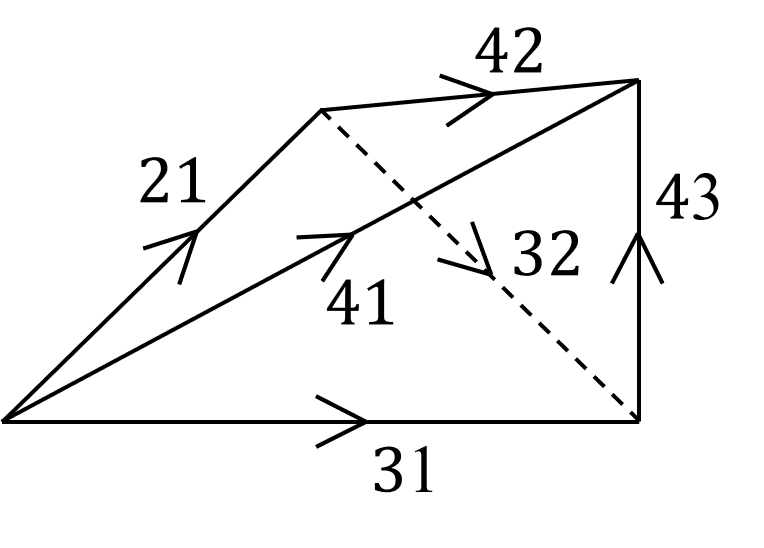}
    \caption{A negatively oriented tetrahedron. The 3-cocycle assigned to it is $\omega^{-1}(g_{43},g_{32},g_{21})$.}
    \label{fig:tetrahedron2}
\end{figure}

The phase $\Tilde{W}^g_v$ is a product of the cocycles corresponding to the tetrahedrons with appropriate orientations in the prism in Fig.~\ref{fig:W-phase}, where the correspondence between a tetrahedron and a 3-cocycle is established in Fig.~\ref{fig:tetrahedron2}. 
Furthermore, this phase factor is the commutator between the right action of $g$ on vertex $v$ and the unitary operator introduced in Eq.~(\ref{eq:cocycleunitary}),
\beq
    \Tilde{W}^g_v=(\prod_{e\supset v}L^g_{\pm e})^{\dagger}U_{\omega}(\prod_{e\supset v}L^g_{\pm e})U_{\omega}^{\dagger},
\eeq
where we use $L^g_{+e}$ ($L^g_{-e}$) when the edge $e$ ends at (emanates from) vertex $v$.

The plaquette operator is
\beq
B_p=\delta\Big(\prod_{e\in p} g_e,1\Big),
\label{eq:BpTQD}
\eeq
where $\delta(x,y)$ is the Kronecker delta function. The resultant state from the gauging map is the ground state of this Hamiltonian,
\beq
    H \,\Gamma(\ket{\Psi_{\text{SPT}}})= E_0 \,\Gamma(\ket{\Psi_{\text{SPT}}}).
\eeq

The local excitations of TQD model are fractional charges called anyons, which can be classified by a unitary modular tensor category (UMTC); see, e.g.,~\cite{lin2021generalized}. 
One thing to remark is that the convention here is slightly different from the one used in \cite{YOSHIDA2017387} and \cite{PhysRevB.87.125114} for the sake of convenience in later discussions.

\subsection{Gauging a subgroup of $G$}
One can introduce a gauging map $\Gamma_N$ that corresponds to gauging only a normal subgroup $N$ of $G$. 
We have the quotient group $Q=G/N$ with an embedding 
\beq
\label{eq:embedding}
    s: Q\rightarrow G.
\eeq
Any element $g\in G$ has a unique decomposition $g=qn$, where $q\in s(Q)$ and $n\in N$.  Under the map, the normal part of vertex DOFs are mapped to edge DOFs, as illustrated in Fig.~\ref{fig:gaugingmap_n},
\beq
    \Gamma_N:\ \ \ket{\{g_i\}}_v\rightarrow\ket{\{q_i\}}_v\otimes\ket{\{n_i n_j^{-1}\}}_e.
    \label{eq:Gamma_N}
\eeq
\begin{figure}[h]
    \centering
    \includegraphics[width=0.9\linewidth]{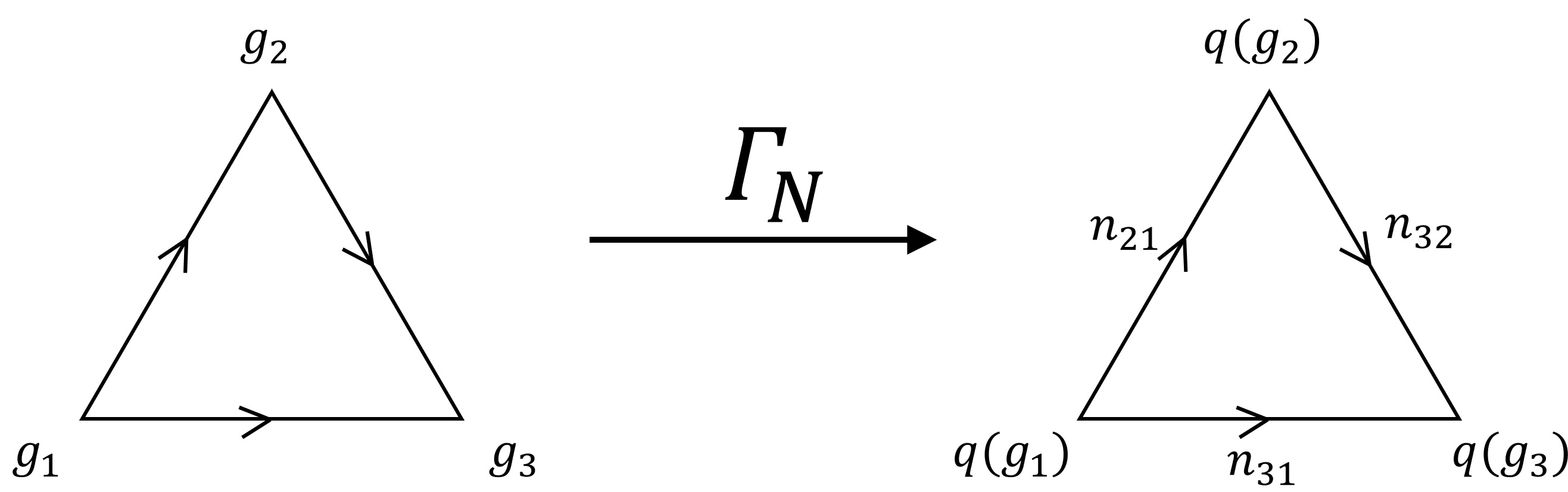}
    \caption{The gauging map $\Gamma_N$ maps the normal part from vertex DOFs to edge DOFs.}
    \label{fig:gaugingmap_n}
\end{figure}
This maps the SPT state to
 \beq
    \ket{\Psi_{\text{SET}}}=\sum_{\{q_v\}}\sum_{\{n_e\},\text{fluxless}}\Omega(\{q_i n_e q_j^{-1}\})\bigotimes_v\ket{q_v}\bigotimes_e\ket{n_e}.
 \eeq

One point to notice is that the above state has a global symmetry $Q$ under action $U^x_Q=U^x \sigma^x$, where $U^x\equiv \prod_v L^x_{-v}$ is the right action of $x\in s(Q)$ on all vertices  (e.g., $g\rightarrow g x^{-1}$) and $\sigma^x\equiv \prod_e\sigma^x_e$ is the conjugation by $x$ on all edges defined as
\beq
    \sigma^x_e\ket{n}_e=\ket{xnx^{-1}}_e.
\eeq
The second point is that the state $\ket{\Psi_{\text{SET}}}$ is a ground state of a Kitaev's Quantum Double-like Hamiltonian,
\beq
H=-\sum_{v}A_v-\sum_{p}B_p-\sum_{v}K_v,
\label{eq:QDSET}
\eeq
where $v$ and $p$ stand for the vertices and plaquettes on the lattice. 

The vertex operator is 
\beq
A_v=\frac{1}{|N|}\sum_{n\in N}\sum_{q\in s(Q)}\big(\prod_{e\supset v}L^n_{\pm e}\big) W_v^{q nq^{-1}}\ket{q}_v\bra{q}.
\label{eq:AvSET}
\eeq
The phase $W^g_v$ is the product of the cocycles corresponding to the tetrahedrons with appropriate orientations of the prism in Fig.~\ref{fig:W-phase1}, where the correspondence between tetrahedron and 3-cocycle is established in Fig.~\ref{fig:tetrahedron4}.
\begin{figure}[h]
    \centering
    \includegraphics[width=0.8\linewidth]{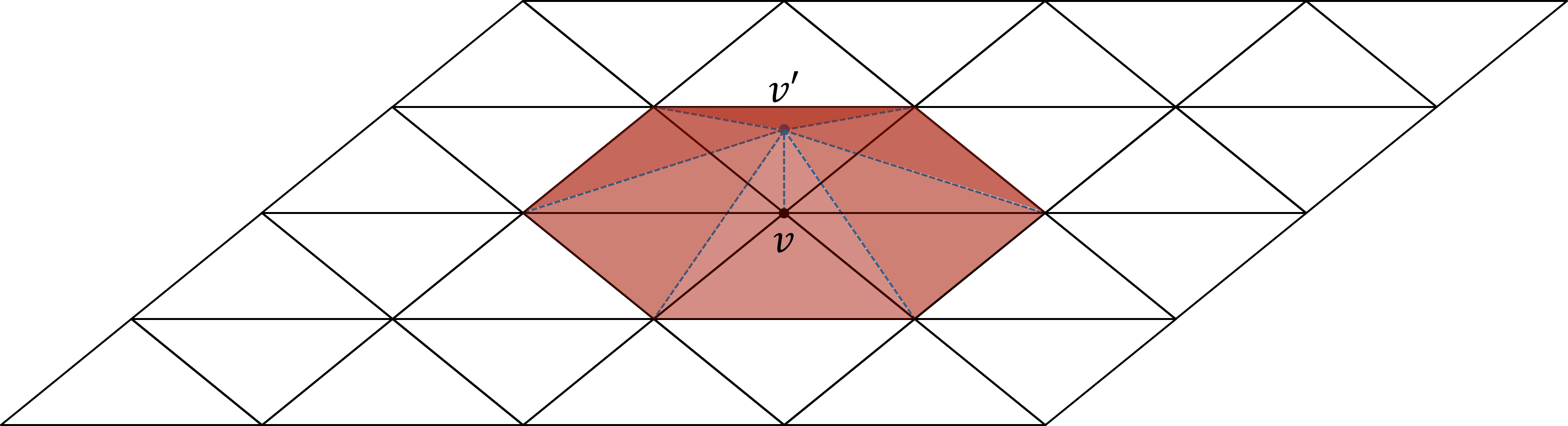}
    \caption{The phase $W^g_v$ is defined as the multiplication of the phases corresponding to the tetrahedrons. $h_{v'v}=x$, $q_{v'}=q_v$.}
    \label{fig:W-phase1}
\end{figure}
\begin{figure}[h]
    \centering
    \includegraphics[width=0.8\linewidth]{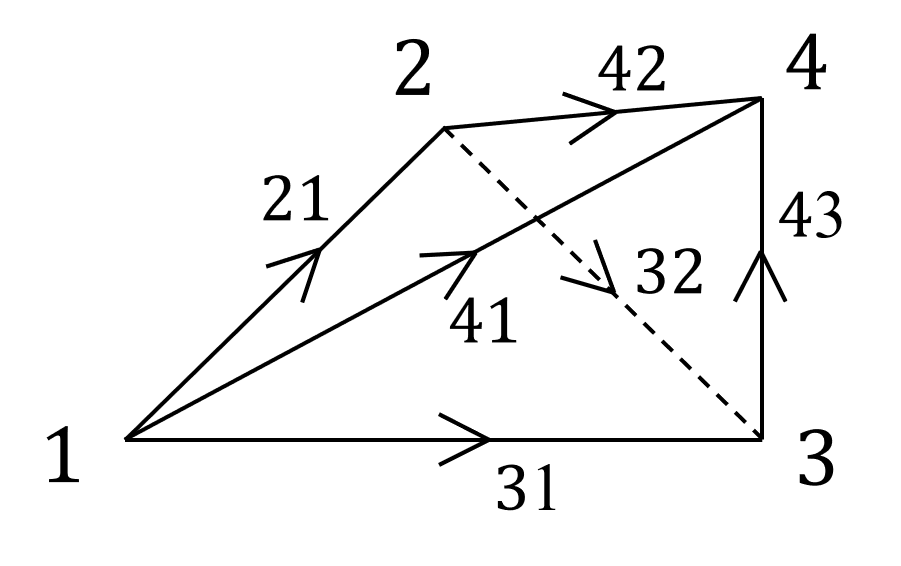}
    \caption{A negatively oriented tetrahedron. The 3-cocycle assigned to it is $\omega^{-1}(g_{4}h_{43}g_3^{-1},g_{3}h_{32}g_2^{-1},g_2 h_{21}g_1^{-1})$.}
    \label{fig:tetrahedron4}
\end{figure}
The plaquette operator is simply the following,
\beq
B_p=\delta\Big(\prod_{e\in p} n_e,1\Big).
\label{eq:BpSET}
\eeq
Suppose $|Q|=m$, we label the embedding of $Q$ in $G$ as $s(Q)=\{q_0,\cdots,q_{m-1}\}$, where  $q_0=1$. The additional vertex operator $K_v$ is
\beq
    K_v=\frac{1}{m}\sum_{k,l=0}^{m-1} W_v^{q_k q_l^{-1}}\ket{q_k}_v\bra{q_l}.
    \label{eq:KvSET}
\eeq

We can always apply a finite-depth local unitary to bring all the vertex DOFs to the identity element (see appendix~\ref{sec:local unitary}) such that the state becomes
\beq
    \ket{\Psi_{\text{TQD}}}=\sum_{\{n_e\}}\Omega(\{n_e\})\ket{\{n_e\}}_e\otimes_v\ket{1}_v.
\eeq
This is a TQD state with the 3-cocycle $\nu(n_1, n_2,n_3)$ being the restriction of $\omega(g_1,g_2,g_3)$ on subgroup $N$. 
Thus we obtain the anyons and their braiding, which is the same as in $D^{\nu}(N)$. 
Furthermore, the state $\ket{\Psi_{\text{SET}}}$ is essentially a ground state in an SET phase with the global $Q$-symmetry. 

\subsection{Classification of SETs}
\label{sec:SET classification}

Here we briefly review some terminology relevant to SET phases for the convenience of later discussions. This section will be based on Refs.~\cite{Barkeshli_2019, chen2017symmetry}. In general, assuming the symmetry preserves locality, an SET phase is determined by its anyon set (UMTC) $\mathcal{C}$, its $Q$-symmetry action as an automorphism on $\mathcal{C}$, symmetry fractionalization class (SFC) and symmetry defectification class (SDC)~\cite{Barkeshli_2019}. We first assume the symmetry actions are unitary and always give the trivial automorphism on $\mathcal{C}$, i.e., the symmetry does not change anyon types. 
Consider a state $\ket{\Psi_{a,b,c,...}}$ with anyons $\{a,b,c,...\}$ present sufficiently far away from each other on a sphere. 
Therefore, the anyons $\{a,b,c,...\}$ should be able to fuse into the vacuum charge. The symmetry operators respect the multiplication rules $U(g)U(h)=U(gh)$. Under our assumptions, the symmetry operator can be decomposed as some local unitaries $U_a(g), U_b(g), ...$, near the anyons,
\beq
    U(g)\ket{\Psi_{a,b,c,...}}=U_a(g)U_b(g)\cdots\ket{\Psi_{a,b,c,...}}.
\eeq

Each local symmetry action can be projective,
\beq
    U_a(g)U_a(h)=\eta_a(g,h)U_a(gh),
\eeq
where the phase $\eta_a(g,h)$ only depends on anyon type $a$ and satisfy
\beq
    \eta_a(g,h)\eta_b(g,h)=\eta_c(g,h),
\eeq
whenever the multiplicity is $N^c_{ab}\neq 0$. 
It is proved in Ref.~\cite{Barkeshli_2019} that a phase $\eta_a(g,h)$ with the above properties is related to the braiding phase between anyon $a$ and some other abelian anyon $\mathcal{w}(g,h)$ in the theory,
\beq
    \eta_a(g,h)=B(\mathcal{w}(g,h),a).
\eeq

We note that one can redefine the local unitary $U_a(g)$ by an arbitrary phase factor $v_a(g)$,
\beq
    U'_a(g)=v_a(g)U_a(g),
\eeq
where $v_a(g)$ satisfies
\beq
    v_a(g)v_b(g)=v_c(g),
\eeq
whenever the multiplicity is $N^c_{ab}\neq 0$. Again, the phase factor $v_a(g)$ can be written as a braiding phase between anyon $a$ and an abelian anyon $\mathcal{v}(g)$, i.e., $v_a(g)=B(\mathcal{v}(g),a)$.
The abelian anyon after the redefinition will be
\beq
    \mathcal{w}'(g,h)=\frac{\mathcal{v}(g)\mathcal{v}(h)}{\mathcal{v}(gh)}\mathcal{w}(g,h).
\eeq
Further, according to the associativity condition $(U_a(g)U_a(h))U_a(k)=U_a(g)(U_a(h)U_a(k))$, we have
\beq
    \frac{\mathcal{w}(h,k)\mathcal{w}(g,hk)}{\mathcal{w}(gh,k)\mathcal{w}(g,h)}=1.
\eeq

To conclude, the distinctive patterns of symmetry fractionalization are characterized by the class $[\mathcal{w}(g,h)]$ in cohomology group $H^2(Q,\mathcal{A})$, where $\mathcal{A}$ is the group formed by abelian anyons via fusion algebra~\cite{chen2017symmetry}.

Another way to see the symmetry fractionalization classes is to construct a $Q$-graded category $\mathcal{C}^{\times}_Q$ from the anyon theory including the point defects
\beq
    \mathcal{C}_Q=\bigoplus_{q\in Q}\mathcal{C}_q,
\eeq
where $\mathcal{C}_{\openone}=\mathcal{C}$ and 
$\openone$ denotes the identity element in $Q$. A distinctive $Q$-graded category $\mathcal{C}^{\times}_Q$ is a candidate for an SET order. According to Ref.~\cite{Barkeshli_2019}, when the symmetry \emph{does not} change anyon types, one can always choose an abelian defect from each sector $\mathcal{C}_q$ and label it as $0_q$. 
The fusion of defects respects the group multiplication structure, 
\beq
    0_g \times 0_h = \mathcal{w}(g,h)_{\openone} \times 0_{gh},
\eeq
for some abelian anyon $\mathcal{w}$. Furthermore,
\beq
    a_g \times b_h =\sum_{c\in \mathcal{C}} N^{c}_{a b} \,c \times \mathcal{w}(g,h)_{\openone} \times 0_{gh}.
    \label{eq:anyonfusion}
\eeq
This abelian anyon $\mathcal{w}(g,h)$ is exactly what we have defined above for projective phase $\eta_a(g,h)$. The class $[\mathcal{w}(g,h)]$ is in the cohomology group $H^2(Q,\mathcal{A})$, which classifies the SFC.

In generic cases, when the symmetry \emph{does} change anyon types as an automorphism of $\mathcal{C}$,
\beq
    \rho: Q \rightarrow Aut(\mathcal{C}),
\eeq
it turns out that not every sector $\mathcal{C}_q$ has an abelian object. Therefore, we cannot write the fusion rule as in Eq.~\eqref{eq:anyonfusion}. 

The fusion rule of a $Q$-graded category $\mathcal{C}^{\times}_Q$ can be written as
\beq
    a_g \times b_h =\sum_{c_{gh}}N^{c_{gh}}_{a_g b_h}c_{gh}.
\eeq
Consequently, each element $[\mathcal{t}]\in H^2_{\rho}(Q,\mathcal{A})$ specifies a potential way of modifying $\mathcal{C}^{\times}_Q$ (the SET order) via
\beq
    a_g \times b_h =\mathcal{t}(g,h)\times\sum_{c_{gh}}N^{c_{gh}}_{a_g b_h}c_{gh}.
    \label{eq:SFCgeneral}
\eeq

Therefore, in generic cases, the potential symmetry fractionalization classes are elements of an $H^2_{\rho}(Q,\mathcal{A})$ torsor. In this work, we will not analyze the SDC in detail, and we simply note that one can enter a different SDC by stacking a $Q$-SPT state onto the SET state. In our framework, after gauging the normal subgroup, a global $x$-transformation will locally serve as an automorphism $\rho_x$ of $\mathcal{C}$, mapping an anyon with flux $n$ to an anyon with flux $xn x^{-1}$. Later in this work, we will analyze the phases of some SETs in which the automorphism $\rho_x$ on $\mathcal{C}$ is either trivial or nontrivial.

\section{$N$-step gauging of 2D SPT via measurement}\label{N-stepGauging}
\label{sec:N-stepGauging}

\begin{figure*}
\begin{minipage}{\linewidth}
\begin{algorithm}[H]
\begin{algorithmic}
\Require (a) Solvable $G$ with $1=G_0<G_1<\cdots <G_N=G$ such that $Q_k\equiv G_k/G_{k-1}$ is abelian $\forall k=1,...,N$.\newline 
(b) $G$-SPT fixed point state $\ket{\Psi_\textrm{SPT}}$ (defined in Eq.~\ref{eq:SPT-wave-function}) on a lattice $(V,E)$ with vertices $i,j\in V$ and edges $\langle i,j\rangle\in E$.
\newline
(c) $U^{(i,j)}_{Q_k}$ as defined in Eq.~\ref{eq:gauge_U}\newline
(d) Generalized Pauli operators as defined in Eqs.~\ref{eq:gauge_X} and~\ref{eq:gauge_Z} acting on the abelian subspace defined by an embedding $s_k$ of $Q_k$ in $G_k$.

\vspace{0.2cm}
\State $k \gets 1$
\State $\textrm{(1) Add ancillas: } \ket{\Psi_\textrm{gauge}} \gets \ket{\Psi_\textrm{SPT}}\otimes \prod_{\langle i,j\rangle\in E}\ket{e}_{\langle i,j\rangle} \textrm{ where } e\in G \textrm{ is the identity element.}$\vspace{0.2cm}
\While{$k\leq N$}
    \State $\textrm{(2) Entangle vertex and edge DOFs: } \ket{\Psi_\textrm{gauge}}\gets \prod_{\langle i,j\rangle \in E}U^{(i,j)}_{Q_k}\ket{\Psi_\textrm{gauge}}$\vspace{0.2cm}
    \State $\textrm{(3) Measure vertex DOFs in the basis given by Eq.~\ref{eq:gauge_fourier_basis} with outcomes $\{\tilde{i}^{v}_1,...,\tilde{i}^{v}_{l_k}\}_{v\in V}$ (neglecting normalization): }\newline \ket{\Psi_\textrm{gauge}}\gets\prod_{v\in V} \ket{\tilde{i}^{v}_1,...,\tilde{i}^{v}_{l_k}}\bra{\tilde{i}^{v}_1,...,\tilde{i}^{v}_{l_k}}\ket{\Psi_\textrm{gauge}} $\vspace{0.2cm}
    \State $\textrm{(4) Correct for random measurement outcomes by applying $Z$-operators on a set of edges $E_{\rm Cor}$: }\newline \ket{\Psi_\textrm{gauge}}\gets\prod_{e\in E_{\rm Cor}}\mathbb{Z}_e \ket{\Psi_\textrm{gauge}}$, where $\mathbb{Z}_e=\prod_{j=1}^{l_k} Z_{j}^{-{p}_{j;e}}$ (specifically, $E_{\rm Cor}$ and ${p}_{j;e}$) can be deduced from the measurement outcomes, given the symmetries (e.g., Eq.~\ref{eq:zn-sym-const}) and so-called \emph{transmutation rules} (e.g., Figs.~\ref{fig:Zn-transmutation},~\ref{fig:transmutation_n} or~\ref{fig:transmutation_q}). \vspace{0.2cm}
    \State $k=k+1$ \Comment{$\ket{\Psi_\textrm{gauge}} \textrm{ is an SET state as analyzed in Sec.~\ref{sec:SymmetryDefect} } \forall k<N$}
\EndWhile

\end{algorithmic}
\caption{$N$-step gauging via measurements}
\label{alg:gauging}
\end{algorithm}
\end{minipage}
\end{figure*}


In this section, we present the procedure to gauge a $G$-SPT state of a group $G$ that can be factorized into $N$ abelian groups with $N$ steps. (We note in this section, $N$ refers to the number of steps rather than a normal subgroup. But it should be clear from the context.) 
A similar method was proposed by~\cite{Hierarchy} and~\cite{Bravyi2022a}. 
In~\cite{Hierarchy}, the authors considered the solvable group $G$ and its derived series which consists of normal subgroups which are commutator subgroups of the previous group in the series. 
They proposed a gauging procedure for a particular sequence of normal subgroups. 
In~\cite{Bravyi2022a}, on the other hand, they proposed to implement the gauging procedure for a solvable group inductively, i.e., implement the gauging of a cyclic group, assuming the remaining quotient group is already gauged.  
In our procedure, we do not restrict ourselves to a particular derived series for the solvable group. 
This in turn helps us to prepare different types of SETs.
We give the steps for gauging a $G$-SPT state explicitly. 

Before presenting the gauging procedure in Algorithm~\ref{alg:gauging}, let us go through the most relevant definitions first. A group $G$ is a solvable group if there are subgroups $1=G_0<G_{1}<\cdots<G_N=G$ such that $G_{k-1}$ is normal in $G_{k}$, and $G_{k}/G_{k-1}\equiv Q_k$ is abelian for $k=1,\cdots,N$. 
Given the embedding map $s_k$ from each $Q_k$ into $G_k\subset G$, every element $g\in G$ can be written as 
\beq \label{eq:embedding-N}
    g=q_{N} q_{N-1} \cdots q_2 q_1,
\eeq
where $q_k\in s_k(Q_k)$. 
Similarly, for another group element $h$, we have the decomposition $h=\tilde{q}_N\cdots \tilde{q}_1$. Under this convention, we write down the  multiplication between $g$ and $h^{-1}$ as
\beq
    gh^{-1}&=q_N\cdots q_3 q_2 \Big(q_1 \tilde{q}^{-1}_1\Big) \tilde{q}^{-1}_2\tilde{q}^{-1}_3 \cdots \tilde{q}^{-1}_N.
    \label{ghinverse0}
\eeq

Now, let us define the relevant unitaries and measured observables that will be used in the gauging procedure. We typically consider a state defined on a lattice $(V,E)$ with vertices $i,j\in V$ and edges $\langle i,j\rangle\in E$ where the local Hilbert space $\ket{g}$ depends on the group $G$ and is labeled by its group elements $g\in G$. Given an embedding $s_k$ of $Q_k$ into $G_k$ as described above, we can define for all $Q_k$ the following unitary controlled on vertices and targeting the shared edge:
\beq
    \,\,\,\,\,\,\,\,U^{(i,j)}_{Q_k}:=\sum_{g_1,g_2,g_3\in G} &\ket{g_1,g_2}_{i,j}\bra{g_1,g_2}\otimes \\ \quad &\ket{q_k(g_1)g_3q_k(g_2)^{-1}}_{\langle i,j\rangle}\bra{g_3},\label{eq:gauge_U}
\eeq
where $q_k\in s_k(Q_k)$. This unitary will be used to entangle vertex DOFs with edge DOFs

Measurements of abelian subgroups will play an important role in the gauging procedure which is why we will now introduce the generalized Pauli-observables for an abelian group $Q_k\equiv \prod_{j=1}^{l_k}Z_{d_k^j}$. 
(We note that it should be clear from context whether the symbol $Z$ represents a group or a Pauli operator.) 
Any element $q\in Q_k$ can be written as $a_1^{i_1}a_2^{i_2}\cdots a_{l_k}^{i_{l_k}}$ where $a_j^{d_k^j}=e$ $\forall j=1,...,l_k$. Given this representation, we write the local Hilbert space basis as $\ket{i_1,...,i_{l_k}}\equiv\ket{a_1^{i_1}a_2^{i_2}\cdots a_{l_k}^{i_{l_k}}}$. This allows us to define the following generalized local Pauli operators by their action on this basis:
\begin{align}
X_1^{t_1}\otimes\cdots\otimes X_{l_k}^{t_{l_k}}\ket{i_1,...,i_{l_k}}&=\ket{i_1\oplus t_1,...,i_{l_k}\oplus t_{l_k}}\label{eq:gauge_X}\\
Z_1^{t_1}\otimes\cdots\otimes Z_{l_k}^{t_{l_k}}\ket{i_1,...,i_{l_k}} &= \omega_1^{i_1t_1}\cdots \omega_{l_k}^{i_{l_k}t_{l_k}}\ket{i_1,...,i_{l_k}},\label{eq:gauge_Z}
\end{align}
where $i_j\oplus x$ indicates addition modulo $d_k^j$ and $\omega_j$ is the $d_k^j$-th root of unity $\forall j=1,...,l_k$. Importantly, this allows us to define a Fourier-transformed basis as follows,
\beq
\ket{\tilde{i}_1,...,\tilde{i}_{l_k}}=Z_1^{i_1}\otimes\cdots Z_{l_k}^{i_{l_k}} \ket{+}\label{eq:gauge_fourier_basis},
\eeq
where $\ket{+}=\sum_{a_1^{i_1}a_2^{i_2}\cdots a_{l_k}^{i_{l_k}}\in Q_k}\ket{i_1,...,i_{l_k}}$.

Note that in Algorithm~\ref{alg:gauging}, the local Hilbert space dimension is given by the non-abelian group $G$, so we understand all the above unitaries and bases as defined on an embedded subspace given by $s_k$. See the discussion above Eq.~(\ref{eq:embedding-N}).

We will now use the above equations to implement an $N$-step gauging procedure. We will gauge the $G$-symmetry of the state defined on the vertices of a lattice sequentially in $N$ steps. 
We present the procedure in Algorithm~\ref{alg:gauging} and consider the details below:
 \begin{enumerate}
	\item[(1)] 
	\emph{Include ancillas.} Add ancillas in the state $\ket{e}$, where $e\in G$ is the identity element, on the edges between the vertices. 
	\item[(2)]
	\emph{Entangle gauge and matter DOFs.} Apply the following 2-controlled-shift operators with controls $c_1, c_2$ on neighboring vertices (oriented as $c_2\rightarrow c_1$) and the target $t$ on the in-between ancilla:
	\begin{align}
	   \begin{split}
	      \qquad U_{Q_1}=\sum_{g_1,g_2,g_3\in G}&\ket{g_1,g_2}_{c_1,c_2}\bra{g_1,g_2}\otimes \\
	  &\quad\ket{q_1(g_1)g_3q_1(g_2)^{-1}}_{\langle 1,2\rangle}\bra{g_3}. 
	   \end{split} 
	\end{align}
	Here we have used $q_1(g)$ to denote the part of the decomposition $g$ which lies in $Q_1$; that is, for $g=q_N\cdots q_1$ with $q_k\in s_k(Q_k)$, $q_1(g)=q_1$.
	\item[(3)]
	\emph{Measure $\{X_1,X_2,\dots,X_{l_1}\}$ on matter DOFs.} 
  After measurement of the quotient part on each vertex (i.e., in the bases defined in $\{X_{j}\}$),  with the outcome being $\{X_{j}=\omega_j^{-p_j}\}_{j=1}^{l_1}$ on a vertex ($\omega_j$ being $d_j$-th root of unity), there is a corresponding phase factor $\prod_{j=1}^{l_1} \omega_{j}^{-p_j i_j}$ from the wave function overlap in step (3) of Algorithm \ref{alg:gauging}. These phase factors can be seen as some abelian chargeons on vertices. See an example in Eq.~\eqref{eq:post-measured-state}.

   \item[(4)] \emph{Correct phase factors.} The phase factors arising from step (3) can be corrected as they can be expressed as a product of phase operators acting on the edge DOFs (see e.g., Fig.~\ref{fig:Zn-transmutation}). Therefore, we can  
   apply counter $Z$-operators on a set of edges $E_{\rm Cor}$, i.e., $\prod_{e\in E_{\rm Cor}}\mathbb{Z}_e$, where the exact form of $\mathbb{Z}_e=\prod_{j=1}^{l_k} Z_{j}^{-p_{j;e}}$ can be deduced from the measurement outcomes. For example, we can move the phase factors on a vertex by performing $Z$ operators on its neighboring edge and by doing this repeatedly we can move all the phase factors to one single vertex (via the transmutation rule, see, e.g. Fig.~\ref{fig:Zn-transmutation}), therefore annihilating them altogether, due to the symmetry constraint (see Eq.~\ref{eq:zn-sym-const} and the discussion below it). The set of such edges is an example of $E_{\rm Cor}$,  but it is not necessarily optimized. (See also the following two sections for concrete examples.) After measurement and correction, the vertex DOF is mapped from $\ket{g}$ to $\frac{1}{|Q_1|}\sum_{q'_1}\ket{g^{(1)}q'_1}$, where $g=q_N\cdots q_2 q_1$ and $g^{(1)}=q_N\cdots q_2$. The resultant state is a $G_1$-SET ground state.

    \item[(5)]
    \emph{Repeat the procedure of entangling gauge DOFs on edges and matter DOFs on vertices.} Apply the following unitary similar to  before:
	\begin{align}
	   \begin{split}
	     \qquad U_{Q_2}=\sum_{g_1,g_2,g_3\in G}&\ket{g_1,g_2}_{c_1,c_2}\bra{g_1,g_2}\otimes\\
	   &\quad\ket{q_2(g_1)g_3q_2(g_2)^{-1}}_{\langle 1,2\rangle}\bra{g_3}.  
	   \end{split} 
	\end{align}
	\item[(6)]
	\emph{Measure $\{X_{1},...,X_{l_2}\}$ on the matter DOFs and correct the corresponding phase factors from the measurement.} This results in a  $G_2$-SET ground state.

    \item[(7)]
    \emph{Repeat this process for all except the last quotient group $Q_N$.} 
    \item[(8)] \emph{At the last step, apply the gauging and measurement procedure for $Q_N$.} Specifically,
    first apply
      \begin{align}
      \begin{split}
         \,\,\,\,U_{Q_N}&=\sum_{g_1,g_2,g_3\in G} \ket{g_1,g_2}_{c_1,c_2}\bra{g_1,g_2}\otimes \\
        &\qquad\qquad\ket{q_N(g_1)g_3q_N(g_2)^{-1}}_{\langle1,2\rangle}\bra{g_3}, 
      \end{split}
       \end{align}
        measure $\{X_{1},...,X_{l_N}\}$, and then correct the corresponding phase factors. This gives us a $G$-TQD state.
\end{enumerate}

It is worth remarking that in the expressions above, we always use the multiplication rules of the entire group $G$. 
For instance, if we take $q_2,q'_2\in Q_2$, their product $q_2q'_2$ is not necessarily in $Q_2$ (it is in $Q_2$ only when the extension $G_2/G_1=Q_2$ is central). 
This seemingly makes the remaining global symmetry algebra non-closed, i.e., $g^{(1)}g'^{(1)}$ would produce components in the $Q_1$ subgroup (recall that $g^{(1)}=q_N\cdots q_2$ for some $q_k\in Q_k$). 
Nonetheless, instead of $\ket{g^{(1)}}$ in the above step (4), we have $\frac{1}{|Q_1|}\sum_{q'_1}\ket{g^{(1)}q'_1}$, which can absorb the potential $Q_1$ components, making the multiplication closed. 
We can therefore define the global symmetry group of $G_1$ SET as such. Moreover, the state after applying $U_{Q_2}$ is symmetric with $X_{q_2}$.
It also follows that the phase operators resulting from the measurement of $X_{q_2}$ can be corrected as the global symmetry gives a constraint on measurement outcomes, which will be discussed later.

In the following, we will consider the 1-step gauging for abelian groups in Sec.~\ref{sec:GaugingAbelian} to illustrate correction processes. 
Then we will consider the 2-step gauging for dihedral groups in Sec.~\ref{sec:GaugingDihedral}. 
The two-step and multi-step gauging can be also applied to abelian groups as well, and in the intermediate steps, SET states can emerge. 
We will discuss the phase of such SET states in Sec.~\ref{sec:GaugingDihedral}. Then in Sec.~\ref{sec:SymmetryDefect}, we introduce the framework of the symmetry defect branch line and discuss the SET phases for several more examples.

We give an alternative procedure for the $N$-step gauging in Appendix~\ref{N-step}. This gauging procedure is implemented for a group $G$ which admits sequential normal subgroups (see Appendix~\ref{N-step} for definition). 
This criterion is in fact equivalent to the group $G$ being solvable (see Appendix~\ref{proof:solvale=seqnormal} for a proof). 
The two procedures differ in the way in which the product of group elements are written down (compare Eq.~\eqref{ghinverse0} and Eq.~\eqref{ghinverse}). 
Apart from that, in the procedure given in this section, we gauge the quotient groups in every step, while in the procedure given in Appendix~\ref{N-step} we gauge normal subgroups in each step. 
As mentioned earlier, in Ref.~\cite{Hierarchy}, the commutator subgroups of the previous group in the series of a solvable group $G$ are gauged successively.

\section{1-step gauging: Abelian groups}
\label{sec:GaugingAbelian}

In this section, we review the gauging procedure for abelian groups \cite{tantivasadakarn2021long, Hierarchy}. We start with the SPT state given in Eq.~\eqref{eq:SPT-wave-function}.
The gauging map can be implemented by first transferring the corresponding group elements to edges and then by measuring the vertices, where our state is projected to a quantum double state with (unwanted) charges whose configuration is given by the measurement outcomes. 
The excitation due to the randomness of measurement is then corrected by a certain finite-depth procedure.  
The steps are described in more detail as follows:
\begin{enumerate}
	\item[(0)] 
	\emph{Prepare the SPT state on vertices.} We use local control-phase gates to prepare the SPT state from a direct product state.
	\item[(1)] 
	\emph{Include ancillas.} Add ancillas in the state $\ket{e}$, where $e\in G$ is the identity element, on edges between adjacent vertices. The ancillas become the gauge DOFs
	\item[(2)]
	\emph{Entangle gauge and matter DOFs} Apply the following controlled-controlled-shift operators with controls $c_1\, \&\, c_2$ on the neighboring vertices of an edge $e$ (oriented as $c_2\rightarrow c_1$) and the target $t$ being the ancilla on the edge between the two controls:
	\begin{equation}
	    \,\,\,\,\,\,\,U_{G}=\sum_{g_1,g_2\in G}\ket{g_1,g_2}_{c_1,c_2}\bra{g_1,g_2}\otimes L^{g_1}_{+e} L^{g_2}_{-e},
	    \label{eq:controlgate}
	\end{equation}
	At this point, the (pre-measurement) state is
	\begin{align}
    \,\,\,\,\,\,\,\,\,\,\,\ket{\Psi_{\text{pre}}}=
    \sum_{\{g_v\}} \Omega(\{g_v\}) 
    \bigotimes_{v\in V} |g_v\rangle_v 
    \bigotimes_{\langle vv'\rangle \in E} | g_v g_{v'}^{-1} \rangle_{vv'}. \  
    \end{align}

	\item[(3)]
	\emph{Choose a measurement basis in the $G$ algebra, then project the matter DOFs onto the basis via measurement.} A natural basis can be chosen if we order elements in $G$ as an ordered list $(g_0,g_1,\cdots,g_{n-1})$, where $g_0=1$, $n=|G|$. (Note that the subscript $j$ in $g_j$ here denotes the labeling of the group elements of $G$, not the vertex.) Then we simply use the Fourier basis $\ket{k}=\sum_{j=0}^{n-1} \exp{{2\pi ij k}/{n}}\ket{g_j}/\sqrt{n}=Z^k|+\rangle$ to perform measurements on vertex DOFs, where we have defined $|+\rangle\equiv \sum_{j=0}^{n-1} |g_j\rangle/\sqrt{n}$. When $G=Z_n$, we project the matter DOFs onto this basis via   measuring the generalized (qudit) $X$ operator.

	\item[(4)]
	\emph{Correct excitations in the $G$ twisted quantum double.} The correction can be done with a finite-depth circuit.
 
\end{enumerate}

We give more explanation on the procedure for the case with $G=Z_n$ below.
For the $Z_n$ group, the wave function can be written using the qudit system.
The basis vector $|a\rangle$ ($a \in \{0,...,n-1 \text{~mod}~n\}$) and the generalized Pauli operators satisfy
\beq
Z|a\rangle = \omega^a |a \rangle , \quad X|a\rangle = |a+1\rangle  , \quad ZX = \omega XZ,
\eeq
with $\omega = \exp(2\pi i /n)$, which is much simpler than the general abelianized basis in Eqs.~(\ref{eq:gauge_X}) and (\ref{eq:gauge_Z}).
Starting from a $Z_n$ SPT on a triangulated lattice, we first add ancillas to all the edges in a product state with $\ket{0}$. Then we apply the controlled gate in Eq.~(\ref{eq:controlgate}), which is a set of controlled-$X$ gates. 
Then the gauge DOFs are as given by the gauging map in Fig.~\ref{fig:gaugingmap}. 
The next step is to disentangle matter DOFs by measuring the  $X$ operator on all vertices. 
After measurements, the matter DOF at a vertex $v$ is projected onto $Z^{k_v}_v\ket{+}_v$, where $k_v=0,1,\cdots,n-1$, and $\ket{+}=\sum_{i}\ket{i}/\sqrt{n}$. 
Suppose the measurement outcome on the vertex $v$ is {$X_v=\omega^{-k_v}$}. 
We write the basis associated with measurement outcomes $\{k_v\}_{v\in V}$ as 
\begin{align}
    |M\rangle := \bigotimes_v Z^{k_v}_v |+\rangle_v .
\end{align}
Then the total wave function after measurements (with the gauge part being projected to $\langle M |\Psi \rangle_{\text{pre}}$) is written as
\begin{align}
\label{eq:post-measured-state}
\ket{\Psi}_{\text{post}}=&\Big(\sum_{\{g_v\}}\Omega(\{g_v\})\prod_{v}
{z^{-k_v}}(g_v) \bigotimes_{e}\ket{g_v g_{v'}^{-1}}_e\Big) \nonumber \\
& \quad \bigotimes_v {Z^{k_v}}\ket{+}_v,
\end{align}
where $z^{a}(g_v):=\bra{g_v}Z^{a}\ket{g_v} = (\omega^{a})^{g_v}$
and $\Omega$ is the phase factor inherited from the SPT state (i.e., a product of 3-cocycles). 
Note that $vv'$ is the edge $e$; many edges can share a vertex $v$ but the factor $z^{k_v}$ only appears once. 
Due to measurement, the vertex DOFs have been disentangled from the edge DOFs, and the edge DOFs form a state that is a ground state of the $Z_n$ twisted quantum double in the flux-free sector up to a factor $z^{k_v}(g_v)$, which can be interpreted as an $\mathbf{e}^{k_v}$ chargeon on the vertex $v$.

In what follows, we describe how to remove the excitations in $|\Psi\rangle_{\text{post}}$.
First, the set of measurement outcomes is restricted to $\big[\sum_v k_v\big]_n=0$, with $[x]_n$ being $x$ mod $n$, due to the global symmetry of the SPT state.
The global symmetry implies
\begin{align}
    \prod_{v \in V} X_v | \Psi_{\text{pre}} \rangle  = | \Psi_{\text{pre}} \rangle, \ 
\end{align}
so it should be satisfied that
\begin{align} \label{eq:zn-sym-const}
    \langle M | \Psi_{\text{pre}} \rangle 
    & =  \langle M | \prod_{v \in V} X_v | \Psi_{\text{pre}} \rangle,  \nonumber \\
    &=\Big( \prod_{v} \omega^{-k_v}\Big) \langle M | \Psi_{\text{pre}} \rangle,  
\end{align}
which gives the constraint $\prod_{v} \omega^{-k_v} = 1$, meaning $\big[\sum_v k_v\big]_n=0$.

Next, the measurement with $|M\rangle$ gives us a phase $\prod_{v} \omega^{-k_v g_v}$
when contracted with the basis $\otimes_v |g_v \rangle$.
Due to the constraint $\big[\sum_v k_v\big]_n=0$, one can always find a set of paths such that we can rewrite 
{the phase factor as $\prod_{v} \omega^{-k_v g_v}$, or equivalently $\prod_{v} z(g_v)^{-k_v}$,}
in terms of the phase operator $Z_e$ supported on the paths. 
Concretely, we use a type of relations, which we call the {\it transmutation rules}, illustrated
in Fig.~\ref{fig:Zn-transmutation}. 
For $G=Z_n$, the relation is
\begin{align}
    z(g_v)=z(g_{v'})z(g_v g_{v'}^{-1}).
\end{align}
We apply the phase operator on the paths to remove the chargeons. Given that these operators commute, they can be applied all at once. Hence, our gauging procedure assisted by measurement requires only finite time steps or a finite-depth quantum circuit (with intermediate measurements).

\begin{figure}[h]
    \centering
    \includegraphics[width=\linewidth]{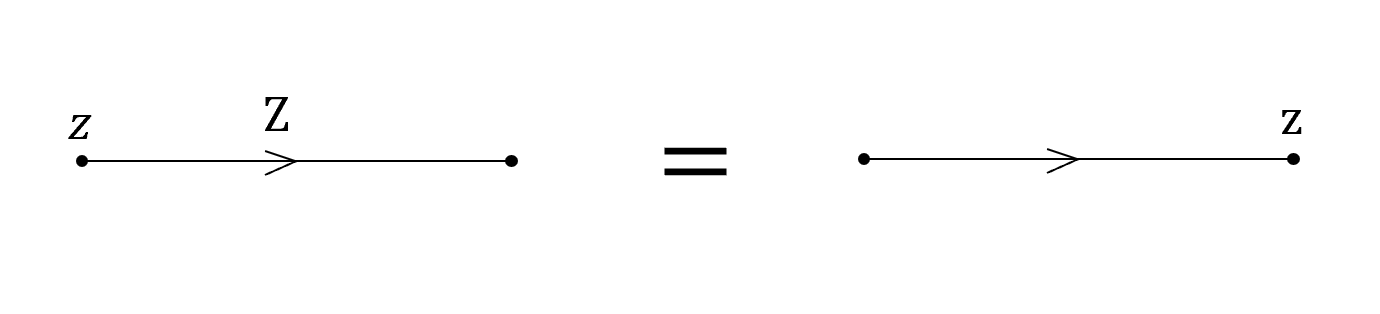}
    \caption{The transmutation rule of $z$ factor according to relation $z(g_{v'})z(g_v g_{v'}^{-1})=z(g_v)$.}
    \label{fig:Zn-transmutation}
\end{figure}

Let us give two remarks here. 
The reason that we can correct the state by moving all factors to one vertex is because of the fact that all chargeons in 
$D^{\omega}(Z_n)$ are abelian anyons. 
This procedure can be straightforwardly generalized to $Z_n \times Z_m \times\cdots$ group, where we measure $X \times 1 \times\cdots$, $1 \times X \times\cdots$, ..., etc. on all vertices after we entangle gauge and matter DOFs. This occurred previously in the general $N$-step gauging in Sec.~\ref{sec:N-stepGauging}. However, to explain the detailed correction there would incur cumbersome notations. The example of $Z_n$ in this section should now make the procedure clearer.
Different measurement outcomes will give rise to different chargeons in the flux-free sector, which are all abelian anyons. (Note that we do not have fluxons, as we began with a  flat-flux configuration followed by the controlled-controlled operation that does not create fluxons.)  Therefore, the state after measurement is still correctable within finite steps.

\section{2-step gauging: Dihedral group and intermediate SET states}
\label{sec:GaugingDihedral}

When we attempt to gauge nonabelian SPT states using measurement, although one can always choose a suitable basis such that the factors are {\it correctable}, one crucial problem is that the phase factors that arise from measurement do not necessarily correspond to abelian anyons as in the $D^{\omega}(Z_n)$ case above; this makes correction with a finite depth circuit a nontrivial problem.
In Ref.~\cite{verresen2021schorodinger, Bravyi2022a}, there were two different ways proposed to prepare the ground state of the $S_3$ quantum double model. 
In Ref.~\cite{Bravyi2022a}, a $Z_3$ toric code ground state is prepared first, and it is coupled to the $Z_2$ product state using controlled gates.
Then the $Z_2$ part is gauged via the measurement-assisted one-step gauging in Ref.~\cite{tantivasadakarn2021long}.

We will show in this section that, for the symmetry group $G$ being the extension of two abelian groups, by choosing some abelianized basis, we can still perform a 2-step gauging procedure on $G$-SPT states via measurement. 
In the case with $G=S_3$, our procedure would be equivalent to first preparing the $Z_3$-TQD ground state, and then coupling to the $Z_2$-SPT state using entangling gates and controlled gates. The correction process for the 2-step gauging is still fairly simple, i.e.,  via  finite-depth quantum circuits. The complete procedure to gauge the abelian $N$-symmetry (i.e., the normal subgroup) of a $G$-SPT state and then to gauge the quotient $Q$-symmetry of an SET is as follows:
\begin{enumerate}
	\item[(1)] 
	\emph{Include ancillas.} Add ancillas in the state $\ket{e}$, where $e\in G$ is the identity element, on edges between adjacent vertices. The ancillas will become the gauge DOFs

	\item[(2)]
	\emph{Entangle gauge and matter DOFs.} Apply the following controlled-controlled-shift operators with controls $c_1\, \&\, c_2$ on neighboring vertices (oriented as $c_2\rightarrow c_1$) and the target $t$ on the ancilla on the in-between edge $e$:
	\begin{equation}
	    U_{N}=\sum_{g_1,g_2\in G}\ket{g_1,g_2}_{c_1,c_2}\bra{g_1,g_2}\otimes L^{n(g_1)}_{+e} L^{n(g_2)}_{-e}.
	\end{equation}
	The purpose of this step is to mimic the gauging map in Eq.~(\ref{eq:Gamma_N}).
	\item[(3\,\&\,4)]
	\emph{Measure $X_{(n)}$ on matter DOFs and correct the $z_n$ factors.} After measurement, with the outcome being $X_{(n)}=\omega^{-k}$, there is a corresponding phase factor $z_n^k$. Using the transmutation rule for $z_n$, one can correct all those factors by moving them to one single vertex, resulting in an SET ground state.

	\item[(5)]
	\emph{Further entangling the quotient part of the gauge and matter DOFs.} We apply a controlled-conjugate operator with the target $e$ being the ancilla (oriented as $c_2\rightarrow c_1$), and the control being $c_2$:
    \begin{equation}
	    U_{Q}=\sum_{g_1,g_2\in G}\ket{q(g_1),q(g_2)}_{c_1,c_2}\bra{q(g_1),q(g_2)}\otimes L^{q(g_1)}_{+e} L^{q(g_2)}_{-e},
	\end{equation}
    where $q(g)$ denotes the quotient part of $g$ via an embedding in Eq.~\eqref{eq:embedding}. Notice that the normal part of the matter DOF has been wiped out by measuring $X_{(n)}$, while the quotient  part $Q=G/N$ still remains, which makes the above controlled-gates possible to implement. The edge DOFs are now $\{q(g_1) n(g_1)n(g_2)^{-1}q(g_2)^{-1}\}=\{g_1g_2^{-1}\}$.

    \item[(6\,\&\,7)] \emph{Measure $X_{(q)}$ on matter DOFs, and correct $z_q$ factors.} Their correction is straightforward; we apply $Z_{(q)}$ operators on edges to move all $z_q$'s to one vertex.

\end{enumerate}
In the following, we will apply the above procedure to several cases.

\subsection{Gauging $S_3$ SPT}

The $S_3$ group is $G=\langle a,x|a^3=e, x^2=e, xax=a^{-1}\rangle $. Any element $g\in G$ can be written as $g=x^i a^j$, where $i=0,1$, $j=0,1,2$. We define the decomposition of a group element respectively as
\begin{align}
    n(x^i a^j)&=a^j,
    \label{eq:n}
    \\
    q(x^i a^j)&=x^i,
    \label{eq:q}
\end{align}
with the former being the normal part ($N=Z_3$), and the latter being the quotient part ($s(Q)=s(Z_2)$) of $S_3$. 
We then define the shift operator in each part as
\beq
    X_{(n)}&=\sum_{i,j}\ket{x^i a^{j+1}}\bra{x^i a^j},
    \\
    X_{(q)}&=\sum_{i,j}\ket{x^{i+1} a^j}\bra{x^i a^j}.
    \label{eq:x}
\eeq

The phase operators, which are known as the clock operators, in each respective part, are
\beq
    Z_{(n)}&=\sum_{g}z_n(g)\ket{g}\bra{g},~z_n(x^i a^j)=\omega^j,
    \\
    Z_{(q)}&=\sum_{g}z_q(g)\ket{g}\bra{g},~z_q(x^i a^j)=(-1)^i,
    \label{eq:z}
\eeq
where $\omega=e^{i\frac{2\pi}{3}}$. 
The gauging step (2) transforms the ancilla DOF on edge $e=\langle v,v'\rangle$ from identity to $n(g_v)n(g_{v'})^{-1}$. 

Then in step (3) we measure $X_{(n)}$ on all the vertex DOFs.
Suppose the measurement outcome is $X_{(n)}=e^{-i\frac{2\pi k_v}{3}}$ on vertex $v$ (where $k_v=0,1,2$). The state after the measurement is projected into 
\beq
    \ket{\Psi_3}=&\sum_{\{g_v\}}\Big(\prod_v z_{n}^{-k_v}(g_v)\Big)  \Omega(\{g_v g_{v'}^{-1}\})\ket{\{n(g_v)n(g_{v'})^{-1}\}}_e\\
    &\bigotimes_v\Big(Z_{(n)}^{k_v}\big(\sum_{r\in N}\ket{q(g_v)r}_v\big)\Big).
\eeq
The phase factor $\prod_v z_{n}^{-k_v}(g_v)$ depends on the measurement outcomes $\{k_v\}$. In order to  correct them, we employ the transmutation rules for $z$ factors 
\beq
    z_n(h)z_n(n(g)n(h)^{-1})=z_n(g).
\eeq
As in the case with $Z_n$ in the previous section, we have $[\sum_v k_v]_3 =0$ due to the $N$-symmetry. By inserting corresponding numbers of $Z_n$ operators on the edges, we can move the factors $z_n$ on vertices around and cancel them altogether. Equivalently, one can simultaneously apply $Z_n$ operators supported on strings whose endpoints correspond to nontrivial measurement outcomes. 
This gives us the state
\beq
    \ket{\Psi_{4}}=&\sum_{\{g_v\}}  \Omega(\{g_v g_{v'}^{-1}\})\ket{\{n(g_v)n(g_{v'})^{-1}\}}_e\\
    &\bigotimes_v\Big(Z_{(n)}^{k_v}\big(\sum_{r\in N}\ket{q(g_v)r}_v\big)\Big).
\eeq

\begin{figure}[h]
    \centering
    \includegraphics[width=0.8\linewidth]{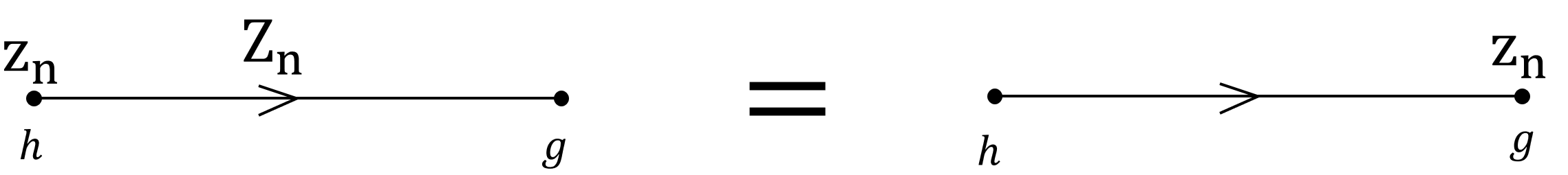}
    \caption{The transmutation rule for $z_n$ factors on step (4).}
    \label{fig:transmutation_n}
\end{figure}

After the gauging step (5), the edge DOFs are conjugated and shifted by $U_Q$, giving rise to
\beq
    \ket{\Psi_5}
    =&\sum_{\{g_v\}}  \Omega(\{g_v g_{v'}^{-1}\})\ket{\{g_v g_{v'}^{-1}\}}_e\\ 
    &\bigotimes_v 
    \Big( Z^{k_v}_{(n)}\sum_{r \in N}\ket{q(g_v) r}_v \Big). \label{eq:psi-5}
\eeq

Step (6) is similar to  step (3). By measurements, the state is projected to
\beq
    \ket{\Psi_6}=&\sum_{\{g_v\}}\Big(\prod_v z_{q}^{-m_v}(g_v)\Big)  \Omega(\{g_v g_{v'}^{-1}\})\ket{\{g_v g_{v'}^{-1}\}}_e\\
    &\bigotimes_v\Big(Z_{(q)}^{m_v}Z_{(n)}^{k_v}\big(\sum_{g\in G}\ket{g}_v\big)\Big),
\eeq
where we have assumed that $X_{(q)}=e^{-i\frac{2\pi m_v}{2}}$ from the measurement on vertex $v$ (where $m_v=0,1$). In order to  correct the phase factor $\prod_v z_{q}^{-m_v}(g_v)$, we employ the transmutation rules for $z_q$ factors: 
\beq
    z_q(h)z_q(gh^{-1})=z_q(g),
\eeq
which is illustrated in Fig.~\ref{fig:transmutation_q}. This rule, just as the rule for $z_n$ in step (4), allows us to move all the $z_q$ factors to a single vertex and annihilate them. This is guaranteed by the $Q$-global symmetry in $|\Psi_5\rangle$ (see Appendix~\ref{sec:q-global-sym}),
\beq
\Big( \prod_{v} X_{(q)} \Big)
|\Psi_5 \rangle = |\Psi_5 \rangle,
\eeq
which implies that the measurement outcomes satisfy $[\sum_v m_v]_2=0$ in this case, as the global symmetry is $Z_2$. After  we apply the corresponding  correcting phase factors to edges, we  thus obtain a $G$-TQD ground state.
\begin{figure}[h]
    \centering
    \includegraphics[width=0.8\linewidth]{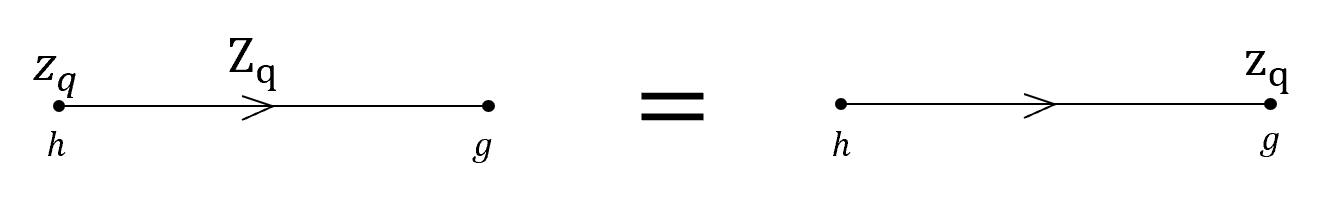}
    \caption{The transmutation rule for $z_q$ factors on step (7).}
    \label{fig:transmutation_q}
\end{figure}

\subsection{$Z_3$ SET with $Z_2$ Symmetry}
\label{sec:Z3SETwithZ2}

 Let us begin by recalling that the gauging map $\Gamma_N$ in Eq.~(\ref{eq:Gamma_N}) gauges the normal subgroup $N$ of $G$. 
 In our procedure, after we correct  the $z_n$ factors in step (3), the state is essentially a ground state of  the $Z_3$ SET phase. 
In what follows, we first look into the entanglement structure of the wave function after gauging the normal subgroup $Z_3$.
Then we identify the class of this SET phase, namely,  the unitary modular tensor category (UMTC) $\mathcal{C}$ that contains all anyonic excitations, the $Z_2$ symmetry action as an automorphism of $\mathcal{C}$, the symmetry fractionalization class and defectification class~\cite{Barkeshli_2019}.

We write the element in $S_3$ as $\Tilde{g}=(G,g)\equiv x^Ga^g$ with some slight abuse of the notation, where $G=0,1$ and $g=0,1,2$. It should be clear from the context when $G$ is a number or a group. A representative of the cocycle in $H^3(S_3,U(1))$ is 
\beq
    &\omega(\Tilde{g},\Tilde{h},\Tilde{l})  \\
    &=\exp{\frac{2\pi i p_1}{9}g(-1)^{H+L}(h(-1)^L+l-[h(-1)^L+l]_3)}\\
    &\qquad \times \exp{\pi i p_2 GHL},
    \label{eq:s_3 cocycle}
\eeq
where $p_1=0,1,2$, and $p_2=0,1$. As pointed out in Sec.~{\ref{sec:GaugingMap}}, the anyon set (UMTC) $\mathcal{C}$ is determined by the restriction of $\omega$ on subgroup $Z_3$ (i.e., setting $G=H=L=0$),
\beq
    \nu(g,h,l)=\exp{\frac{2\pi i p_1}{9}g(h+l-[h+l]_3)}.
\eeq
Different values of $p_1$ are in one-to-one correspondence with  different $Z_3$ twisted quantum double phases $D^{\nu}(Z_3)$. An anyon in these phases is characterized by its flux $a\in Z_3$, and a projective representation of $Z_3$, satisfying
\beq
    \mu_a(g)\mu_{a}(h)=\exp{\frac{2\pi i p_1 a}{9}(g+h-[g+h]_3)}\mu_a(gh),
\eeq
which means $\mu_a(g)=e^{\frac{2\pi i a p_1 g}{9}}v(g)$, where $v(g)$ is an ordinary representation of $Z_3$. 

Using Lyndon–Hochschild–Serre spectral sequence~\cite{lyndon1948cohomology, Hochschild1953,chen2012symmetry}, we can decompose the cohomology class of the $S_3$ group as
\beq
 \label{eq:S3-decomposition-LHS}
H^3(S_3,U(1))=H^3(Z_3,U(1))\oplus H^3(Z_2,U(1)).
\eeq
This suggests that this SET state is composed of a TQD $D^{\nu}(Z_3)$ and a $Z_2$-SPT state.
A natural question is whether the wave function of the whole system is decomposed into a product of the two corresponding parts.

It turns out that we can write the 3-cocycle in Eq.~\eqref{eq:s_3 cocycle} as
\beq
    \omega(\tilde{g},\tilde{h},\tilde{l})=\omega_n\cdot\omega_q\cdot\omega'(\tilde{g},\tilde{h},\tilde{l}),
    \label{eq:s_3 decomposition}
\eeq
\begin{figure}[h]
    \centering
    \includegraphics[width=\linewidth]{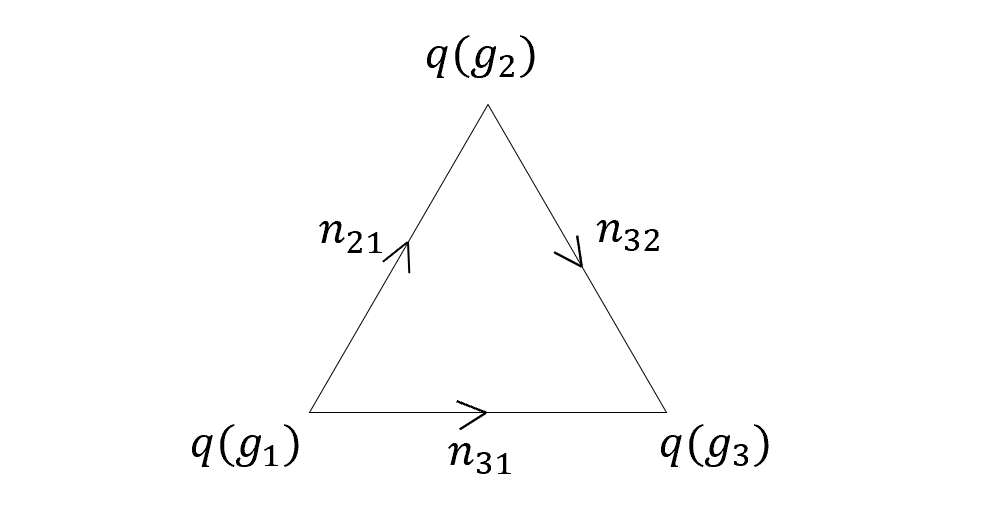}
    \caption{One plaquette on a triangulated lattice.}
    \label{fig:triangle}
\end{figure}
where the phase $\omega_n$ is defined from  the 3-cocycle of $Z_3$,
\beq
    &\omega_{n}(g_3 g_2^{-1},g_2 g_1^{-1},g_1)
    \\
    &\equiv\nu\big(n(g_3) n(g_2)^{-1},n(g_2) n(g_1)^{-1},n(g_1)\big)
    \\
    &=\nu\big(\sigma^{q(g_2)^{-1}}(n(g_3 g_2^{-1})),\sigma^{q(g_1)^{-1}}(n(g_2 g_1^{-1})),n(g_1)\big),
\eeq
and $\nu$ is a representative in $H^3(N,U(1))$. Every plaquette on the spatial manifold is associated with such a phase. The product of them over all plaquettes gives
\beq
    \Omega_{\nu}(\{n(g_v)n(g_{v'})^{-1}\})=\prod_{\Delta}\nu(\{n(g_v)\})^{s(\Delta)}.
\eeq
The sum of the above phase over all possible $\{g_v\}$ configurations gives the wave function of a ground state of $D^{\nu}(Z_3)$. 

The phase $\omega_q$ is defined from the 3-cocycle of $Z_2$, $\alpha\in H^3(Z_2,U(1))$,
\beq
    \omega_q(\tilde{g},\tilde{h},\tilde{l})=\exp{\pi i p_2 GHL}=:\alpha(G,H,L).
\eeq
The product of this type of phases gives 
\beq
    \Omega_q(\{q(g_v)q(g_{v'})^{-1}\})=\prod_{\Delta}\alpha(\{q(g_v)\})^{s(\Delta)}.
\eeq
The sum of the above phases over all possible $\{g_v\}$ configurations results in the wave function of a $Z_2$-SPT state. 
The $\omega'$ part in Eq.~\eqref{eq:s_3 decomposition} is
\beq
    \omega'(\tilde{g},\tilde{h},\tilde{l})=\exp{\frac{2\pi i p_1}{3}(-1)^{H}g(1-\delta_{h,0})\delta_{L,1}},
\eeq
which is nontrivial when $p_1\neq 0$. Similarly, we define the product of this type of phases over the spatial manifold as $\Omega'$,
\beq
    \Omega'(\{q(g_v)\},\{n(g_v)\})=\prod_{\Delta}\omega'(\{q(g_v)\},\{n(q_v)\})^{s(\Delta)}.
\eeq

Thus the resulting state after gauging $Z_3$ from an $S_3$-SPT state is 
\beq
    \ket{\Psi_{\text{SET}}}=&\sum_{\{g_v\}}\Omega'(\{q(g_v)\},\{n(g_v)n(g_{v'})^{-1}\})\\
    &\Big(\Omega_n\ket{\{n(g_v)n(g_{v'})^{-1}\}}_e\Big)\bigotimes\Big(\Omega_q\ket{\{q(g_v)\}}_v\Big),
    \label{eq:S_3 entangle}
\eeq
which is an entangled state between a $Z_3$-TQD ground state and a $Z_2$-SPT state. When $p_1=0$, we have $\omega'= \Omega_n = 1$, hence the wave function of the system becomes
\beq
    \ket{\Psi}=&\sum_{\{g_v\}}\Omega_q\ket{\{n(g_v)n(g_{v'})^{-1}\}}_e\bigotimes\ket{\{q(g_v)\}}_v
    \\
    =&\sum_{\{g_v\}}\Big(\ket{\{n(g_v)n(g_{v'})^{-1}\}}_e\Big)\bigotimes\Big(\Omega_q\ket{\{q(g_v)\}}_v\Big)
    \\
    =&\ket{Z_3~\text{TC}}\otimes\ket{Z_2~\text{SPT}},
\eeq
which is a product state of a $Z_3$ Toric code ground state and a $Z_2$-SPT state.
 
Having obtained the SET wave functions, we now discuss the effect of the global symmetry action. The $Z_2$ symmetry action $U^x=\prod_v L^x_{-v}$ in the $S_3$-SPT state is mapped to $U^x_Q=U^x \sigma^x$, under which
an anyon with flux $a^i$ will be mapped to one with flux $\sigma^x(a)=a^{[-i]_3}$. And a chargeon will be mapped to its antiparticle under the symmetry. According to Sec.~\ref{sec:SET classification}, the possible SFC will be given by elements in a $H^2_{\rho}(Z_2,\mathcal{A})$ torsor. With different values of $p_1$, the abelian group $\mathcal{A}$ could be either $Z_3\times Z_3$ or $Z_9$. In
either cases, the cohomology group $H^2_{\rho}(Z_2,\mathcal{A})$ turns out to be trivial, and so is its torsor. Therefore, the TQD $D^{\nu}(Z_3)$ has only one possible $Z_2$ symmetry fractionalization pattern. Moreover, different values of $p_2$ in the 3-cocycle of $S_3$ result in different $Z_2$ symmetry defectification classes (SDC) in the SETs, which are obtained by gauging the normal $Z_3$ group. This is  expected because as seen from  Eq.~\eqref{eq:S_3 entangle},  different $p_2$ values represent  different $Z_2$-SPT states entangled with some $Z_3$-TQD state.

Let us remark that this construction for $S_3$ has a natural generalization on $D_{2n+1}$ groups, where one first gauge $Z_{2n+1}$, resulting in an $Z_{2n+1}$ SET on which the $Z_2$ symmetry acts to conjugate the gauge DOFs. Different parent SPT phases will result in different anyon theories and different SDCs, but always some unique symmetry fractionalization pattern. One can further gauge the quotient $Z_2$ symmetry to obtain $D_{2n+1}$ TQD.

\subsection{Gauging $D_{2n}$ SPT}

Now we discuss the process of 2-step gauging a generic $D_{2n}$ SPT state via measurement under a similar type of abelianized basis. As in $S_3$, an element in group $D_{2n}$ can be written as
\begin{equation}
    g=x^i a^j,~ i=0,1,~ j=0,1,\cdots,2n-1.
\end{equation}
We will thus use a generalized definition of operators as for $S_3$ in Eqs.~(\ref{eq:n}), (\ref{eq:q}), (\ref{eq:x}), and~(\ref{eq:z}). 

After applying the control gate to set the DOFs on edges, e.g., $\langle ij \rangle$, to $n(g_i)n(g_j)^{-1}$,  measuring $X_{(n)}$ on vertices, and correcting all chargeon excitations, the resultant state is a ground state in a $Z_{2n}$-SET phase with a global $Z_2$ symmetry at this intermediate stage. According to the multiplication rule of $D_{2n}$, just as in $S_3$, the symmetry transformation $U^x$ conjugates all the gauge DOFs.

As an illustration, we will consider the $D_4$ group. Using the Lyndon-Hochschild-Serre sequence, the cohomology group can be decomposed as 
\beq
    H^3(D_4,U(1))=&H^3(Z_4,U(1))\oplus H^3(Z_2,U(1))
    \\
    &\oplus H^2(Z_2,H^1(Z_4,U(1))).
\eeq
Again, the cocycle factorizes as
\beq
    \omega= \omega_{n}\cdot\omega_{q}\cdot\omega'\cdot\omega_{nq},
\eeq
where $\omega_{n}$ and $\omega_{q}$ are defined similarly as for the $S_3$ group in the last section, while $\omega'$ and $\omega_{nq}$ depend on both quotient and normal parts of the vertex DOFs. From this decomposition, it is clear that after gauging the normal $Z_{2n}$, we have an entangled state between the $Z_{2n}$-TQD state and the $Z_2$-SPT state. Because of the additional entanglement via $\omega_{nq}$, we expect to have a nontrivial SFC from gauging the $Z_4$ symmetry of a $D_4$-SPT state.

Indeed, in the next section, we will show that the current 2-step gauging setup could result in different SFCs.  
To do so, we first introduce symmetry branch line operators and other necessary tools to determine the SFC.
We will also discuss several examples, including the above $Z_4$-SET phase.

\section{Symmetry Defect in SPT and SET}
\label{sec:SymmetryDefect}
In this section, we will apply the notion of symmetry branch lines introduced in Ref.~\cite{Barkeshli_2019} and formulate corresponding symmetry branch line operators  in SPT phases, as well as their relation to ribbon operators in TQD. 
We then show the gauging procedure transforms such operators in SPT phases into symmetry branch lines in SET phases and discuss how their fusions relate to the symmetry fractionalization classes (SFC) in a few examples.

\subsection{Symmetry Branch Lines}\label{sec:sbl}

To introduce symmetry branch lines, we start with the symmetry action in an SPT wave function. We recall in Eq.~(\ref{eq:PsiSPT})  the SPT wave function on a triangulated spatial manifold, 
\beq
\ket{\Psi_{\text{SPT}}}=\sum_{\{g_v\}}\prod_{\Delta}\omega(g_3g_2^{-1},g_2g_1^{-1},g_1)^{s(\Delta)}\bigotimes_{v} \ket{g_v}.
\eeq
When the manifold is closed, the global symmetry action $U^x\equiv\prod_v L^x_{-v}$ leaves the entire SPT state invariant. We can also consider the symmetry action on a sub-manifold $\mathcal{R}$~\cite{else2014classifying}, such as the one shown in Fig.~\ref{fig:closed_symmetry},
\beq
    U^x_{\mathcal{R}} \prod_{\Delta} \omega^{s(\Delta)}\ket{\{g_v\}}=\prod_{\Delta}Amp^{\mathcal{R}}(\{g_v\},x)\omega^{s(\Delta)}\ket{\{g_v\}}.
\eeq
 Triangulating the frustum created by lifting vertices in $\mathcal{R}$, we have multiple tetrahedrons. We associate each tetrahedron in the frustum as in Fig.~\ref{fig:closed_symmetry} to a 3-cocycle, such as
 the one in Fig~\ref{fig:tetrahedron}, to which we assign a phase factor $\omega(g_4g_3^{-1},g_3g_2^{-1},g_2g_1^{-1})$. The product of all such cocycles
composes the factor $Amp^{\mathcal{R}}$, namely,
\beq  
    Amp^{\mathcal{R}}=\prod_{\text{tetra}\in\mathcal{R}} \omega(\text{tetra})^s.
\eeq

Using the cocycle conditions, it turns out that $\text{Amp}^{\mathcal{R}}$ only depends the DOFs around $\partial \mathcal{R}$ and does not depend on those DOFs deep inside $\mathcal{R}$, see Fig.~\ref{fig:closed_defect}, and its expression is
\beq
    Amp^{\mathcal{R}}=\Tilde{\Theta}^{g_v x g_v^{-1}}_{\partial\mathcal{R}}\prod_{\text{tetra}\in\partial\mathcal{R}} \omega(\text{tetra})^s,
\eeq
where the extra factor to the product on the r.h.s. is
\beq
    \Tilde{\Theta}^{g_n x g_n^{-1}}_{\partial\mathcal{R}}=\theta_{g_n x g_n^{-1}}(g_n g_{n-1}^{-1},g_{n-1}g_1^{-1})\cdots\theta_{g_3 x g_3^{-1}}(g_3 g_{2}^{-1},g_{2}g_1^{-1}),
\eeq
when $\partial\mathcal{R}$ contains vertices equipped with the branching structure, $ 1\rightarrow 2\rightarrow 3\rightarrow \cdots \rightarrow n \leftarrow 1$, see Appendix~\ref{sec: branch line operator appendix}. In fact, one can introduce an operator $\mathcal{B}^x_{\partial \mathcal{R}}$ that is supported only on $\partial \mathcal{R}$ such that~\cite{else2014classifying} 
\beq
    ``\mathcal{B}^x_{\partial\mathcal{R}}"\ket{\Psi_{\text{SPT}}}=U^x_{\mathcal{R}} \ket{\Psi_{\text{SPT}}}.
    \label{eq:defect-closed}
\eeq

\begin{figure}[h]
    \centering
    \includegraphics[width=0.8\linewidth]{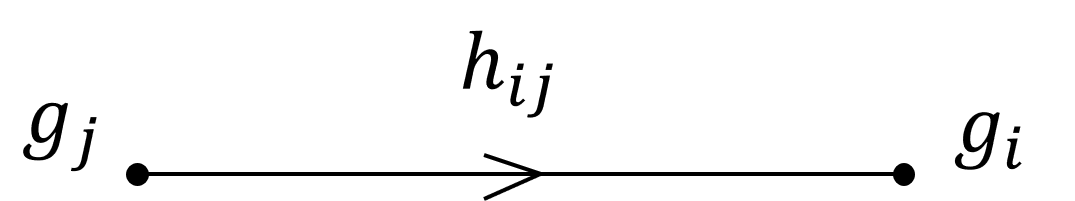}
    \caption{$\ket{h_{ij}}_e$ is located on edge $e=\langle i,j\rangle$, $\ket{g_i}_v$ and $\ket{g_j}_v$ are located on the two endpoints $i$ and $j$.}
    \label{fig:pregauge}
\end{figure}

To do this carefully, we need to first introduce a pre-gauge structure~\cite{dolev2022gauging}. Namely, we introduce a $G$-DOF $h_{ij}$ on every edge $\langle i,j\rangle$ (which we will be set to the identity group element $\ket{1}_e$ to begin with). One can think of the edge DOFs as the discrete gauge field. After introducing such a gauge field, one may write the local symmetry action  on both vertex $v$ and the surrounding edges $e\supset v$. This is also called the ``gauge transformation" operator on a vertex $v$ as
\beq \label{eq:gauge-transformation}
G^x_v\equiv L^x_{-v} \prod_{e\supset v}L^{x}_{\pm e},
\eeq
where $e\supset v$ denotes those edges with one end being $v$ and where $L^{\pm,x_e}_e$ is the left (right) actions of $x_e$ on $\ket{h}_e$, when the edge $e$ flows to (emanates from) $v$. Furthermore, the ``interactions'' should also be written in a gauge invariant way. Namely, instead of the original SPT Hamiltonian, one can write the gauge invariant version as follows,
\beq
\label{eq:hamiltonian SPT-pre}
    H_{\text{SPT-pre}}=-\sum_v \sum_{g\in G} \frac{1}{|G|}L^g_{+v} W^g_v,
\eeq
where the $W$ phase was previously introduced in Sec.~\ref{sec:GaugingMap}.

Taking any ground state of this Hamiltonian, we can impose the gauging map in the presence of the pre-gauge structure, 
\beq
    \Gamma:\ \  \ket{\{g_i\}}_v\ket{h_{ij}}_e\rightarrow \ket{\{g_i h_{ij}g_j^{-1}\}}_e,
\eeq 
under which, the state would be mapped to a ground state of the TQD model. One special ground state of this Hamiltonian would be 
\beq
    \ket{\Psi_{\text{SPT-pre}}}=\ket{\Psi_{\text{SPT}}}\bigotimes_e\ket{1}_e,
\eeq
where we have taken the original SPT state, which has $\ket{h}_e\equiv \ket{1}_e$ on all edges. One can easily verify that $G^x_v\ket{\Psi_{\text{SPT-pre}}}$ is still a ground state, for any vertex $v$ on lattice and $x\in G$. We define the gauge transformation over a region $\mathcal{R}$ as $G^x_{\mathcal{R}}\equiv\prod_{v\in \mathcal{R}}G^x_v$. If the spatial lattice $\Gamma$ is closed, we write $G^x\equiv \prod_{v\in \Gamma}G^x_v$ and one can check that
\beq
    G^x\ket{\Psi_{\text{SPT-pre}}}=\ket{\Psi_{\text{SPT-pre}}}.
\eeq
Therefore, the operator $G^x$ mimics the behavior of global symmetry operator $U^x$ after introducing the pre-gauge structure.

Now we introduce the definition  of symmetry branch line operators $\tilde{\mathcal{B}}^x_{\partial\mathcal{R}}$ on states with trivial edges $\ket{h}_e\equiv\ket{1}_e$,
\beq
    \label{eq:branch line operator}
    \tilde{\mathcal{B}}^x_{\partial\mathcal{R}}\ket{\Psi_{\text{SPT-pre}}}=G^x_{\mathcal{R}}\ket{\Psi_{\text{SPT-pre}}},
\eeq
but unlike $G^x_{\mathcal{R}}$, the operator $\tilde{\mathcal{B}}^x_{\partial\mathcal{R}}$ only takes effect on $\partial\mathcal{R}$. We find that the following expression of $\tilde{\mathcal{B}}^x_{\partial\mathcal{R}}$,
\begin{align}
    \label{eq:tilde B}
    \tilde{\mathcal{B}}^x_{\partial\mathcal{R}}=\sum_{g_v}\mathcal{L}^{x}_{\partial \mathcal{R}} \Tilde{W}^{g_v x g_v^{-1}}_{\partial\mathcal{R}}\Tilde{\Theta}^{g_v x g_v^{-1}}_{\partial\mathcal{R}}\ket{g_v}_v\bra{g_v}, 
\end{align}
where $v$ is a reference vertex on $\partial\mathcal{R}$ as shown in Fig.~\ref{fig:closed_defect}. The phase $\Tilde{W}^{g_v x g_v^{-1}}_{\partial\mathcal{R}}$ is the product of cocycles associated to the tetrahedrons in Fig.~\ref{fig:closed_defect}, with $g_{v'}g_v^{-1}=g_v x g_v^{-1}$,
\beq
\label{eq:W}
\Tilde{W}^{g_v x g_v^{-1}}_{\partial\mathcal{R}}=\prod_{\text{tetra}\in \partial\mathcal{R}}\omega(\text{tetra})^s.
\eeq
The operator $\mathcal{L}^{x}_{\partial\mathcal{R}}$ is a product of shift operators on the edges crossed by $\partial\mathcal{R}$. In general, $\mathcal{L}^x_{l}$ with $x\in G$ on a ribbon $l$ is defined as follows,

\begin{widetext}
\begin{equation}
\label{eq:Lx1}
    \begin{tikzcd}
	& {g_1} && {g_2} && {g_3} && {g_4} \\
	{\mathcal{L}_l^x} &&&&&&& \textcolor{white}{ } \\
	& {g_5} && {g_6} && {g_7} \\
	& {g_1} && {g_2} && {g_3} && {g_4} \\
	{=} \\
	& {g_5} && {g_6} && {g_7}
	\arrow["{1}"', from=3-2, to=1-2]
	\arrow["{1}"', from=3-4, to=1-4]
	\arrow["{1}"', from=3-6, to=1-6]
	\arrow["{1}"', from=1-4, to=1-2]
	\arrow["{1}"', from=1-6, to=1-4]
	\arrow["{1}"', from=1-8, to=1-6]
	\arrow["{1}"', from=4-8, to=4-6]
	\arrow["{1}"', from=4-6, to=4-4]
	\arrow["{1}"', from=4-4, to=4-2]
	\arrow["{x}"', from=6-2, to=4-2]
	\arrow["{x}"', from=6-4, to=4-4]
	\arrow["{x}"', from=6-6, to=4-6]
	\arrow["l", color={rgb,255:red,214;green,92;blue,92}, dashed, from=2-1, to=2-8]
    \end{tikzcd}
\end{equation}
\end{widetext}

This operation 
$\tilde{\mathcal{B}}^x_{\partial\mathcal{R}}$ 
can be regarded as the non-onsite symmetry action on the boundary $\partial\mathcal{R}$, and its definition can be extended to the case whenever there is no flux in the state (i.e., for every plaquette, $\prod_{e} h_e=1$.). From direct calculation using Eq.\eqref{eq:branch line operator} and the fact that $G^x_{\mathcal{R}}G^y_{\mathcal{R}}=G^{xy}_{\mathcal{R}}$, one can easily show that the multiplication rule of $\tilde{\mathcal{B}}^x_{\partial\mathcal{R}}$,
\beq
    \tilde{\mathcal{B}}^x_{\partial\mathcal{R}}\tilde{\mathcal{B}}^y_{\partial\mathcal{R}}=\tilde{\mathcal{B}}^{xy}_{\partial\mathcal{R}},
\eeq
is exactly the multiplication rule of the group $G$, as expected for the symmetry branch lines. One can also show the multiplication rule of $\tilde{\mathcal{B}}^x_{\partial\mathcal{R}}$ directly using the form in Eq.~\eqref{eq:tilde B}, see Appendix~\ref{sec: branch line operator appendix}.  Indeed, if we use the operator 
\beq
    \prod_v \big(\frac{1}{|G|}\sum_{x\in G} G^x_v\big)
    \label{eq:projector}
\eeq
to project any ground state onto the gauge invariant sector, we would also have a TQD ground state. The operator above in Eq.~\eqref{eq:projector} can be seen as a superposition of different meshes of symmetry branch lines.

\begin{figure}[h]
    \centering
    \begin{subfigure}[h]{0.9\linewidth}
    \centering
    \includegraphics[width=\linewidth]{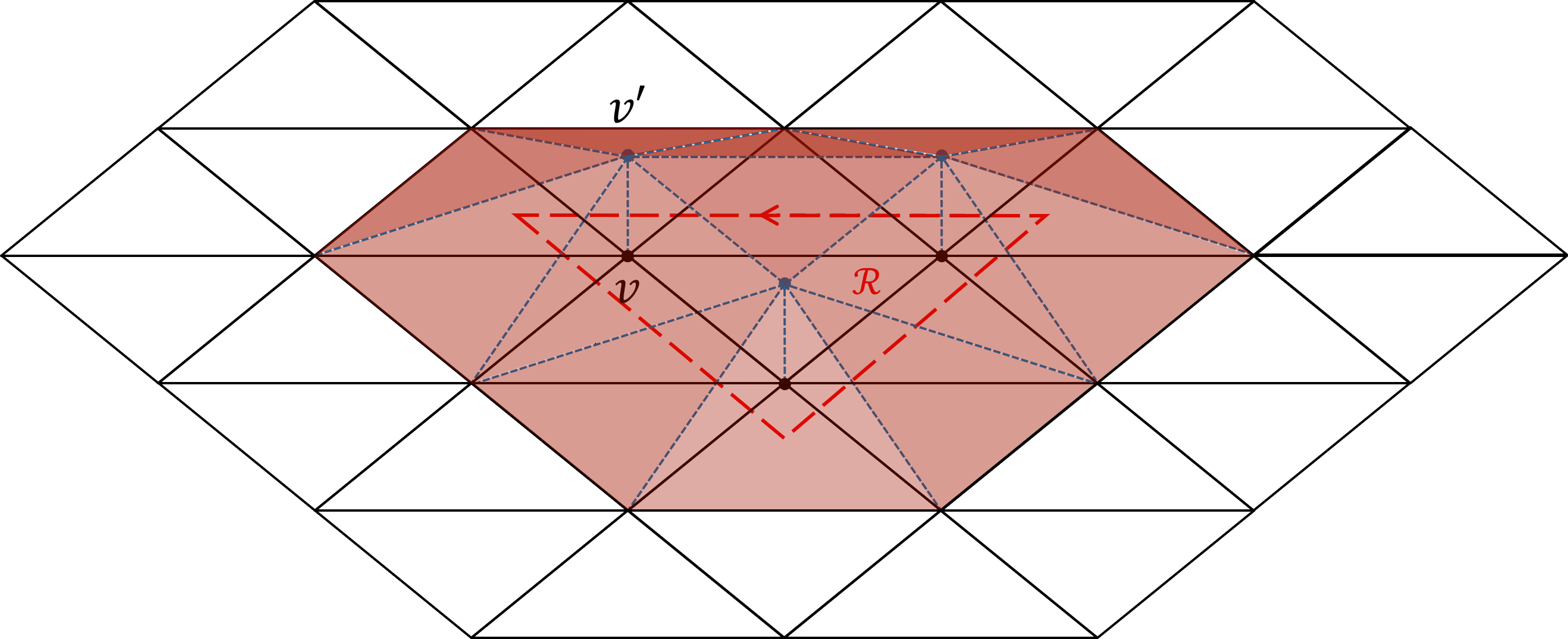}
    \caption{}
    \label{fig:closed_symmetry}
    \end{subfigure}
    \hfill
    \begin{subfigure}[h]{0.9\linewidth}
    \centering
    \includegraphics[width=\linewidth]{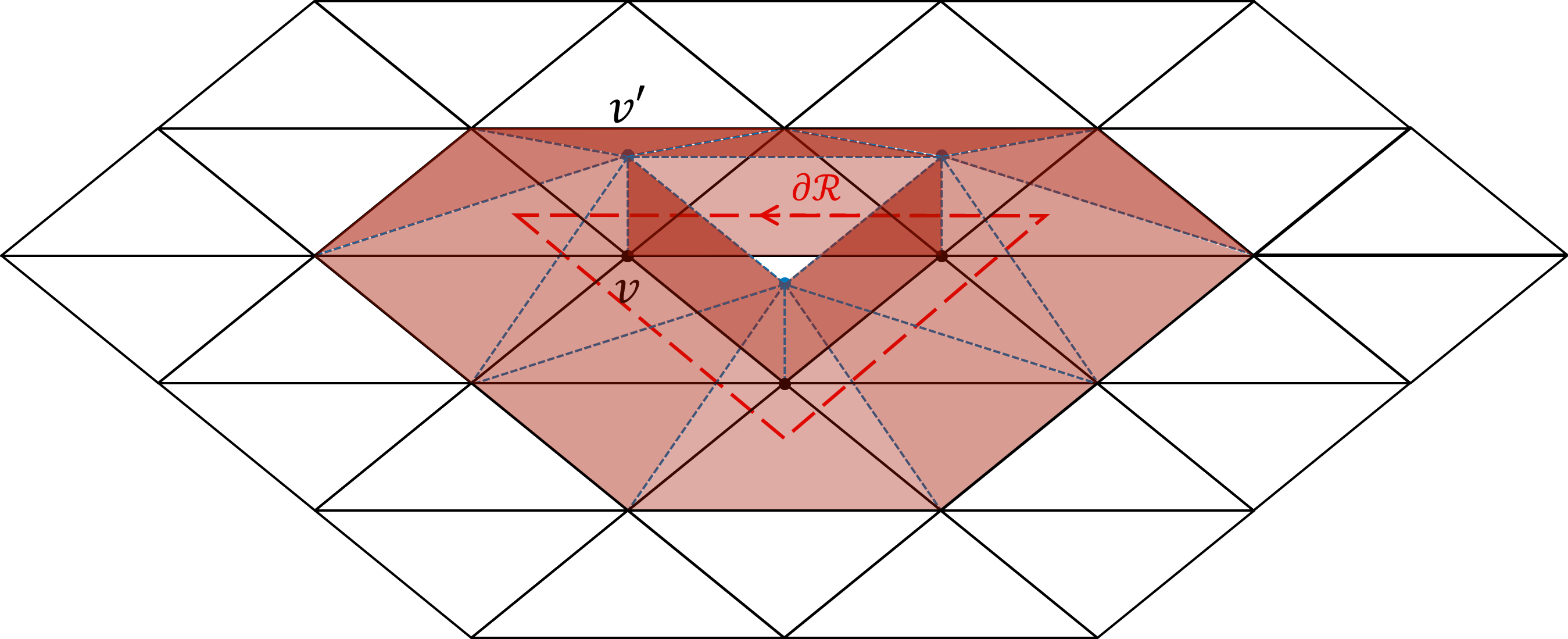}
    \caption{}
    \label{fig:closed_defect}
    \end{subfigure}
    \caption{(a) Symmetry action inside of region $\mathcal{R}$ ``lifts'' $\mathcal{R}$ such that all the simplex in the region correspond to $\Tilde{\omega}$. (b) This symmetry action can 
 be equivalently  regarded as the insertion of symmetry branch line on $\partial\mathcal{R}$. }
\end{figure}

There are two  important remarks here.  The first is that inserting a symmetry branch line on an SPT state is to create a point symmetry defect $0_x$, move it along $\partial\mathcal{R}$ and annihilate it with $0_{x^{-1}}$. The multiplication rule above indicates that the fusion between point defects is $0_x \cdot 0_y = 0_{xy}$.
Second, we can consider the global symmetry action $G^y$ on the state with a branch line on $\partial\mathcal{R}$, i.e., $\tilde{\mathcal{B}}^x_{\partial\mathcal{R}}\ket{\Psi_{\text{SPT-pre}}}$. It turns out that 
\beq
    G^y \tilde{\mathcal{B}}^x_{\partial\mathcal{R}}\ket{\Psi_{\text{SPT-pre}}}=&G^yG^x_{\mathcal{R}}\ket{\Psi_{\text{SPT-pre}}}\\
    =&G^yG^x_{\mathcal{R}}G^{y^{-1}}\ket{\Psi_{\text{SPT-pre}}}\\
    =&G^{yxy^{-1}}_{\mathcal{R}}\ket{\Psi_{\text{SPT-pre}}}\\
    =&\tilde{\mathcal{B}}_{\partial\mathcal{R}}^{yxy^{-1}}\ket{\Psi_{\text{SPT-pre}}},
\eeq
where $G^y_{\mathcal{R}}\equiv\prod_{v\in \mathcal{R}}G^y_v$. In other words, the global symmetry transformation $\rho^y$ 
 on $0_x$ is $\rho^y(0_x)=0_{yxy^{-1}}$.

\subsection{More on Symmetry Branch Lines}
\label{sec:more on sbl}

To discuss further the symmetry branch lines, we first remind the readers of some definitions in group cohomology. Given a 3-cocycle $\omega(g,h,l)$ as a representative of element in $H^3(G,U(1))$, the slant product is defined as \cite{propitius1995topological}
\beq\label{eq:thetadef}
    \theta_x(g,h)\equiv\frac{\omega(x,g,h)\omega(g,h,(gh)^{-1}xgh)}{\omega(g,g^{-1}xg,h)},
\eeq
which is naturally a conjugated 2-cocycle, i.e., a representative in $H^2(G,U(1)[G])$. Namely, it satisfies the following condition,
\beq
    \tilde{\delta}\theta_x(g,h,l)\equiv\frac{\theta_{g^{-1}xg}(h,l)\theta_x(g,hl)}{\theta_x(gh,l)\theta_x(g,h)}=1.
\eeq
When $\theta$ is a representative of the trivial element in $H^2(G,U(1)[G])$, there exists a conjugated 1-cochain $\epsilon$, such that \beq \label{eq:conj1cochain}
\theta_x(g,h)=\tilde{\delta}\epsilon_x(g,h)\equiv\frac{\epsilon_{g^{-1}xg}(h)\epsilon_x(g)}{\epsilon_x(gh)}.
\eeq
There is another product that will also become useful later:
\beq
    \gamma_g(x,y)\equiv\frac{\omega(x,y,g)\omega(g,g^{-1}xg,g^{-1}yg)}{\omega(x,g,g^{-1}yg)},
\eeq
which, however, is not a 2-cocycle nor a conjugated one.

For a general state with a nontrivial pre-gauge structure, we give a conjecture for the expression of symmetry branch lines. 
We can follow the similar idea as in the previous section to introduce the branch line operator $\mathcal{B}^x_{\partial\mathcal{R}}$ from the symmetry action in the region $\mathcal{R}$, when all the plaquettes on $\partial\mathcal{R}$ are fluxless (i.e. $\prod_{e\in\partial p}h_e=1$), and the flux on each plaquette $p\in\mathcal{R}$ ($\prod_{e\in \partial p}h_e$) is in the centralizer group $\mathcal{Z}_x$, see Appendix~\ref{sec: branch line operator appendix} for details. If we start from vertex $v$ and go along $\partial\mathcal{R}$, the holonomy, i.e., the product of group elements on all the edges along the path, is $g_v \prod_e h_e g_v^{-1}$, and the resulting symmetry branch line is 
\beq
    \mathcal{B}^{x}_{\partial\mathcal{R}}&=\sum_{g}\mathcal{B}^{x,g}_{\partial\mathcal{R}}\epsilon_{g_v x g_v^{-1}}(g),
\eeq
with  the operator on the r.h.s. being
\beq
\mathcal{B}^{x,g}_{\partial\mathcal{R}}&\equiv\sum_{g_v}\mathcal{L}^{x}_{\partial\mathcal{R}} W^{g_v x g_v^{-1}}_{\partial\mathcal{R}}\Theta^{g_v x g_v^{-1}}_{\partial\mathcal{R}}\delta_{g,g_v (\prod_e h_e )g_v^{-1}}\ket{g_v}_v\bra{g_v},
\eeq
where $\epsilon_x(g)$ is a 1-cochain defined in Eq.~(\ref{eq:conj1cochain}), and $v$ is a reference vertex. The phase $W^{g_v x g_v^{-1}}_{\partial\mathcal{R}}$ is the product of cocycles associated to the tetrahedrons in Fig.~\ref{fig:closed_defect}, with $h_{v'v}=x$ and $g_{v'}=g_v$, namely,
\beq
W^{g_v x g_v^{-1}}_{\partial\mathcal{R}}=\prod_{\text{tetra}\in \partial\mathcal{R}}\omega(\text{tetra})^s.
\eeq
The operator $\Theta^{g_v x g_v^{-1}}_{\partial\mathcal{R}}$ is defined in Eq.~\eqref{eq:Theta},
and $\mathcal{L}^x_{\partial\mathcal{R}}$ is defined as follows.

\begin{widetext}
\begin{equation}
\label{eq:Lx}
    \begin{tikzcd}
	& {g_1} && {g_2} && {g_3} && {g_4} \\
	{\mathcal{L}_l^x} &&&&&&& \textcolor{white}{ } \\
	& {g_5} && {g_6} && {g_7} \\
	& {g_1} && {g_2} && {g_3} && {g_4} \\
	{=} \\
	& {g_5} && {g_6} && {g_7}
	\arrow["{k_1}"', from=3-2, to=1-2]
	\arrow["{k_2}"', from=3-4, to=1-4]
	\arrow["{k_3}"', from=3-6, to=1-6]
	\arrow["{h_1}"', from=1-4, to=1-2]
	\arrow["{h_2}"', from=1-6, to=1-4]
	\arrow["{h_3}"', from=1-8, to=1-6]
	\arrow["{h_3}"', from=4-8, to=4-6]
	\arrow["{h_2}"', from=4-6, to=4-4]
	\arrow["{h_1}"', from=4-4, to=4-2]
	\arrow["{x\, k_1}"', from=6-2, to=4-2]
	\arrow["{h_1^{-1}xh_1 \, k_2}"', from=6-4, to=4-4]
	\arrow["{(h_1h_2)^{-1}x(h_1h_2)\, k_3}"', from=6-6, to=4-6]
	\arrow["l", color={rgb,255:red,214;green,92;blue,92}, dashed, from=2-1, to=2-8]
\end{tikzcd}
\end{equation}
\end{widetext}

In order to make $\mathcal{L}^x_{\partial \mathcal{R}}$ well-defined, we require that $x$ and $\prod_e h_e$ commute. For an SPT state, we have that $\prod_e h_e=1$, and this is trivially satisfied. 
With the above discussions, we
 can
 derive the multiplication of branch line operators (see Appendix~\ref{sec: branch line operator appendix} for the proof),
\begin{align}
    \mathcal{B}^{x,g}_{\partial\mathcal{R}}\,\mathcal{B}^{y,g'}_{\partial\mathcal{R}}=\mathcal{B}^{xy,g}_{\partial\mathcal{R}}\,\gamma_g(g_v x g_v^{-1},g_v y g_v^{-1})\,\delta_{g,g'}.
    \label{eq:fusion}
\end{align}

We also note that given a 3-cocycle $\omega$, the phase factor $\epsilon$ is not unique. One can always replace $\epsilon_x$ by $\epsilon_x v_x$, where $v_x$ satisfies $\tilde{\delta}v_x(g,h)\equiv1$. This corresponds to a different choice of $0_x$ in $\mathcal{C}_x$, and we will illustrate this with concrete examples below.

We recall the gauging map in the presence of the pre-gauge structure, 
\beq
    \Gamma:\ \  \ket{\{g_i\}}_v\ket{h_{ij}}_e\rightarrow \ket{\{g_i h_{ij}g_j^{-1}\}}_e.
    \label{eq:generalized gauging map}
\eeq
Under this, the operator $\mathcal{B}^x_{\partial\mathcal{R}}$ is mapped to 
\beq
\Gamma(\mathcal{B}^x_{\partial\mathcal{R}})&=\frac{1}{|G|}\sum_{k\in G} \mathcal{L}^{kxk^{-1}}_{\partial \mathcal{R}} W^{k x k^{-1}}_{\partial\mathcal{R}}\Theta^{k xk^{-1}}_{\partial\mathcal{R}}\epsilon_{k x k^{-1}}(kgk^{-1}).
\eeq
Notice that after gauging, the $g_v$ dependence of the operator is summed over as in the sum of $k\in G$ above. Therefore, the resultant operator under the gauging map  does not depend on any reference vertex.

In the case when $\omega\equiv1$, we can choose $\epsilon\equiv1$, then the operator $\Gamma(\mathcal{B}^x_{\partial\mathcal{R}})$ is reduced to  the trace of the ribbon operator that creates an $x$-fluxed anyon in the quantum double (see, e.g., Ref.~\cite{BombinDelgado2008}),
\beq
    F^{C_x,\bold{1}}_{\partial\mathcal{R}}=\frac{1}{|\mathcal{Z}_x|}\sum_{a\in \mathcal{Z}_x}F^{c_i,b_i a b_j^{-1}}_{\partial\mathcal{R}},
    \label{eq:ribbon qd}
\eeq
where on the left-hand side,  $C_x$ is the conjugacy class of $x$, $\bold{1}$ is the trivial representation of the centralizer group $\mathcal{Z}_x$. In order to construct the operators $F^{C_x,\bold{1}}_{\partial\mathcal{R}}$, one enumerates the elements of the conjugacy class as $C_x=\{c_i\}$, together with a suitable subset $\{b_i\}_{i=1}^{|C_x|}\subset G$ such that $c_i=b_i x
b_i^{-1}$. The operator $F^{C_x,\bold{1}}_{\partial\mathcal{R}}$ on the left-hand side is a ribbon operator labeled by topological charges, and $F^{c_i,b_i n b_j^{-1}}_{\partial\mathcal{R}}$ on the right-hand side is a ribbon operator in a basis labeled by group elements.  

For a generic 3-cocycle $\omega$, when $x$ is in the center of $G$, the $x$-fluxed anyon is abelian. Then we have the following  relation via the map $\Gamma$,
\beq
    \Gamma(\mathcal{B}^{x,g}_{\partial\mathcal{R}})=F^{x,g}_{\partial\mathcal{R}},
\eeq
where  $F^{x,g}_{\partial\mathcal{R}}$ is a ribbon operator defined in TQD for abelian groups in Ref.~\cite{MesarosRan2013}. 

The algebra of ribbon operators can be inferred from the quasi-Hopf algebra
\cite{propitius1995topological, dijkgraaf1991quasi}. To be more concrete, suppose we insert an $x$-flux on ribbon $l$. Operator $F^{x,g}_{l}$ thus 
satisfies the multiplication rule
\beq
    F^{x,g}_{l}F^{y,g'}_{l}=F^{xy,g}_{l}\delta_{g,g'}\gamma_g(x,y),
    \label{eq:quasi-hopf}
\eeq
which is consistent with our result of  the branch line multiplication in Eq.~(\ref{eq:fusion}), since we expect the gauging map to preserve the operator algebra, which is quasi-Hopf in this case.

One thing to notice is that, in order to write down $\mathcal{B}^x_{\partial\mathcal{R}}$, we have assumed the existence of $\epsilon$. This is not always possible. When such $\epsilon$ does not exist, even when $x$ is in the center of $G$, the $x$-fluxed anyon can still be nonabelian~\cite{propitius1995topological}. For example, if $\omega$ is a type-3 cocycle of $Z_2\times Z_2\times Z_2$, the anyons in the TQD are generally nonabelian. Therefore, one could not expect to write down ribbon operators as we have defined above. However, when we gauge one of the $Z_2$ groups from the SPT, we would enter an SET order with global symmetry $Z_2\times Z_2$. It turns out that we can try to write the branch line operators in the SET order, and from the algebra of which, one can infer the symmetry fractionalization patterns. We will leave this for further discussion later in this paper. 

Throughout this paper, we have mostly used branch lines on closed curves. As we have seen earlier, for closed branch line operators, we need to specify a reference vertex $v$. Imagine if we could define branch line operators on an open ribbon that starts from vertex $v_1$ and ends at vertex $v_n$, we have two natural reference vertices. In the case when $G$ is abelian, the gauging map will take such operators to ribbon operators defined in \cite{cheng2017exactly}. If we assume that the multiplication rule stays the same, then we have
\beq
    \mathcal{B}^x_l \mathcal{B}^{x^{-1}}_l= \sum_{g_1,g_n}\beta_{g_n x g_n^{-1}}(g_{n}g_{1}^{-1})\ket{g_1,g_n}_{v_1,v_n}\bra{g_1,g_n},
\eeq
where the factor 
\beq \label{eq:betaxg}
\beta_x(g)\equiv\epsilon_x(g)\epsilon_{x^{-1}}(g)\gamma_g(x,x^{-1})
\eeq
is the only object that is `pumped' out when we apply the anomalous non-onsite transformation and its inverse accordingly. One can check from the cocycle condition (see also Ref.~\cite{dijkgraaf1991quasi}) that
\beq
    \frac{\beta_{g^{-1}xg}(h)\beta_x(g)}{\beta_x(gh)}=1,
\eeq
showing that $\beta_x(g)$ satisfies the `twisted' cocycle condition.  One can thus use this to write 
\beq
    \beta_{g_n x g_n^{-1}}(g_{n}g_{1}^{-1})=\beta_{g_n x g_n^{-1}}(g_n)\beta_{g_1 x g_1^{-1}}(g_1)^{-1}.
\eeq
If $v$ is  one of the endpoints of the defect operators, then we essentially pump a factor $\beta_{g_v x g_v^{-1}}(g_v)$ to the state. The phase factor $\beta_x(g)$ is a conjugated 1-cocycle, i.e., $\beta_x(g)\in H^1(G,U(1)[G])\simeq \oplus_i H^1(\mathcal{Z}_i,U(1))$, where $i$ labels conjugacy classes of $G$, and $\mathcal{Z}_i$ is the corresponding centralizer group. We will call this factor the SPT pumping factor, since, by fusing defect operators, we pump a lower dimensional SPT state (in this case a $0$d SPT state) on the boundary of the line $l$ after the application of $\mathcal{B}^x_l \mathcal{B}^{x^{-1}}_l$. This is a generalization of a previous abelian case analyzed in Ref.~\cite{wen2022bulk}. If $G$ is abelian, for every element $x\in G$, the factor $\beta_x(g)$ is a 1d representation of group $G$, i.e., one pumps a 0d $G$-SPT state on the endpoints of an open ribbon.
\begin{figure}[h]
    \centering
    \includegraphics[width=\linewidth]{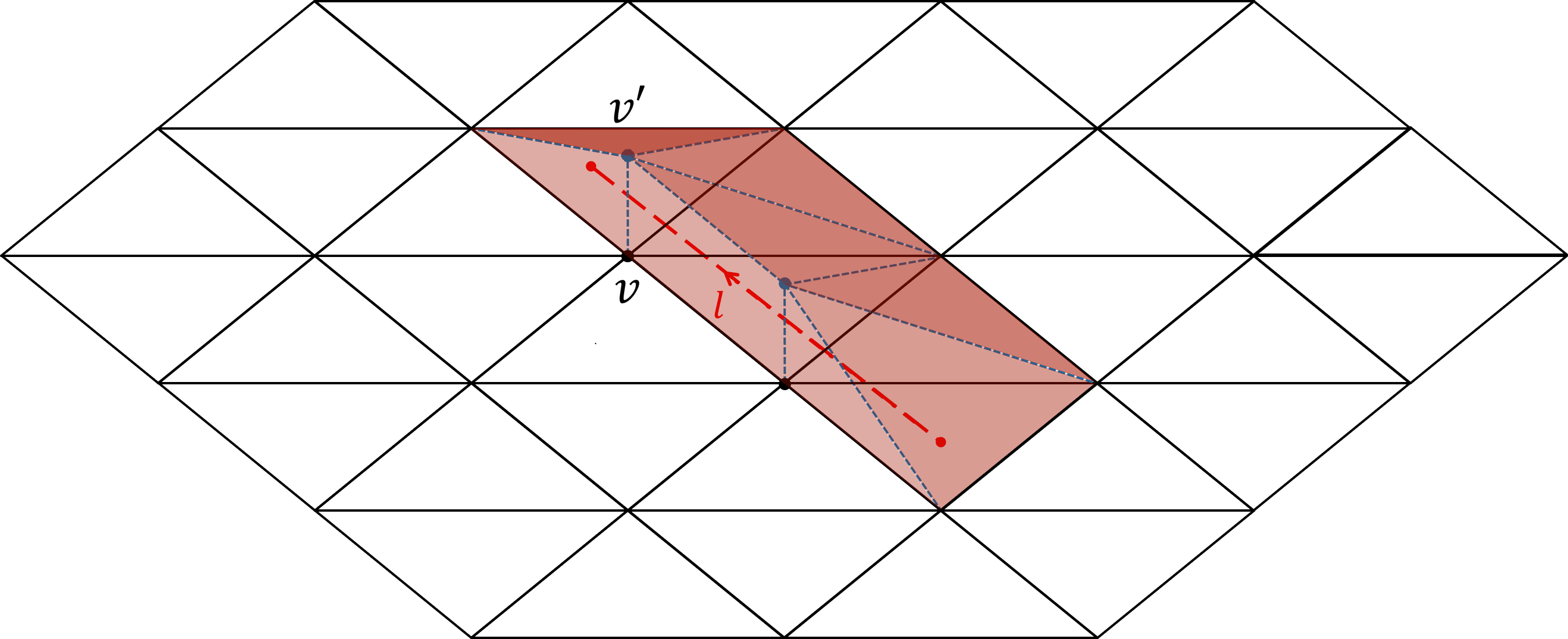}
    \caption{An open defect on a triangulated lattice.}
    \label{fig:open_defect}
\end{figure}

\subsection{Symmetry Branch Lines in SET Phase}
\label{sec:SymmetryBranchLinesSET}

As discussed in previous sections, we can gauge a normal subgroup $N$ of $G$. Then we enter an $N$ SET phase with global symmetry $Q=G/N$. Any group element $g$ in $G$ can be decomposed as $g=q(g)n(g)$. The SET wave function is 
 \beq
    \ket{\Psi_{\text{SET}}}=\sum_{\{q_v\},\{n_e\}}\Omega(\{q_i n_e q_j^{-1}\})\bigotimes_v\ket{q_v}\bigotimes_e\ket{n_e}.
 \eeq 
 
When the group element $x\in G$ commutes with all the elements $n\in N$ in the normal subgroup, a branch line operator $\mathcal{B}^x_{\partial\mathcal{R}}$ we introduced for the $G$-SPT state will be mapped to another operator $H^x_{\partial\mathcal{R}}$ under the gauging map $\Gamma_N$. Further, the operators defined as such respect the same multiplication rules as $\mathcal{B}^x_{\partial\mathcal{R}}$ do. Denoting $q_v=q(g_v)$, the operator $H^x_{\partial\mathcal{R}}$ is written as
\begin{align}
    H^{x}_{\partial\mathcal{R}}&=\sum_g H^{x,g}_{\partial\mathcal{R}}\epsilon_{q_v x q_v^{-1}}(g)\delta_{g,q_v (\prod_e n_e )q_v^{-1}},\\
&   \mbox{with}\,\, H^{x,g}_{\partial\mathcal{R}}=\mathcal{L}^{x}_{\partial\mathcal{R}} W^{q_v x q_v^{-1}}_{\partial\mathcal{R}}\Theta^{q_v x q_v^{-1}}_{\partial\mathcal{R}}\ket{q_v}\bra{q_v}.
\end{align}

When $x\in N$, the operator is a ribbon operator creating a gauge flux in the SET. When $n(x)=1$, i.e., $x=s(q)$ for some element $q\in Q$, the operator creates a flux that corresponds to an element in the global symmetry group and thus is a branch line operator. The multiplication rule of such operators is
\begin{align}
    H^{x,g}_{\partial\mathcal{R}}\,H^{y,g'}_{\partial\mathcal{R}}=H^{xy,g}_{\partial\mathcal{R}}\,\gamma_g(q_v x q_v^{-1},q_v y q_v^{-1})\,\delta_{g,g'},
    \label{eq:Hoperatormultiply}
\end{align}
 which, in turn, leads to
\beq
    H^{x}_{\partial\mathcal{R}}\,H^{y}_{\partial\mathcal{R}}=&\sum_{g}H^{xy,g}_{\partial\mathcal{R}}\beta_{q_v x q_v^{-1},q_v y q_v^{-1}}(g)\delta_{g,q_v(\prod_e n_e )q_v^{-1}},
    \label{eq:H_fusion}
\eeq
where the phase factor $\beta$ on the r.h.s. is 
\beq
    \beta_{x,y}(g)=\epsilon_{x}(g)\epsilon_{y}(g)\gamma_g(x,y).
    \label{eq:beta_xy}
\eeq

When we take two branch line operators $H^{x,g}_l$ and $H^{y,g}_l$, i.e., both $n(x)=n(y)=1$, and multiply them together, then the resulting $H^{xy,g}_l$ is not necessarily a branch line operator, because $n(xy)$ is not always trivial. As we will see in more details later, this indicates a nontrivial symmetry fractionalization class (SFC) of the SET order.

We can always apply a finite-depth local unitary to take all the vertex DOFs to the identity element (see appendix~\ref{sec:local unitary}), such that the state becomes
\beq
    \ket{\Psi}=\sum_{\{n_e\}}\Omega(\{n_e\})\ket{\{n_e\}}_e\otimes_v\ket{1}_v.
\eeq
This is a TQD with the 3-cocycle $\nu(n_1, n_2,n_3)$ being the restriction of $\omega(g_1,g_2,g_3)$ on subgroup $N$.   An abelian anyon in this model is determined by its flux (i.e., conjugacy class $C_a$) and charge (i.e., a conjugated 1-cochain $\mu_a$ such that $\tilde{\delta}\mu_a=\theta_a|_N$, i.e., $\mu_a=\epsilon_a|_N$). The symmetry action on the anyons are given by 
\begin{align}
    &\rho^x: a\text{-flux} ~ \rightarrow ~ x^{-1}ax\text{-flux},\label{eq: sym automorphism flux}\\
    &\rho^x: \mu_a  ~ \rightarrow ~ \mu'_{x^{-1}ax}. \label{eq: sym automorphism charge}
\end{align}

From our previous analysis, we know that the automorphism $\rho^x$ maps $a$-fluxed anyon to $x^{-1}ax$-fluxed anyon, thus we have Eq.~\eqref{eq: sym automorphism flux}. Furthermore, the automorphism will map a chargeon (i.e., a representation $\mu\in Rep(N)$) to $\mu'$, where $\mu'(x^{-1}ax)=\mu(a)$. Therefore, we can infer the general map of the anyon charge as
\beq
\mu'_{x^{-1}ax}(x^{-1}nx)=\frac{\theta_a(x,x^{-1}nx)}{\theta_a(n,x)}\mu_{a}(n).
\label{eq: sym automorphism general}
\eeq

However, the symmetry branch line operators we introduced above for the SET states assumed that $x$ commutes with all the elements in normal subgroup $N$. Therefore, the results obtained from analyzing such operators are limited to the case when the symmetry does not change the anyon type. To determine the SFC in this case, without loss of generality, we will suppose that the global symmetry is $Z_2$ in most of the cases from now on. (Examples beyond $Z_2$ will be discussed later.)
We first pick one abelian anyon in $\mathcal{C}_x$, i.e., we pick one $\epsilon_x$ cochain such that the corresponding branch line operator creates $0_x$. Then the SFC should be determined by the fusion rule $0_x\cdot 0_x=\mathcal{w}_{\tiny \openone}$. 

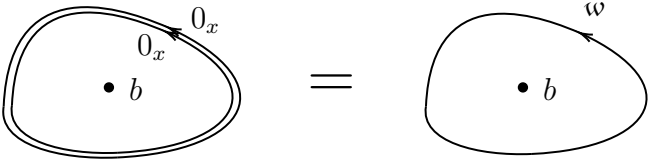
\begin{figure}[h]

\tikzset{every picture/.style={line width=0.75pt}}       

\begin{tikzpicture}[x=0.45pt,y=0.45pt,yscale=-1,xscale=1]

\draw    (66,138.47) .. controls (71,23.47) and (172,66.94) .. (195,76.47) ;
\draw    (66,138.47) .. controls (61,168.47) and (110,180.47) .. (155,178.47) ;
\draw    (155,178.47) .. controls (259.48,173.49) and (287.72,120.01) .. (196.39,77.11) ;
\draw [shift={(195,76.47)}, rotate = 24.81] [color={rgb, 255:red, 0; green, 0; blue, 0 }  ][line width=0.75]    (10.93,-3.29) .. controls (6.95,-1.4) and (3.31,-0.3) .. (0,0) .. controls (3.31,0.3) and (6.95,1.4) .. (10.93,3.29)   ;
\draw  [fill={rgb, 255:red, 0; green, 0; blue, 0 }  ,fill opacity=1 ] (144,123.5) .. controls (144,121.57) and (145.57,120) .. (147.5,120) .. controls (149.43,120) and (151,121.57) .. (151,123.5) .. controls (151,125.43) and (149.43,127) .. (147.5,127) .. controls (145.57,127) and (144,125.43) .. (144,123.5) -- cycle ;
\draw    (414,139.47) .. controls (419,24.47) and (520,67.94) .. (543,77.47) ;
\draw    (414,139.47) .. controls (409,169.47) and (458,181.47) .. (503,179.47) ;
\draw    (503,179.47) .. controls (607.48,174.49) and (635.72,121.01) .. (544.39,78.11) ;
\draw [shift={(543,77.47)}, rotate = 24.81] [color={rgb, 255:red, 0; green, 0; blue, 0 }  ][line width=0.75]    (10.93,-3.29) .. controls (6.95,-1.4) and (3.31,-0.3) .. (0,0) .. controls (3.31,0.3) and (6.95,1.4) .. (10.93,3.29)   ;
\draw  [fill={rgb, 255:red, 0; green, 0; blue, 0 }  ,fill opacity=1 ] (492,123.5) .. controls (492,121.57) and (493.57,120) .. (495.5,120) .. controls (497.43,120) and (499,121.57) .. (499,123.5) .. controls (499,125.43) and (497.43,127) .. (495.5,127) .. controls (493.57,127) and (492,125.43) .. (492,123.5) -- cycle ;
\draw    (59.03,139.57) .. controls (64.39,16.26) and (172.68,62.87) .. (197.34,73.09) ;
\draw    (59.03,139.57) .. controls (53.66,171.73) and (106.2,184.6) .. (154.45,182.45) ;
\draw    (154.45,182.45) .. controls (266.47,177.12) and (296.75,119.77) .. (198.82,73.78) ;
\draw [shift={(197.34,73.09)}, rotate = 24.81] [color={rgb, 255:red, 0; green, 0; blue, 0 }  ][line width=0.75]    (10.93,-3.29) .. controls (6.95,-1.4) and (3.31,-0.3) .. (0,0) .. controls (3.31,0.3) and (6.95,1.4) .. (10.93,3.29)   ;

\draw (168.86,76.4) node [anchor=north west][inner sep=0.75pt]  [font=\large]  {$0_{x}$};
\draw (162,113.4) node [anchor=north west][inner sep=0.75pt]  [font=\large]  {$b$};
\draw (543.86,49.4) node [anchor=north west][inner sep=0.75pt]  [font=\large]  {$\mathcal{w}$};
\draw (510,113.4) node [anchor=north west][inner sep=0.75pt]  [font=\large]  {$b$};
\draw (213.72,53.69) node [anchor=north west][inner sep=0.75pt]  [font=\large]  {$0_{x}$};
\draw (312,110) node [anchor=north west][inner sep=0.75pt]  [font=\Huge]  {$=$};

\end{tikzpicture}

\caption{The operator $(H^{x}_{\partial\mathcal{R}})^2$ applying on curve $\partial\mathcal{R}$ is equivalent to the braiding phase between the $g$-flux of an anyon $b$ ($g$ is the holonomy along $\partial\mathcal{R}$), and the charge of the anyon $\mathcal{w}$.\label{fig:braiding}}
\end{figure}

As shown in Fig.~\ref{fig:braiding}, the operator $(H^{x}_{\partial\mathcal{R}})^2$, when applied  on the DOFs on curve $\partial\mathcal{R}$, is to create two point defects $0_x$, move them along $\partial\mathcal{R}$, and finally make them cancel with each other.  
This process will introduce a phase factor that corresponds exactly to the braiding phase between the $g$-flux of an anyon $b$ ($g$ is the holonomy along $\partial\mathcal{R}$), and the charge of the anyon $\mathcal{w}$. 
According to Eq.~\eqref{eq:H_fusion}, this phase factor is $\beta_{x,x}(g)$. 
Notice that, when the symmetry group is $Q=G/N=Z_2$, the $q_v$ dependence in Eq.~\eqref{eq:H_fusion} disappears.
In general, such as in more than 2-step gauging, this factor still depends on $q_v$.  
Since the flux of $\mathcal{w}$ is $x^2$ (which is not necessarily the identity element, as $x$ is an embedding of the generator of $Q=Z_2$ into $G$),  suppose the charge of anyon $b$ is given by a 1-conjugated-cochain $\mu_g$ (such that $\tilde{\delta}\mu_g=\theta_g|_N$), then the braiding  phase between the $x^2$-flux and $\mu_g$ charge is $\mu_g(x^2)$. To determine $\mathcal{w}$, we write down the braiding phase between anyon $\mathcal{w}$ and $b$, which is the product of the above two factors,
\beq
    B(\mathcal{w},b)=\beta_{x,x}(g)\mu_g(x^2),
    \label{eq:braiding}
\eeq
and we illustrate this relation in Fig.~\ref{fig:braiding-phase}.
\begin{figure}[t]
\includegraphics[width=0.7\linewidth]{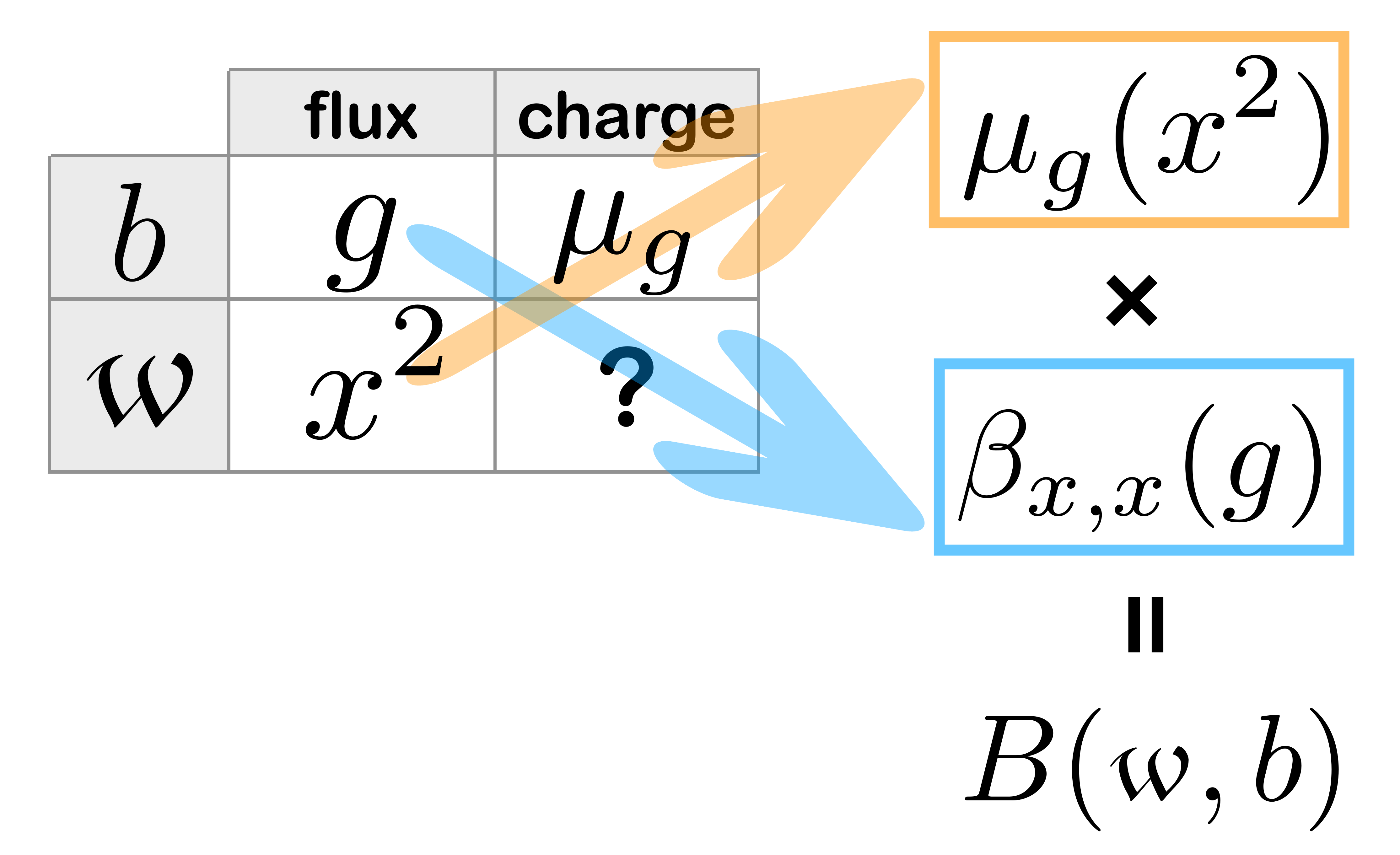}
\caption{
The braiding phase between a $g$-fluxed $\mu_g$-charged anyon $b$ and the anyon $\mathcal{w}(=0_x \times 0_x)$ is written as a product of two braiding phases, as presented in Eq.~\eqref{eq:braiding}. The second factor enclosed by the blue line is a consequence of Eq.~\eqref{eq:H_fusion}.
From the phase $B(\mathcal{w},b)$, we infer the unknown charge of the anyon $\mathcal{w}$ (written as ``?"), which determines the SFC.}
\label{fig:braiding-phase}
\end{figure}

From  this result, one can determine the SFC of the SET state after gauging some normal subgroup. Since we can always attach a 1-conjugated-cocycle $v_x$ to $\epsilon_x$ such that $\epsilon'_x=\epsilon_x v_x$, the phase $B(\mathcal{w},b)$ we derived above also has this ambiguity. But it just corresponds to the freedom to choose any abelian anyon in $\mathcal{C}_x$ that is labeled as $0_x$. Furthermore, we can always attach a coboundary to the 3-cocycle $\omega$ such that $\omega'=\omega\,\delta\alpha$, where $\alpha$ is a 2-cochain. One can check that the phase $B(\mathcal{w},b)$ is always the same as long as we choose $\omega$ in the same cohomology class.

In the next few subsections, we will use the above result to determine the SFC in different SET phases, resulting from gauging a normal subgroup $N$ of an SPT phase. In the later subsections, we will deal with  the case when symmetry does change the anyon type and conjecture the form of branch line operators  so as to use them to discuss the SFC of the SET orders that  we obtain from gauging Dihedral SPT states.

\subsection{SETs from partially gauging $Z_2\times Z_2$ SPT}
\label{sec:gaugingZ2timesZ2SPT}

The third cohomology of $Z_2\times Z_2$ has three generators, two of which are of type-1, and the other one is of type-2~\cite{propitius1995topological}. Assume the two generators of the $Z_2\times Z_2$ group are $t$ and $x$, then we can denote any group elements as $g=(g^{(1)},g^{(2)})\equiv t^{g^{(1)}}x^{g^{(2)}}$, where $g^{(1)},g^{(2)} \in \{0,1\}$. The representative of the 3-cocycle is then
\beq
    \omega_{g,h,l}=e^{\pi i (k_1 g^{(1)}h^{(1)}l^{(1)}+k_2g^{(2)}h^{(2)}l^{(2)}+k_3g^{(1)}h^{(2)}l^{(2)})},
    \label{eq:Z2Z2-cocycle}
\eeq
where $k_1,k_2,k_3=0,1$.

In the SET phase obtained from gauging the first $Z_2$ group, the anyon set is the same as that of a TQD $D^{\nu}(Z_2)$, where 
\beq 
    \nu(g,h,l)=e^{\pi i k_1 g^{(1)}h^{(1)}l^{(1)}}
\eeq
is a representative in $H^3(Z_2,U(1))$ obtained by the restriction of $\omega$ in the first $Z_2$ group (namely, one restricts the cocycles to those with $g^{(2)}=h^{(2)}=l^{(2)}=0$) and is used as the `twisting' of Kitaev's $Z_2$ QD model.

Therefore, when $k_1=0$ it is a toric code model, and when $k_1=1$ it is a double-semion model. Since $Z_2\times Z_2$ is the trivial central extension of $Z_2$ by $Z_2$, the symmetry action on anyons is trivial.

The second $Z_2$ group $\{1,x\}$ represents the global symmetry and therefore we can consider the multiplication of two branch line operators $H^{x}_{\partial\mathcal{R}}$, which according to Eq.~\eqref{eq:H_fusion}, gives 
\beq
    H^{x}_{\partial\mathcal{R}} H^{x}_{\partial\mathcal{R}} =\sum_{g} \beta_{x,x}(g) \delta_{g,\prod_e t_e}.
    \label{eq:algebra-H}
\eeq   
From Eq.~(\ref{eq:beta_xy}) and Eq.\eqref{eq:braiding}, we find the braiding phase between $\mathcal{w}$ and a $t$-fluxed anyon $b$,
\beq
B(\mathcal{w},b)=\epsilon_{x}(g)^2\gamma_g(x,x)=\theta_x(g,g)\gamma_g(x,x)=(-1)^{k_3 g^{(1)}}.
\eeq

For the conjugated 1-cochain we have used $\mu_g(x^2)= \mu_g(\tiny\openone)=1$.
Notice that since $g^2 = \openone$ for $g \in G$, 
we have $\theta_x(g,g)= \epsilon_x(g)^2 / \epsilon_x(g^2)$ by definition but $\epsilon_x(g^2)=\epsilon_x(\openone)=1$ so the second equality above follows.
Therefore, different choices of $t$-fluxed $b$ anyon and different choices of $\epsilon_x$ cochain (i.e., different choices of $0_x$) give rise to the same braiding phase.

Now we discuss the consequence of the resultant braiding $B(\mathcal{w},b)$ in different cases. As mentioned above, when $k_1=0$, we have a toric code model. From our previous general analysis in Sec.~\ref{sec:SymmetryBranchLinesSET}, we can infer that the anyon $\mathcal{w}$ braiding with $b$ ($m$ or $em$) gives rise to a phase $(-1)^{k_3}$. 
Therefore, 
\beq
    0_{x}\times 0_{x}=0 ~\text{or}~ e, ~ k_3=0\,\, \mbox{or}\,\, 1.
\eeq

As mentioned earlier, when $k_1=1$, we have a double-semion model. The fluxless anyon $\mathcal{w}$ braiding with $b$ ($s$ or $\bar{s}$) results in a phase $(-1)^{k_3}$. Therefore,
\beq
    0_x\times 0_x=0 ~\text{or}~ s\bar{s}, ~ k_3=0\,\, \mbox{or}\,\, 1.
\eeq

The discussion of $k_1$ and $k_3$ above completely specifies the SFC of the SET in this case. We have not discussed the consequence of $k_2$, but if we further gauge the second $Z_2$, different values of $k_i$ will give rise to different topological orders, due to the 1-to-1 correspondence between the SPT and TQD phases \cite{LevinGu2012, YOSHIDA2017387}. 
Therefore, we know that the intermediate SETs with different values of $k_2$ must belong to different phases. Since all the topological order parameters of SET, except the SDC, are already fixed by $k_1$ and $k_3$, we can safely conclude that $k_2=0,1$ corresponds to two defectification classes, respectively. Different defectifications intuitively can be regarded as stacking  or gluing different SPT phases~\cite{Barkeshli_2019} to the SET.
This particular case of SET phase was previously discussed in Ref.~\cite{MesarosRan2013}.

If we further gauge the global symmetry $Z_2$ in the SET, then it becomes a twisted quantum double $D^{\omega}(Z_2\times Z_2)$. As we discussed earlier in Sec.~\ref{sec:more on sbl}, for abelian groups, the symmetry branch line operators will be mapped to ribbon operators creating certain abelian anyons after gauging. Indeed, from the Slant product, the 1-conjugated cochain could be chosen as $\epsilon_{t_2}(g)= i^{k_2 g^{(2)}}$. The operator $H^x_{\partial\mathcal{R}}$ is mapped to a ribbon operator creating a $t_2$-flux anyon $\eta$,
\beq
    F^{\eta}_l=F^{t_2,1}_l+i^{k_2} F^{t_2,t_2}_l+F^{t_2,t_1}_l+i^{k_2} F^{t_2,t_1 t_2}_l.
\eeq
Since the gauge group is $Z_2\times Z_2$, the fusing of two such ribbon operators becomes a ribbon operator exciting a flux-less anyon (chargeon) $a$, $(F^{\eta}_l)^2=F^{a}_l$, similar to Eq.~(\ref{eq:algebra-H}),
\beq
    F^{a}_l=F^{1,1}_l+ F^{1,t_2}_l+(-1)^{k_3}F^{1,t_1}_l+(-1)^{k_3} F^{1,t_1 t_2}_l.
\eeq

For example, when $k_1=k_2=k_3=0$, the anyon $\eta$ is just a boson $m$ in toric-code model, and the anyon $a$ is the vacuum anyon. For any values of the parameters, we will see that the multiplication rules of branch line operators become the fusion rules of anyons under the gauging map.

\subsection{SETs from partially gauging $Z_4$ SPT}
We will use both multiplicative and additive representations of abelian groups interchangeably, e.g., $g=2$ means $g=x^2$ in multiplicative representation (for $x$ being the generator of $Z_4$). We take representative cocycles in $H^3(Z_4,U(1))=Z_4$ as 
\beq
    \omega_{g,h,l}=\exp{\frac{2\pi i p}{16}g(h+l-[h+l]_4)},
\eeq
where $p=0,1,2,3$. The slant product $\theta_{x^k}(g,h)=\exp{\frac{2\pi i p}{16}k(g+h-[g+h]_4)}$ corresponds to a projective representation given by $\epsilon_{x^k}(g)=\xi^{kg}$, where $\xi\equiv\exp{\frac{2\pi i p}{16}}$. 

An SET phase can be obtained by gauging the normal $Z_2=\{1,t\}$ group, where $t\equiv x^2$. By restricting $\omega$ in $H^3(Z_2,U(1))$, we have $\nu(g,h,l)=e^{\pi i p ghl}$, where  now $g,h,l=0,1$ are $Z_2$-valued. Therefore, when $p=0$ or 2, it is a toric code, and when $p=1$ or 3, it is a double-semion model. Since $Z_4$ is a central extension of $Z_2$ by $Z_2$, the symmetry action on anyons is trivial.

Let us recall that the branch line operators in the SET ground state are 
\beq                H^{x}_{\partial\mathcal{R}}\equiv\sum_g\epsilon_{x}(g)H^{x,g}_{\partial\mathcal{R}}=\sum_g \xi^{g}H^{x,g}_{\partial\mathcal{R}}.
\eeq

The product of two such branch line operators $H^x_{\partial\mathcal{R}}$ gives rise to a factor (see  Eq.~(\ref{eq:beta_xy})) 
\beq
    \beta_{x,x}(g)=\epsilon_{x}(g)^2\gamma_g(x,x)=e^{\frac{2\pi i p g}{8}},
\eeq
for $g=0,2$. The charge of $g$-fluxed anyon $b$ is given by \beq
\mu_g(h)=e^{\frac{2\pi i pgh}{16}+\frac{2\pi i r_g h}{4}},
\label{eq:mu for z4}
\eeq
for $g,h\in\{0,2\}=s(Q)$. Furthermore, $r_g=0,1$ corresponds to different choices of charges of anyon $b$. Thus, the braiding phase between anyon $\mathcal{w}$ and $b$  should be given as
\beq
\label{eq:betaZ4}
    B(\mathcal{w},b)=\beta_{x,x}(g)\mu_g(t)=e^{\frac{2\pi i p g}{4}}(-1)^{r_g}.
\eeq

When $p=0$, we have a toric code model. When $g=0$, $B(\mathcal{w},b)=(-1)^{r_0}$ where charge of anyon $b$ is given by $\mu_0(t)=(-1)^{r_0}$. There are two chargeons $0$ and $e$, corresponding to $r_0=0$ or $1$, respectively. Therefore, the braiding phase between $\mathcal{w}$ and $0$ ($e$) is $1$ ($-1$), according to Eq.~(\ref{eq:betaZ4}).  Moreover, when $g=2$, the braiding $B(\mathcal{w},b)=(-1)^{r_2}$ where the charge of $b$ is given by $\mu_2(h=t)=(-1)^{r_2}$ according to Eq.~\eqref{eq:mu for z4}. Therefore we could say that the braiding phase between $\mathcal{w}$ and $m$ ($em$) is $1$ ($-1$), which corresponds to $r_2=0,1$ respectively.  Therefore, we have a toric code with the following SFC:
\beq
    0_x\times 0_x=m.
\eeq

When $p=1$, this is a double-semion model. The braiding phase is $B(\mathcal{w},b)=(-1)^{g/2+r_g}$ where the charge of the $g$-fluxed anyon $b$ is given by $\mu_g(h=t)=e^{\frac{\pi i}{2}\cdot\frac{g}{2}+\pi i r_g}$  according to Eq.~\eqref{eq:mu for z4}. When $g=0$, two chargeons $0$ and $s\bar{s}$ correspond to $r_0=0$ and $1$ respectively. Anyon $\mathcal{w}$ braiding with $s\bar{s}$ gives $-1$. When $g=2$, the braiding phase between $\mathcal{w}$ and $s$ ($\bar{s}$) is $-1$ ($1$), which corresponds to $r_2=0,1$ respectively.
Therefore we have a double-semion model with SFC:
\beq
    0_x\times 0_x=s.
\eeq

When $p=2$, this is a toric code model. The braiding phase is $B(\mathcal{w},b)=(-1)^{r_g}$ where the charge of the $g$-fluxed anyon $b$ is given by $\mu_g(h=t)=e^{\pi i(\frac{g}{2}+r_g)}$  according to Eq.~\eqref{eq:mu for z4}. When $g=0$, two chargeons $0$ and $e$ correspond to $r_0=0$ and $1$ respectively. Anyon $\mathcal{w}$ braiding with $0$ ($e$) gives $1$ ($-1$). When $g=2$, the braiding phase between $\mathcal{w}$ and $em$ ($m$) is $1$ ($-1$), which corresponds to $r_2=0,1$ respectively. Therefore, we have a toric code with SFC:
\beq
    0_x\times 0_x=em.
\eeq

When $p=3$, this is a double-semion model. The braiding phase is $B(\mathcal{w},b)=(-1)^{g/2+r_g}$, where the charge of the $g$-fluxed anyon $b$ is given by $\mu_g(h=t)=e^{\frac{3\pi i}{2}\cdot\frac{g}{2}+\pi i r_g}$  according to Eq.~\eqref{eq:mu for z4}. When $g=0$, two chargeons $0$ and $s\bar{s}$ correspond to $r_0=0$ and $1$,  respectively. Anyon $\mathcal{w}$ braiding with $s\bar{s}$ gives $-1$. When $g=2$, the braiding phase between $\mathcal{w}$ and $\bar{s}$ ($s$) is $-1$ ($1$), which corresponds to $r_2=0,1$ respectively. Therefore, we have a double-semion code with SFC:
\beq
    0_x\times 0_x=\bar{s}.
\eeq
One could check that, if we choose other $\epsilon_x$ instead of what we used above, we would derive exactly the same fusion rule as above.

\subsection{SETs from partially gauging  $Z_2^3$ SPT}
\label{sec:gaugingz2cubeSPT}
The third cohomology group of $Z_2^{(1)}\times Z_2^{(2)}\times Z_2^{(3)}$ has seven generators, three of which are of type-1, three of which are of type-2, and one of type-3~\cite{propitius1995topological}. Assume the three generators of the $Z_2^3$ group are $t$, $x_1$ and $x_2$, then we can denote any group elements as $g=(g^{(1)},g^{(2)},g^{(3)})\equiv t^{g^{(1)}}x_1^{g^{(2)}}x_2^{g^{(3)}}$, where $g^{(1)},g^{(2)},g^{(3)} \in \{0,1\}$. For simplicity, in this section, we will demonstrate the analysis for representatives of some of the 3-cocycles, and then derive the general result without further explanation. The representatives that we take are
\beq
    \omega_{g,h,l}=e^{\pi i (k_1 g^{(1)}h^{(1)}l^{(1)}+k_2g^{(1)}h^{(2)}l^{(3)})},
    \label{eq:Z2Z2Z2-cocycle}
\eeq
where $k_1,k_2=0,1$.

In the SET phase obtained from gauging the group $Z_2^{(1)}$, the anyon set is the same as that of a TQD $D^{\nu}(Z_2)$, where 
\beq 
    \nu(g,h,l)=e^{\pi i k_1 g^{(1)}h^{(1)}l^{(1)}}
\eeq
is a representative in $H^3(Z_2,U(1))$ obtained by the restriction of $\omega$ in the first $Z_2$ group. Therefore, when $k_1=0$ it is a toric code model, and when $k_1=1$ it is a double-semion model. Since $Z_2^{(1)}\times Z_2^{(2)}\times Z_2^{(3)}$ is the trivial central extension of $Z_2^{(2)}\times Z_2^{(3)}$ by $Z_2^{(1)}$, the symmetry actions on anyons are trivial.

Since the slant product of the cocycle given above belongs to a class $[\theta]$ that is not the trivial element in $H^2(Z_2^3,U(1)[Z_2^3])$, it is impossible to find $\epsilon_h(g)$, such that 
\beq 
\theta_h(k,g)=\tilde{\delta}\epsilon_h(k,g)\equiv\frac{\epsilon_{h}(g)\epsilon_h(k)}{\epsilon_h(kg)},
\label{eq:epsilonequation}
\eeq
for any $g,h,k\in Z_2^3$. However, in defining the symmetry branch line operators, we only need phase factors $\epsilon_h(g)$ where the group element $h\in Z_2^{(2)}\times Z_2^{(3)}$ and $g\in N=Z_2^{(1)}$. Indeed in this case, there exists such a phase factor that satisfies Eq.\eqref{eq:epsilonequation} when restricting the group elements $h,g$ in their corresponding subgroups. 

The slant product of the cocycle is trivial,
\beq
    \theta_{h}(k,g)=(-1)^{h^{(2)} k^{(1)}g^{(3)}+h^{(3)} k^{(1)}g^{(2)}}=1,
\eeq
when $h\in Z_2^{(2)}\times Z_2^{(3)}$ and $k,g\in Z_2^{(1)}$. Therefore, we can choose $\epsilon_h(g)\equiv1$. 

In general, when the symmetry group is $Q=Z_2\times Z_2$, we take two elements $h_1,h_2\in s(Q)\subset G$ that are the embedding of elements $\tilde{h}_1,\tilde{h}_2\in Q$. The consistency condition of embedding is
\beq
    q(h_1h_2)=s(\tilde{h}_1\tilde{h}_2).
    \label{eq:embeddingequation}
\eeq
The fusion rule of $\mathcal{C}^{\times}_Q$ is of the form,
\beq
    0_{\tilde{h}_1}\times 0_{\tilde{h}_2}=\mathcal{w}(\tilde{h}_1,\tilde{h}_2)\times 0_{\tilde{h}_1\tilde{h}_2}.
\eeq
From our previous analysis, the braiding phase between $0_{\tilde{h}_1}\times 0_{\tilde{h}_2}$ and anyon $b=(g,\mu_g)$ is $B_1=\beta_{h_1,h_2}(g)\mu_g(n(h_1h_2))$. According to Eq.~\eqref{eq:embeddingequation}, the braiding phase between $0_{\tilde{h}_1\tilde{h}_2}$ and anyon $b$ is $B_2=\epsilon_{q(h_1h_2)}(g)$.
As a result, the braiding phase between abelian anyon $\mathcal{w}(\tilde{h}_1,\tilde{h}_2)$ and $b$ should be the ratio
\beq
    B(\mathcal{w}(\tilde{h}_1,\tilde{h}_2),b)=\frac{\epsilon_{h_1}(g)\epsilon_{h_2}(g)}{\epsilon_{q(h_1h_2)}(g)}\gamma_g(h_1,h_2)\mu_g(n(h_1h_2)).
    \label{eq:braidinggeneral}
\eeq

Later on for simplicity, we will use $\mathcal{w}(h_1,h_2)$ to denote $\mathcal{w}(\tilde{h}_1,\tilde{h}_2)$. Since we choose $\epsilon_h(g)\equiv1$ in this case and $n(h_1 h_2)\equiv 1$, we have 
\beq
    B(\mathcal{w}(h_1,h_2),b)=\gamma_g(h_1,h_2)=e^{\pi i k_2 g h_1^{(2)}h_2^{(3)}},
\eeq
for $g\in Z_2^{(1)}$. Since the group extension of $Z_2^{(2)}\times Z_2^{(3)}$ by $Z_2^{(1)}$ corresponds to the trivial element in $H^2(Z_2^{(2)}\times Z_2^{(3)},Z_2^{(1)})$, we know that the abelian anyon $\mathcal{w}(h_1,h_2)$ is always a chargeon for any $h_1,h_2\in Z_2^{(2)}\times Z_2^{(3)}$. When $k_1=0$, we have a $Z_2$ toric code model. From the above braiding phase we can conclude that when $k_2=0$, 
\beq
    \mathcal{w}(h_1,h_2)\equiv 1;
\eeq
when $k_2=1$, 
\beq
    \mathcal{w}(h_1,h_2)=\begin{cases}
    e, & h_1^{(2)}=h_2^{(3)}=1,\\
    1, & \text{others.}
    \end{cases}
\eeq

On the other hand, when $k_1=1$, we have a $Z_2$ double-semion model. From the above braiding phase we can conclude that when $k_2=0$, 
\beq
    \mathcal{w}(h_1,h_2)\equiv 1;
\eeq
when $k_2=1$, 
\beq
    \mathcal{w}(h_1,h_2)=\begin{cases}
    s\bar{s}, & h_1^{(2)}=h_2^{(3)}=1,\\
    1, & \text{others.}
    \end{cases}
\eeq

One can check that all the abelian anyons $\mathcal{w}(h_1,h_2)$'s above satisfy the cocycle condition,
\beq
    \frac{\mathcal{w}(h_2,h_3)\mathcal{w}(h_1,h_2h_3)}{\mathcal{w}(h_1h_2,h_3)\mathcal{w}(h_1,h_2)}=1.
\eeq

If one chooses different $\epsilon_h(g)$ other than what we used above, the derived anyon $\mathcal{w}(h_1,h_2)$ will be differed by a coboundary. Therefore we conclude, different values of $k_2$ will give different symmetry fractionalization patterns that correspond to different elements in $H^2(Z_2^{(2)}\times Z_2^{(3)},\mathcal{A})$, where $\mathcal{A}=Z_2\times Z_2$ is the group of abelian anyons.

One can generalize the above result to an arbitrary 3-cocycle. The cohomology group of $Z_2^3$ can be decomposed as such:
\begin{widetext}
    \begin{equation}
    \begin{tikzcd}
	{H^3(Z_2^3,U(1))} & {=} & {Z_2} & \times & {Z_2^3} & \times & {Z_2} \\
	&& {\text{anyon theory}} && {\text{SDC}} && {\text{SFC}_1} \\
	& \times & {Z_2} & \times & {Z_2} \\
	&& {\text{SFC}_1} && {\text{SFC}_2}
	\arrow["{\text{type-1 of } Z_2^{(1)}}", from=2-3, to=1-3]
	\arrow["{\text{type-1 and 2 of } Z_2^{(2)}\times Z_2^{(3)}}", from=2-5, to=1-5]
	\arrow["{\text{type-2 of } Z_2^{(1)}\times Z_2^{(2)}}", from=2-7, to=1-7]
	\arrow["{\text{type-2 of } Z_2^{(1)}\times Z_2^{(3)}}", from=4-3, to=3-3]
	\arrow["{\text{type-3 of }Z_2^{(1)}\times Z_2^{(2)}\times Z_2^{(3)}}", from=4-5, to=3-5]\,\,.
\end{tikzcd}
\label{eq:Z2cubecohomology}
\end{equation}
\end{widetext}

In this example, we have illustrated 2 out of the 7 generators in $H^3(Z_2^3,U(1))$ as in Eq.~\eqref{eq:Z2Z2Z2-cocycle} and showed that $k_1$ gives the anyon theory and $k_2$ (which is associated with the type-3 cocycle) gives a symmetry fractionalization pattern named SFC$_2$ in the above diagram. To understand the rest of SET properties, we note that the two SFC$_1$'s are the symmetry fractionalization pattern associated with type-2 cocycles of $Z_2^{(1)}\times Z_2^{(2)}$ and $Z_2^{(1)}\times Z_2^{(3)}$, respectively,  which were already discussed in Sec.~\ref{sec:gaugingZ2timesZ2SPT}. The SDC part is the symmetry defectification class associated with cocycles of $Z_2^{(2)}\times Z_2^{(3)}$, both of type-1 and type-2. In the cases when the $Z_2^3$-SPT phase corresponds to the cohomology class which is trivial in the first $Z_2$ subgroup in Eq.~\eqref{eq:Z2cubecohomology}, one can choose a representative that is of some specific form. Then after gauging $Z_2^{(1)}$ subgroup, according to Ref.~\cite{cheng2017exactly}, one can determine the symmetry fractionalization patterns of the SET order, which agrees with our general results above.

\subsection{SETs from partially gauging $D_4$ SPT}
\label{sec:gaugingd4spt}
Now we consider the non-central extension of $Z_2$ by $Z_4$. We write the element in $D_4$ as $\Tilde{g}=(G,g)\equiv x^Ga^g$. We construct a representative of 3-cocycle in $H^3(D_4,U(1))$ as follows:
\beq
\label{eq:D4-cocycle}
    &\omega(\Tilde{g},\Tilde{h},\Tilde{l})  \\
    &=\exp{\frac{2\pi i p_1}{16}g(-1)^{H+L}(h(-1)^L+l-[h(-1)^L+l]_4)}\\
    &\qquad \times \exp{\pi i p_2 GHL+ \pi i p_3 gHL},
\eeq
where $p_1=0,1,2,3$, and $p_2,p_3=0$ or $1$. 
There are four nontrivial abelian normal subgroups in $D_4$, which leads to four options in the first step when gauging this group. We will consider three of them here.

\smallskip \noindent {\bf Gauging $Z_2$}.   The first option is to gauge the normal subgroup $Z_2=\{1,a^2\}$, resulting in a state in an SET that has the same anyon set as $D^{\nu}(Z_2)$, where
\beq
    \nu(g,h,l)=\exp{\frac{2\pi i p_1}{16}ghl}
\eeq
is the restriction of $\omega$ on $Z_2$, i.e. $g,h,l\in \{1,a^2\}$. When $[p_1]_2=0$ it is a toric code model, and when $[p_1]_2=1$ it is a double-semion model. Since the group extension of $Z_2\times Z_2$ by $Z_2$ is central,
the symmetry actions on the anyons are trivial. Therefore, according to Eq.~\eqref{eq:braidinggeneral}, the braiding phase $B(\mathcal{w}(h_1,h_2),b)$ is given by
\beq
    B(\mathcal{w}(h_1,h_2),b)=\frac{\epsilon_{h_1}(g)\epsilon_{h_2}(g)}{\epsilon_{q(h_1h_2)}(g)}\gamma_g(h_1,h_2)\mu_g(n(h_1h_2)).
\eeq

We write the embedding of quotient group elements
\beq
g=(g^{(1)},g^{(2)})\in Z_2\times Z_2\equiv\{1,t_1\}\times\{1,t_2\}
\eeq
in $D_4$ as 
\beq
s(g)=x^{g^{(1)}}a^{g^{(2)}}.
\eeq

From the group multiplication rule, one can infer that the abelian anyons $\mathcal{w}(t_2,t_1)$, $\mathcal{w}(t_2,t_2)$, $\mathcal{w}(t_1t_2,t_1)$ and $\mathcal{w}(t_1t_2,t_2)$ have nontrivial flux, while $\mathcal{w}(h_1,h_2)$ for other $h_1,h_2$ are chargeons. We list the detailed symmetry fractionalization patterns below.

When $[p_1]_2=0$, we have a $Z_2$ toric code model. From the above braiding phase we can conclude that the SFC is characterized by $[\mathcal{w}(h_1,h_2)]\in H^3(Z_2\times Z_2,Z_2\times Z_2)$, where
\beq
    \mathcal{w}(h_1,h_2)=\begin{cases}
    m, & (h_1,h_2)=(t_2,t_1), (t_2,t_2),\\
    & ~~~~~~~~~(t_1t_2,t_1), (t_1t_2,t_2)\\
    1, & \text{others.}
    \end{cases}
\eeq

When $[p_1]_2=1$, we have a $Z_2$ double-semion model. From the above braiding phase we can conclude that, the SFC is characterized by $[\mathcal{w}(h_1,h_2)]\in H^3(Z_2\times Z_2,Z_2\times Z_2)$, where
\beq
    \mathcal{w}(h_1,h_2)=\begin{cases}
    s, & (h_1,h_2)=(t_2,t_1), (t_1t_2,t_2), (t_2,t_2), \\
    \bar{s}, & (h_1,h_2)=(t_1t_2,t_1), \\
    s\bar{s}, & (h_1,h_2)=(t_1,t_1), (t_1,t_1t_2), (t_1t_2,t_1t_2), \\
    1, & \text{others}.
    \end{cases}
\eeq

Other parameters of the cohomology group $H^3(D_4,U(1))$, including $\frac{p_1-[p_1]_2}{2}, p_2$ and $p_3$, will give rise to different SDCs that form an $H^3(Z_2\times Z_2,U(1))=Z_2^3$ torsor.

\smallskip \noindent {\bf Gauging $Z_4$}. The second option is to gauge the normal subgroup $Z_4$, resulting in a state in an SET that has the same anyon set as $D^{\nu}(Z_4)$, where
\beq
    \nu(g,h,l)=\exp{\frac{2\pi i p_1}{16}g(h+l-[h+l]_4)}
\eeq
is the restriction of $\omega$ on $Z_4$. Different values of $p_1$ exactly correspond to different $Z_4$ TQD models. The symmetry action takes $e$ to $e^3$, and takes $m$ to $m^3 e^{2p_1}$ according to Eq.~\eqref{eq: sym automorphism general}, which is not a trivial automorphism on $\mathcal{C}$. One can still manage to write a phase factor $B(\mathcal{w},b)=\beta_{x,x}(g)$ for $g\in N=Z_4$. However, two obvious problems will emerge from this factor. The first one is that unlike in the case when symmetry does not change the anyon type, when we change the representative 3-cocycle for the $D_4$-SPT state by a coboundary, $\omega'=\omega\cdot \delta \alpha$, the ``braiding phase'' is not invariant anymore, $B(\mathcal{w},b)'=B(\mathcal{w},b)\frac{\alpha(g^{-1}xg,g^{-1}xg)}{\alpha(x,x)}$. The second problem seems to be even worse. In a generic case, it could be impossible to find an abelian object in sector $\mathcal{C}_x$ as the $0_x$ we take before. Therefore, it is suspicious to talk about abelian anyon $\mathcal{w}$ as the fusion between $0_x$ and itself. Indeed, in sector $\mathcal{C}_x$, there are 4 objects of quantum dimension $2$. If we nonetheless pick one of them and name it as $0_x$, by  counting the dimension, we can write a fusion rule of the form,
\beq
    0_x \times 0_x = a + b + c + d,
\eeq
where $a,b,c,d\in \mathcal{C}$ are abelian anyons.

Motivated by the ribbon operator in the quantum double model as in Eq.~\eqref{eq:ribbon qd}, we choose $b_1=1$ and $b_2=a$ and we write a matrix-valued operator on an open ribbon as
\beq
    (H^x_l)_{ii'}=\sum_{n\in\{1,a^2\}}H_l^{b_i x b_i^{-1},b_i n b_{i'}^{-1}}\epsilon_{b_i x b_i^{-1}}(b_i n b_{i'}^{-1}),
    \label{eq:branchlinenonabelian}
\eeq
where the matrix indices $i,i'=1,2$, and the operator $H_l^{x,g}$ satisfies the same multiplication rule as in Eq.~\eqref{eq:Hoperatormultiply},
\beq
    H_l^{x,g}H_l^{y,g'}=H_l^{xy,g}\gamma_g(x,y)\delta_{g,g'}.
\eeq

We conjecture that the operator as in Eq.~\eqref{eq:branchlinenonabelian} creates an object in sector $\mathcal{C}_x$ on the endpoint of $l$. We call this object $0_x$ even though it is of dimension 2. Then the object $0_x\times 0_x$ should be created on the endpoint of $l$ by operator $(H^x_l)^{\otimes2}$. It can be shown that, when we change the representative 3-cocycle for the $D_4$-SPT state by a coboundary, $\omega'=\omega\cdot \delta \alpha$, the matrix $(H^x_l)^{\otimes2}$  differs by a similar transformation. Therefore, the fusion rule remains invariant under different representative choices. According to the detailed analysis in Appendix \ref{sec:fusionruled4 appendix}, we see that different values of $p_3$ give different SFCs where the fusion rules are shifted by anyon $[e^2]\in H^2_{\rho}(Z_2,\mathcal{A})$.

\smallskip\noindent {\bf Gauging $Z_2\times Z_2$}. One could also gauge the $Z_2\times Z_2=\{1,x,t,xt\}$ in $D_4$, resulting a state in the phase of $D^{\nu'}(Z_2\times Z_2)$. We write $t\equiv a^2$ and $g=x^{g^{(1)}}t^{g^{(2)}}=x^{g^{(1)}}a^{2g^{(2)}}$. The 3-cocycle $\omega$ restricted in this group is obtained from Eq.~\eqref{eq:D4-cocycle} and is given as
\beq    
    \nu'(g,h,l)=&\exp\Big\{\frac{2\pi i p_1}{4}g^{(2)}(-1)^{h^{(1)}+l^{(1)}}\big(h^{(2)}(-1)^{l^{(1)}}+l^{(2)}\\
    &-[h^{(2)}(-1)^{l^{(1)}}+l^{(2)}]_2\big)+\pi i p_2 g^{(1)}h^{(1)}l^{(1)}\Big\},\\
    =&(-1)^{p_1 (g^{(2)}h^{(2)}l^{(2)}+g^{(2)}h^{(2)}l^{(1)})+p_2 g^{(1)}h^{(1)}l^{(1)}}.
    \label{z2z2nu'}
\eeq
Notice that there is no contribution from the third part in Eq.~(\ref{eq:D4-cocycle}) as $e^{\pi i p_3 (2g^{(2)})h^{(1)}l^{(1)}}\equiv1$. In Appendix \ref{braidingphase}, we analyze the fluxes and charges of all the anyons in the theory from 3-cocycle $[\nu']$. Let $b_1=1$ and $b_2=x$, one can write the matrix-valued operator on an open ribbon $l$ as
\beq
    (H^a_l)_{ii'}=\sum_{n\in\{1,t\}}H_l^{b_i a b_{i}^{-1},b_i n b_{i'}^{-1}}\epsilon_{b_i a b_{i}^{-1}}(b_i n b_{i'}^{-1}),
\eeq
where the matrix indices have the range $i,i'=1,2$. As we conjectured, the object $0_a\times 0_a$ should be created on the endpoint of $l$ by operator $(H^a_l)^{\otimes2}$. According to the detailed analysis in Appendix~\ref{sec:fusionruled4 appendix z2z2}, we see that different values of $p_3$ give different SFCs where the fusion rules are shifted by anyon $[e^{(1)}]\in H^2_{\rho}(Z_2,\mathcal{A})$.

\subsection{SETs from partially gauging $S_3$ SPT}
Now with the conjecture made in Sec.~\ref{sec:gaugingd4spt}, we can revisit our first example in Sec.~\ref{sec:Z3SETwithZ2}. Recall that we write the element in $S_3$ as $\Tilde{g}=(G,g)\equiv x^Ga^g$. We construct a representative of 3-cocycle in $H^3(S_3,U(1))$ as in Eq.~\eqref{eq:s_3 cocycle}. Gauging the normal subgroup $Z_3$ of a $S_3$-SPT state, results in a state in an SET phase that has the same anyon set as $D^{\nu}(Z_3)$, where
\beq
    \nu(g,h,l)=\exp{\frac{2\pi i p_1}{9}g(h+l-[h+l]_3)}
\eeq
is the restriction of $\omega$ on $Z_3$. Different values of $p_1$ exactly correspond to different $Z_3$ TQD models. The symmetry action takes $e$ to $e^2$, and takes $m$ to $m^2 e^{2p_1}$ according to Eq.~\eqref{eq: sym automorphism general}, which is also not a trivial automorphism on $\mathcal{C}$, as we have seen something similar in the previous $D_4$ case. 

In sector $\mathcal{C}_x$, there is only one object of quantum dimension $3$. We name it $0_x$. By the dimension counting, we can write a fusion rule of the form,
\beq
    0_x \times 0_x = \sum_{i=1}^9 a_i,
\eeq
where $a_i\in \mathcal{C}$ are abelian anyons. Let $b_1=1$, $b_2=a$ and $b_3=a^2$, one can write a matrix-valued operator on an open ribbon as
\beq
    (H^x_l)_{ii'}=H_l^{b_i x b_i^{-1},b_i b_{i'}^{-1}}\epsilon_{b_i x b_i^{-1}}(b_i b_{i'}^{-1}),
\eeq
where the matrix indices are in the range $i,i'=1,2,3$. As we conjectured in the last example, the object $0_x\times 0_x$ should be created on the endpoint of $l$ by operator $(H^x_l)^{\otimes2}$. According to the detailed analysis in Appendix \ref{sec:fusionruleS3 appendix}, we can obtain the fusion rule of the $Z_3$ SET from the conjectured branch line operator,
\beq
    0_x\times 0_x=1 + e + e^2 + m + em + e^2m + m^2 + em^2 + e^2m^2.
\eeq
From the analysis in Sec.~\ref{sec:gaugingd4spt}, we know that there is only one symmetry fractionalization pattern for every value of $p_1$. This unique SFC result is consistent with the fact that the fusion in the above  equation is the same for all values of $p_2$ (with a fixed $p_1$, i.e., fixing a distinct anyon theory), and thus $p_2$ gives different SDCs, unrelated to the SFC.

\subsection{SETs from  $D_{n}$ SPT}

Here we comment on the SET phase obtained from gauging the $N=Z_{n}$ subgroup in the $D_{n}$ SPT state.
When $n=2m+1$ is odd, from the similar argument we used for the $S_3$-SPT state, there is only one symmetry fractionalization pattern of such SET. The cohomology group can be decomposed as
\beq
H^3(D_{2m+1},U(1))&=H^3(Z_{2m+1},U(1))\oplus H^3(Z_{2},U(1))\\
&=Z_{2m+1}\oplus Z_2.
\eeq
Therefore just as in $S_3=D_3$ case, a representative 3-cocycle $[\omega]\in H^3(D_{2m+1},U(1))$ will have two parameters $p_1=0, ..., 2m$ and $p_2=0,1$. Different values of $p_1$ give different anyon theory of the SET order, while  different values of $p_2$ give different SDCs. Furthermore, there is only one object $0_x$ in sector $\mathcal{C}_x$ of dimension $2m+1$, and from a similar calculation, one expects the fusion rule to be
\beq
    0_x \times 0_x=\sum_{a\in\mathcal{C}}a.
\eeq

When $n=2m$ is even, the cohomology group can be decomposed as
\beq
H^3(D_{2m},U(1))=&H^3(Z_{2m},U(1))\oplus H^3(Z_{2},U(1))\\
&\quad\oplus H^2(Z_2,H^1(Z_{2m},U(1)))\\
=&Z_{2m}\oplus Z_2\oplus Z_2.
\eeq
Therefore just as in $D_4$ case, a representative 3-cocycle $[\omega]\in H^3(D_{2m},U(1))$ will have three parameters $p_1=0, ..., 2m-1$, $p_2=0,1$ and $p_3=0,1$. Different values of $p_1$ give different anyon theory of the SET order and different values of $p_2$ give different SDCs. Furthermore, different values of $p_3$ will differ in the fusion of $0_x\times 0_x$ by anyon $e^m$,  and, therefore, correspond to different symmetry fractionalization patterns.

\subsection{SETs from partially gauging $Q_8$ SPT}
Another group extension of $Z_2$ by $Z_4$ is the $Q_8$ group. We write the element in $Q_8$ as $\Tilde{g}=(G,g)\equiv x^Ga^g$ just as for $D_4$. The only difference is that $x^2=a^2$ instead of identity now. A representative of 3-cocycle in $H^3(Q_8,U(1))$ is~\cite{propitius1995topological} 
\beq
    &\omega(\Tilde{g},\Tilde{h},\Tilde{l})=\exp\Big\{\frac{2\pi i p}{8}\Big(-2GHL+ \\
    &g(-1)^{H+L}\big(h(-1)^L+l-[h(-1)^L+l+2HL]_4\big)\Big)\Big\},
\eeq
where $p=0,1,2,3$. We note that despite the fact that $H^3(Q_8,U(1))=Z_8$, we only present half of the cocycles here, as we are not aware of the other half. After gauging the normal subgroup $Z_2=\{1,a^2\}$, we obtain an SET states in which the anyon set is the same as in $D^{\nu}(Z_2)$ where $\nu(g,h,l)=\exp{\frac{2\pi i p}{8}ghl}$, with $g,h,l=0,2$ representing elements from the set $\{1,a^2\}$. Therefore, $p=0,1,2,3$ all correspond to the $Z_2$ toric code model after gauging. The symmetry action on anyons are trivial since the group extension is central. According to Eq.~\eqref{eq:braidinggeneral}, the  braiding phase $B(\mathcal{w}(h_1,h_2),b)$ is given by
\beq
    B(\mathcal{w}(h_1,h_2),b)=\frac{\epsilon_{h_1}(g)\epsilon_{h_2}(g)}{\epsilon_{q(h_1h_2)}(g)}\gamma_g(h_1,h_2)\mu_g(n(h_1h_2)).
\eeq

Let us write the  quotient group elements as $Z_2\times Z_2\equiv\{1,t_1\}\times\{1,t_2\}$, where the embedding of $t_1(t_2)$ is $x(a)$. After carrying out the detailed calculations from the cocycles above, the SFC corresponds to $[\mathcal{w}(h_1,h_2)]\in H^2(Z_2\times Z_2,Z_2\times Z_2)$, where
\beq
    \mathcal{w}(h_1,h_2)=\begin{cases}
    m, & (h_1,h_2)=(t_1,t_1), (t_2,t_1), (t_2,t_2),\\
    & ~~~~~~~~~~~~(t_1,t_1t_2), (t_1t_2,t_2),\\
    1, & \text{others.}
    \end{cases}
\eeq

 Therefore, the parameter $p$ characterizes different SDCs of the SET order. One expects that for the other 4 classes of cocycles in $H^3(Q_8,U(1))$, the anyon theory after gauging the normal $Z_2$ subgroup would be $Z_2$ double-semion model, and from similar calculations, one can determine the SFC accordingly.

\section{Conclusion}\label{sec:conclusion}

Recently, it has been realized that a wide class of topologically ordered states described by the (twisted) quantum double models with solvable gauge groups can be prepared with finite depth local operations as long as local measurements are included~\cite{verresen2021schorodinger, Bravyi2022a, Hierarchy, tantivasadakarn2022shortest, iqbal2023topological,foss-feig2023experimental}.
We have re-examined such a measurement-based gauging approach which transforms a non-trivial SPT state into a corresponding TQD state. We provided two alternative gauging procedures: one using a particular decomposition in terms of successive quotient groups and another one exploiting a new and equivalent definition of solvable groups. This flexibility in our method may allow us different options in preparing mid-gauging SET states.

In the case of non-abelian groups, the gauging procedure involves multiple steps where intermediate steps only partially gauge the system so that some symmetry remains. Starting from an initial $G$-SPT state, we have presented an in-depth analysis of the intermediate states and have found them to be topologically ordered states enriched by the remaining ungauged symmetry.
We have constructed the generic lattice (parent) Hamiltonian for these states, and showed that they are connected to twisted quantum double (TQD) ground states via a finite-depth local unitary circuit (without measurements) which does not respect the global symmetry.  

Furthermore, we have shown that the algebra of the symmetry branch line operators can be used to extract the symmetry fractionalization classes and infer symmetry defectification classes of the SET phases given the input data $G$ and $[\omega]\in H^3(G,U(1))$. When the SET order  in the intermediate step of the $N$-step gauging has a global symmetry that \emph{does not} change the anyon type, using the algebra of symmetry branch line operators, we have developed a general formula for the braiding phases between any abelian anyon in the theory and the anyons obtained from fusing point defects, which exactly characterize the symmetry fractionalization patterns. We have given various examples for this case. When the SET order we enter has a global symmetry that \emph{does} change anyon types, we conjectured the form and algebra of non-abelian symmetry branch line operators that can create the corresponding symmetry  defects. Then by calculating the tensor product of such operators, we showed that fusion rules of these symmetry defects can be derived, which is sufficient to characterize the symmetry fractionalization patterns. We have used the dihedral SPT states and the associated SET states as  examples to illustrate this 
 latter case.

In this work, we mainly focused on the SFC, and a framework to characterize the SDC is left for future study. 
We note that, according to Ref.~\cite{Barkeshli_2019}, the SDC forms a $H^3(Q,U(1))$ torsor, and two defectification classes are differed by an element in $H^3(Q,U(1))$.
One can always enter another SDC by applying a unitary $U_{\omega'}$ to vertex DOFs (where $[\omega']\in H^3(Q,U(1))$). 
This is equivalent to stacking a $Q$-SPT state onto the current SET state. 

Our method to probe SET phases using fusion was inspired by Ref.~\cite{Hierarchy} and it turns out to be specifically useful when the gauge group (or the normal subgroup in the case of multi-step gauging) is abelian. We expect our formalism holds for non-abelian TQD models with some global symmetry.
We leave as open questions how to consistently define ribbon operators for probing more complex SET phases with non-abelian gauge groups. For this purpose, it may be useful to re-examine some literature regarding the quasi-Hopf algebra~\cite{propitius1995topological, dijkgraaf1991quasi,majid1998quantum}. 

As a technical issue, in writing down the branch line operators, we have assumed the existence of $\epsilon_x(g)$ (see Eq.~\eqref{eq:conj1cochain}) for the 3-cocycle $[\omega]\in H^3(G,U(1))$. However, this factor may not exist in general. 
We did encounter this situation in the example of gauging  the $Z_2^{(1)}$ symmetry in the $Z_2^{(1)}\times Z_2^{(2)}\times Z_2^{(3)}$-SPT phase when the cocycle is of type-3. Nonetheless, when restricting $x$ and $g$ to some specific subgroups, we can still define $\epsilon_x(g)$, and we used them for the branch line operators to characterize the symmetry fractionalization patterns. It is not clear how to overcome the non-existence of $\epsilon_x(g)$ in general, and this is also left for future exploration.

As conjectured by the hierarchy of topological orders conceptualized in Ref.~\cite{Hierarchy}, topologically ordered states with non-solvable groups or even more general anyon models without any group structure cannot be prepared using finite-depth measurement-assisted circuits.
Nonetheless, it was recently shown in Ref.~\cite{Lu2022} that the Fibonacci anyon state can be prepared using $\log L$-depth circuits with mid-circuit measurements (where $L$ is the linear size of the system).
It would be worth exploring mid-gauging topological phases beyond solvable groups and group-based anyon theories.

\acknowledgments
The authors would like to  thank Norbert Schuch, David T. Stephen,  Nat Tantivasadakarn, and Ruben Verresen for useful discussions. This work was supported by the U.S. Department of Energy, Office
of Science, National Quantum Information Science Research
Centers, Co-design Center for Quantum Advantage (C2QA)
under Contract No. DE-SC0012704, in
particular, regarding the algorithmic procedure of the  gauging, 
and 
 by the National Science Foundation under Grant No. PHY 1915165, in particular, regarding the physical properties related to symmetry-enriched topological phases. 
 HPN acknowledges support by the Austrian Science Fund (FWF) through the SFB BeyondC F7102, and the European Union (ERC Advanced Grant, QuantAI, No. 101055129). Views and opinions expressed are however those of the author(s) only and do not necessarily reflect those of the European Union or the European Research Council. Neither the European Union nor the granting authority can be held responsible for them.
 \bibliography{ref}
 \appendix
 \begin{widetext}

\section{Some properties of twisted quantum double (TQD) Hamiltonian}\label{TQDprop}
In this section, we check some properties of twisted quantum double (TQD) Hamiltonian. Although some proofs of the claims here are given in the original paper \cite{PhysRevB.87.125114}, we demonstrate those in our notation.
The TQD Hamiltonian is given by
\begin{align}
    H=-\sum_v A_v-\sum_p B_p,
\end{align}
where $A_v$ is the vertex operator and $B_p$ is the plaquette operator explicitly given by
\begin{subequations}
    \begin{align}
    A_v&=\frac{1}{|G|}\sum_{g\in G} \left(\prod_{e\supset v}L_{\pm e}^g\right)\Tilde{W}_v^g=\frac{1}{|G|}\sum_{g\in G}A_{v,g},\\
    B_p&=\delta\Big(\prod_{e\in p} g_e,1\Big).
\end{align}
\end{subequations}
The operator $\left(\prod_{e\supset v}L_{\pm e}^g\right)$ denote the operator which implement left action ($L_{+e}^g$) or right action($L_{-e}^g$) on the edges adjacent to the vertex $v$ when the edge flows to vertex $v$ or emanates from vertex $v$ respectively. $\Tilde{W}_v^g$ is the phase operator defined as follows
    \begin{align}
    \Tilde{W}_v^g=\left(\prod_{e\supset v}L_{\pm e}^g\right)^{\dagger}U_{\omega}\left(\prod_{e\supset v}L_{\pm e}^g\right)U_{\omega}^{\dagger}
    \label{eq:W_v^g},
\end{align}
where $U_{\omega}$ is the phase operator which assigns a phase to a given configuration of edges in the TQD ground state(see Eq. \eqref{TQDwfn}). The purpose of $\Tilde{W}_v^g$ is to change the phase factor in TQD wave-function after the operation $\prod_{e\supset v}L_{\pm e}^g$ is applied on the vertex $v$. \\
\par 
\begin{claim}
    The action of $\Tilde{W}_v^g$ on vertex $v$ can be interpreted geometrically as a product of cocylces of tetrahedrons in Fig.~\ref{fig:W-phase} with appropriate signs in the exponent.
\end{claim} 
\begin{proof}
    To prove this equivalence, we mention the following fact: given a tetrahedron with a branching structure, equating one with the product of (a) cocycles (with appropriate signs in the exponent) on the faces and (b) the cocylce on the tetrahedron (also with appropriate sign in the exponent) gives the cocycle condition. The signs in the exponent can be found using the following rule:
\begin{itemize}
    \item First we define the orientation of a face of a tetrahedron. Consider the branching structure of a face. Curl your fingers on the right hand along the direction of two arrows which point one after the other. The direction your thumb points at gives you the orientation of the face.
    \item For the tetrahedron, consider the vertex where all the arrows end and  the face opposite to it. If the orientation of the face points inward to the tetrahedron, the sign is +1, otherwise it is -1.
    \item For a face, simply consider its orientation. If the orientation points inward to the tetrahedron, assign the sign to be +1, otherwise assign -1.
\end{itemize}
As an example, consider the following tetrahedron given in Fig.~\ref{fig:tetrahedron3}.
\begin{figure}[h!]
    \centering
    \includegraphics[scale=0.3]{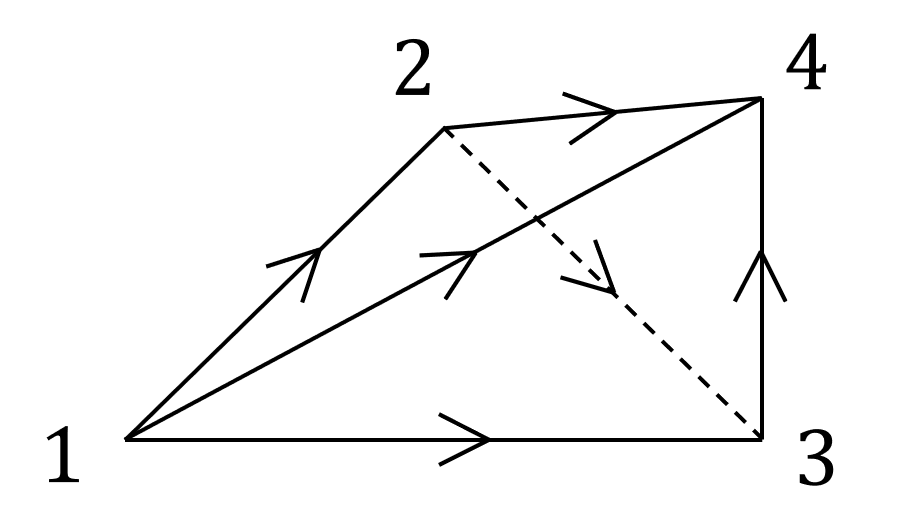}
    \caption{Tetrahedron with orientation}
    \label{fig:tetrahedron3}
\end{figure}\\
The product of cocycles with appropriate signs in the exponent gives the cocycle condition,
\begin{align}
    \frac{\omega(g_4g_2^{-1},g_2g_1^{-1},g_1)\omega(g_4g_3^{-1},g_3g_2^{-1},g_2)}{\omega(g_4g_3^{-1},g_3g_2^{-1},g_2g_1^{-1})\omega(g_3g_2^{-1},g_2g_1^{-1},g_1)\omega(g_4g_3^{-1},g_3g_1^{-1},g_1)}=1.
\end{align}
Now, multiplying all the cocycle conditions coming from all the tetrahedron adjacent to vertex $v$ as in Fig.~\ref{fig:W-phase}, one can clearly see that the cocycles coming from the faces shared by two tetrahedron cancel in pairs since the signs in the exponent coming from the two adjacent tetrahedrons of a face are opposite. Finally, the remaining product of cocycles can be rewritten as
\begin{align}
    \prod_{tetra}\omega(\text{tetra})^{s(\text{tetra})}=\prod_{\Delta}\frac{1}{\omega(\Delta')^{s(\Delta')}\omega(\Delta)^{s(\Delta)}}=\prod_{\Delta}\left(\frac{\omega(\Delta')}{\omega(\Delta)}\right)^{s(\Delta)},
    \label{eq:prodtetra}
\end{align}
where $\Delta$ and $\Delta'$ denote the triangles in the original and the lifted plane (after action by $A_v$), $s(\Delta)$ and $s(\Delta')$ denote the signs in the exponent for the cocycle coming from $\Delta$ and $\Delta'$. The last equality in Eq.~\eqref{eq:prodtetra} follows from the fact that $s(\Delta')=-s(\Delta)$. The expression which follows the last equality is exactly what Eq.~\eqref{eq:W_v^g}  achieves.
\end{proof}
\begin{claim}
   $A_v$ is hermitian, i.e, $A_v^{\dagger}=A_v$.
\end{claim}
\begin{proof}
By definition, we have
\begin{align}
    A_{v,g}^{\dagger}=\left(\left(\prod_{e\supset v}L_{\pm e}^g\right)\Tilde{W}_v^g\right)^{\dagger}=(\Tilde{W}_v^g)^*\left(\prod_{e\supset v}L_{\pm e}^g\right)^{\dagger},
\end{align}
where $(\Tilde{W}_v^g)^*$ denote the complex conjugate of $\Tilde{W}_v^g$. We note that $\left(\prod_{e\supset v}L_{\pm e}^g\right)^{\dagger}=\left(\prod_{e\supset v}L_{\pm e}^{g^{-1}}\right)$. From the definition Eq. \eqref{eq:W_v^g}, we have
\begin{align}
    (\Tilde{W}_v^g)^*=U_{\omega}\left(\prod_{e\supset v}L_{\pm e}^g\right)^{\dagger}U_{\omega}^{\dagger}\left(\prod_{e\supset v}L_{\pm e}^g\right).
\end{align}
Now we compute $A_{v,g}^{\dagger}$:
\begin{subequations}
    \begin{align}
   (\Tilde{W}_v^g)^*\left(\prod_{e\supset v}L_{\pm e}^g\right)^{\dagger}&= U_{\omega}\left(\prod_{e\supset v}L_{\pm e}^g\right)^{\dagger}U_{\omega}^{\dagger}\\
   &=\left(\prod_{e\supset v}L_{\pm e}^g\right)^{\dagger}\left(\prod_{e\supset v}L_{\pm e}^g\right)U_{\omega}\left(\prod_{e\supset v}L_{\pm e}^g\right)^{\dagger}U_{\omega}^{\dagger}\\
   &=\left(\prod_{e\supset v}L_{\pm e}^{g^{-1}}\right)\left(\prod_{e\supset v}L_{\pm e}^{g^{-1}}\right)^{\dagger}U_{\omega}\left(\prod_{e\supset v}L_{\pm e}^{g^{-1}}\right)U_{\omega}^{\dagger}\\
   &=\left(\prod_{e\supset v}L_{\pm e}^{g^{-1}}\right)\Tilde{W}_v^{g^{-1}}=A_{v,g^{-1}}.
\end{align}
\end{subequations}
Hence it holds that
\begin{align}
    A_v^{\dagger}=\frac{1}{|G|}\sum_{g\in G}A_{v,g}^{\dagger}=\frac{1}{|G|}\sum_{g\in G}A_{v,g^{-1}}=A_v.
\end{align}
\end{proof}
\begin{claim}
    $A_v$ is a projector. $A_v^2=A_v$.
\end{claim}
\begin{proof}
Note the following observation.
\begin{subequations}
    \begin{align}
    A_{v,g}A_{v,h}&=U_{\omega}\left(\prod_{e\supset v}L_{\pm e}^{g}\right)U_{\omega}^{\dagger}U_{\omega}\left(\prod_{e\supset v}L_{\pm e}^{h}\right)U_{\omega}^{\dagger}\\
    &=U_{\omega}\left(\prod_{e\supset v}L_{\pm e}^{gh}\right)U_{\omega}^{\dagger}=A_{v,gh}.
\end{align}
\end{subequations}
Hence, we can square $A_v$ and arrive at
\begin{align}
    A_v^2=\frac{1}{|G|^2}\sum_{g,h\in G}A_{v,g}A_{v,h}=\frac{1}{|G|^2}\sum_{g,h\in G}A_{v,gh}=\frac{1}{|G|}\sum_{g\in G}A_{v,g}=A_v.
\end{align}
\end{proof}
\begin{claim}
    $B_p$ is hermitian as well as a projector.
\end{claim}
\begin{proof}
    This follows trivially from the definition Eq.~(\ref{eq:BpTQD}). 
\end{proof}
\begin{claim}
    $[A_v,A_{v'}]=0$ for any vertices $v$ and $v'$.
\end{claim}
\begin{proof}
First we prove $[A_{v,g},A_{v',h}]=0$ when $v\neq v'$.
    \begin{align}
    \begin{split}
        [A_{v,g},A_{v',h}]&=\left[U_{\omega}\Big(\prod_{e\supset v}L^g_{\pm e}\Big)U_{\omega}^{\dagger},U_{\omega}\Big(\prod_{e\supset v'}L^h_{\pm e}\Big)U_{\omega}^{\dagger}\right]\\
        &=U_{\omega}\Big[\prod_{e\supset v}L^g_{\pm e},\prod_{e\supset v'}L^h_{\pm e}\Big]U_{\omega}^{\dagger}=0.
    \end{split}
    \end{align}
    The last line is trivial when $v$ and $v'$ are not adjacent. When they are adjacent, the two vertices have opposite (right/left) action on the edge DOF. So their commutator is again zero. Hence $[A_{v,g},A_{v',h}]=0$ when $v\neq v'$. This imply $[A_v,A_{v'}]=0$ when $v\neq v'$. When $v=v'$,  the commutator is trivially zero. So $[A_v,A_{v'}]=0$, $\forall$ $v$ and $v'$.
\end{proof}
  \begin{claim}
      $[A_v,B_p]=0$ $\forall$ $v$ and $p$.
  \end{claim}
  \begin{proof}
      When the vertex $v$ is not on the boundary of the plaquette $p$, the two terms commute trivially. When the vertex is on the boundary of $p$, we consider the two edges which are adjacent to the vertex $v$ as well as lie on the boundary of $p$. Now consider three cases. 
      \begin{itemize}
          \item Case 1: Both edges point toward the vertex $v$. Action of $A_{v,g}$ on this configuration is given by left multiplying $g$ on the corresponding edges. This preserves the fluxless condition imposed by the $B_p$ operator.
          \item Case 2: Both edges point away from the vertex $v$. Action of $A_{v,g}$ on this configuration is given by right multiplying $g^{-1}$ on the corresponding edges. This preserves the fluxless condition imposed by the $B_p$ operator.
          \item Case 3: One edge point toward the vertex $v$ and the other edges point away from it. Action of $A_{v,g}$ on this configuration is given by left multiplying by $g$ and right multiplying by $g^{-1}$ respectively. This also preserves the fluxless condition.
      \end{itemize}
      From this observation it follows that $[A_v,B_p]=0$ $\forall$ $v$ and $p$.
  \end{proof}
  \begin{claim}
      $[B_p,B_{p'}]=0$ $\forall$ plaquettes $p$ and $p'$.
  \end{claim}
  \begin{proof}
      The proof follows straightforwardly from the fact $B_p=1$ on the configurations for which there is no flux around plaquette $p$, and zero otherwise.
  \end{proof}
Now we consider the quantum double like Hamiltonian in the presence of a global symmetry given in Eq.~(\ref{eq:QDSET})
\begin{align}
    H=-\sum_v A_v-\sum_p B_p-\sum_v K_v,
\end{align}
where $A_v$, $B_p$ and $K_v$ are defined in Eq.~(\ref{eq:AvSET}), Eq.~(\ref{eq:BpSET}) and Eq.~(\ref{eq:KvSET}) respectively.
Again $A_v$ is hermitian as well as a projector. Similarly $B_p$ is also hermitian as well as a projector. $A_v$ and $B_p$ commute among themselves and with each other. The proofs follow by repeating the steps in the twisted quantum-double case. Now we consider the last term $K_v$.
\begin{claim}
    $K_v$ is hermitian as well as a projector.
\end{claim}
\begin{proof}
    Let us write 
    \begin{align}
        K_v=\frac{1}{|Q|}\sum_{k,l=0}^{|Q|-1}K_{v,kl},
    \end{align}
    where $K_{v,kl}=W_v^{q_kq_l^{-1}}\ket{q_k}_v\bra{q_l}$. We can write the phase operator $W_{v}^{q_kq_l^{-1}}$ as
    \begin{align}
        W_v^{q_kq_l^{-1}}\ket{q_k}_v\bra{q_l}=U_{\omega}\ket{q_k}_v\bra{q_l}U_{\omega}^{\dagger} .
    \end{align}
Using the Hermitian conjugation of the above equation, we have  \begin{subequations}
        \begin{align}
        K_{v,kl}^{\dagger}
        &=U_{\omega}\ket{q_l}\bra{q_k}U_{\omega}^{\dagger}=K_{v,lk}.
    \end{align}
    \end{subequations}
    Hence, we have $K_v^{\dagger}=K_v$. 
     Next, we prove $K_v^2=K_v$ using the following steps,
\begin{subequations}
        \begin{align}
        K_v^2&=\frac{1}{|Q|^2}\sum_{k,l,m,n=0}^{|Q|-1}U_{\omega}\ket{q_k}_v\bra{q_l}U_{\omega}^{\dagger}U_{\omega}\ket{q_m}_{v}\bra{q_n}U_{\omega}^{\dagger}\\
        &=\frac{1}{|Q|^2}\sum_{k,l,m,n=0}^{|Q|-1}U_{\omega}\ket{q_k}_v\bra{q_n}U_{\omega}^{\dagger}\delta_{l,m}\\
        &=\frac{1}{|Q|}\sum_{k,l=0}^{|Q|-1}K_{v,kn}=K_v.
    \end{align}
    \end{subequations}
\end{proof}
\begin{claim}
    $K_v$ commute with $A_{v'}$, $B_p$ and $K_{v'}$ $\forall$ $v$ and $v'$.
\end{claim}
\begin{proof}
First we prove $[K_v,K_{v'}]=0$. For this, we show $[K_{v,kl},K_{v',mn}]=0$ when $v\neq v'$, as shown below,
\begin{subequations}
    \begin{align}
    [K_{v,kl},K_{v',mn}]&=\Bigl[U_{\omega}\ket{q_k}_v\bra{q_l}U_{\omega}^{\dagger},U_{\omega}\ket{q_m}_{v'}\bra{q_n}U_{\omega}^{\dagger}\Bigr]\\
    &=U_{\omega}\Bigl[\ket{q_k}_v\bra{q_l},\ket{q_m}_{v'}\bra{q_n}\Bigr]U_{\omega}^{\dagger}=0.
\end{align}
\end{subequations}
Hence $[K_v,K_{v'}]=0$ $\forall$ $v$ and $v'$.
Note that $K_v$ operator only changes the vertex DOF on the lattice,  so $K_v$ does not change the fluxless condition around any of the plaquettes. Hence, $[K_v,B_p]=0$.

Now let us prove $[K_v,A_{v'}]=0$ by the following steps, \begin{subequations}
     \begin{align}
     [K_v,A_{v'}]&=\Bigl[\frac{1}{|Q|}\sum_{k,l=0}^{|Q|-1}U_{\omega}\ket{q_k}_v\bra{q_l}U_{\omega}^{\dagger},\frac{1}{|G|}\sum_{g\in G}U_{\omega}\prod_{e\supset v'}L^g_{\pm e}U_{\omega}^{\dagger}\Bigr]\\
     &=\frac{1}{|Q|}\sum_{k,l=0}^{|Q|-1}\frac{1}{|G|}\sum_{g\in G}\Bigl[U_{\omega}\ket{q_k}_v\bra{q_l}U_{\omega}^{\dagger},U_{\omega}\prod_{e\supset v'}L^g_{\pm e}U_{\omega}^{\dagger}\Bigr]\\
     &=\frac{1}{|Q|}\sum_{k,l=0}^{|Q|-1}\frac{1}{|G|}\sum_{g\in G}U_{\omega}\Bigl[\ket{q_k}_v\bra{q_l}, \prod_{e\supset v'}L^g_{\pm e}\Bigr]U_{\omega}^{\dagger}=0.
 \end{align}
 \end{subequations}
From the above equations, we thus conclude that  $K_v$ commutes with $A_{v'}$, $B_p$ and $K_{v'}$.
\end{proof}

\section{Group Extension}
\label{app:extension}
Suppose we are given two groups $Q$ and $N$, then one can construct an extension of $Q$ by $N$ which we denote by $G$ if one has the following short exact sequence
\begin{align}
    1\rightarrow N\xrightarrow[]{i} G \xrightarrow[]{\pi} Q \rightarrow 1.
\end{align}
where $i$ denote the inclusion map and $\pi$ denote the projection map.
Given this short exact sequence, one can define a choice of embedding of $Q$ in $G$ 
\begin{align}
    Q\xrightarrow[]{s} G,
\end{align}
such that $\pi\circ s=id_Q$. Although the inclusion $i$ and projection $\pi$ are homomorphisms, the section $s$ is not a homomorphism (however, we have $s(1_Q)=1_G$). The failure to become a homomorphism is captured by a cocycle in the group cohomology $H^2(Q,N)$. The failure of the section to be a homomorphism is given by
\begin{align}
    s(q_1q_2)^{-1}s(q_1)s(q_2)=i(\omega(q_1,q_2)),
\end{align}
where $\omega\in H^2(Q,N)$.
We consider conjugation operation $\sigma:Q\rightarrow Aut(N)$ which satisfies
\begin{align}
    i(\sigma^{q_2^{-1}}(n_1))=s(q_2)^{-1}i(n_1)s(q_2).
\end{align}
Note that the conjugation operation is dependent on the choice of $s$.
We denote an element $g\in G$ as $g=(q,n)$, where $q\in Q$ and $n\in N$. With the given choice of section $s$, one can equivalently write
\begin{align}
   g=s(q)i(n). 
\end{align} 
Suppose $g_1=(q_1,n_1)$ and $g_2=(q_2,n_2)$ then 
\begin{align}
    g_1.g_2=(q_1,n_1).(q_2,n_2)=(q_1q_2,\omega(q_1,q_2)\sigma^{q_2^{-1}}(n_1)n_2).
\end{align}
The associativity condition on the group multiplication gives the cocycle condition for $\omega$,
\begin{align}
    \omega(q_1,q_2q_3)\omega(q_2,q_3)=\omega(q_1q_2,q_3)\sigma^{q_3^{-1}}[\omega(q_1,q_2)],
\end{align}
when $N$ is abelian. Note that this cocycle condition is different from the one where the conjugation acts on $\omega(q_2,q_3)$ considered in \cite{Hierarchy}.
\newline
 As an example, consider the central extension of $\mathbb{Z}_2$ by $\mathbb{Z}_2$. Since $H^2(\mathbb{Z}_2,\mathbb{Z}_2)=\mathbb{Z}_2$, there are two  possible extensions.
\begin{itemize}
    \item Case 1: $\omega$ is trivial. In this case, we have $G=\mathbb{Z}_2\times \mathbb{Z}_2$ since $Aut(\mathbb{Z}_2)$ consists only  of the identity map.
    
    \item Case 2: $\omega$ is the nontrivial class. Then the group $G$ is $\mathbb{Z}_4$. Suppose we denote $N=\mathbb{Z}_2=\{1,t_1\}$  and $Q=\mathbb{Z}_2=\{1,t_2\}$, then $G=\mathbb{Z}_4$ is generated by $(t_2,1)$.
\end{itemize}

\section{$Q$-global symmetry in two-step gauging} 
\label{sec:q-global-sym}

Here we show the $Q$-global symmetry of $|\Psi_5\rangle$ in \eqref{eq:psi-5}:
\beq
U_{(q)} |\Psi_5 \rangle = |\Psi_5 \rangle,
\eeq
with
\beq
&|\Psi_5 \rangle 
:= \sum_{\{g_v\}}
\Omega (\{g_v g^{-1}_{v'}\})
 |\{g_v g^{-1}_{v'}\} \rangle_e 
\bigotimes_v 
\Big( Z^{k_v}_{(n)}
\sum_{r \in N}| q(g_v) r \rangle_v
\Big), \\
& U_{(q)} := \prod_{v} X_{(q)}.
\eeq 
We have that $\sum_{r\in N}X_{(q_1)} | q(g_v) r \rangle 
:=\sum_{r\in N} |s(q_1) q(g_v)r \rangle = \sum_{r\in N}|q(s(q_1) g_v))n(s(q_1) q(g_v))r  \rangle =\sum_{r'\in N} |q(g_v s(q_1) )r'  \rangle$, where the first equality is by definition (see \eqref{eq:x}), 
the second equality is by the definitions of $q(\dots)$ and $n(\dots)$ in Eq.~\eqref{eq:n}, and in the last equality we have relabeled the group element $r \in N$ by $r'=n(s(q_1)q(g_v) )r$ and used $q(s(q_1) g_v) = q(g_v s(q_1))$.
(Note that in the second equality, for $G=S_3$ the element $n(s(q_1) q(g_v))$ is trivial, but for $G=Q_8$, for example, it is nontrivial, and thus we kept it present in the equation.)
We find that the change of variable $g_v = \widetilde{g}_v s(q_1)^{-1}$ gives us
\beq
U_{(q_1)} |\Psi_5 \rangle
=  \sum_{\{\widetilde{ g}_v\}}
\Omega (\{\widetilde{g}_v \widetilde{g}^{-1}_{v'}\})
 |\{ \widetilde{g}_v \widetilde{g}^{-1}_{v'}\} \rangle_e 
\bigotimes_v 
\Big( Z^{k_v}_{(n)}
\sum_{r' \in N}| q(\widetilde{ g}_v) r' \rangle_v , 
\Big),
\eeq
and thus the state is invariant.

\section{Branch line operator $\mathcal{B}^x_{\partial\mathcal{R}}$}\label{sec: branch line operator appendix}

With a pre-gauge structure we have introduced the gauge transformation as
\beq 
G^x_v\equiv L^x_{-v} \prod_{e\supset v}L^{x}_{\pm e}.
\eeq
To write down the branch line operator on $\partial\mathcal{R}$, we impose some gauge transformation in region $\mathcal{R}$, the effect of which can be contained only on its boundary. In Sec.~\ref{sec:sbl}, we claim that when the edge configuration is trivial (i.e., $h_e=1$ for all edges), the operator $G^x_{\mathcal{R}}=\prod_{v\in \mathcal{R}}G^x_v$ has this property. Now we make some more detailed analysis.
Recall that
\beq
    \ket{\Psi_{\text{SPT-pre}}}=\ket{\Psi_{\text{SPT}}}\bigotimes_e\ket{1}_e.
\eeq
Now we insert a $G^x_{\mathcal{R}}$ on the state, we have the following,
\beq
    G^x_{\mathcal{R}}\ket{\Psi_{\text{SPT-pre}}}&=\prod_{v\in\mathcal{R}}L^x_{-v}\big(\prod_{e\supset v}L^x_{\pm e}\big)\ket{\Psi_{\text{SPT}}}\bigotimes_e\ket{1}_e\\
    &=\sum_{\{g_v\}}\prod_{\Delta}\omega(\{g_v\})^{s(\Delta)}\bigotimes_{v_i\in\mathcal{R}}\ket{g_{v_i} x^{-1}}\bigotimes_{v_o\notin\mathcal{R}}\ket{g_{v_o}}\bigotimes_{e}\big(\prod_{{e'
    \cap \partial\mathcal{R}}}L^x_{\pm e'}\big)\ket{1}\\
    &=\sum_{\{g_v\}}\prod_{\Delta}\omega(\{g_{v_i} x\},\{g_{v_o}\})^{s(\Delta)}\bigotimes_{v}\ket{g_v}\bigotimes_{e}\big(\prod_{{e'
    \cap \partial\mathcal{R}}}L^x_{\pm e'}\big)\ket{1}\\
    &=\sum_{\{g_v\}}Amp^{\mathcal{R}}(\{g_v\},x)\prod_{\Delta}\omega(\{g_v\})^{s(\Delta)}\bigotimes_{v}\ket{g_{v}}\bigotimes_{e}\big(\prod_{{e'
    \cap \partial\mathcal{R}}}L^x_{\pm e'}\big)\ket{1}.\label{eq:global gauge}
\eeq

\begin{figure}[h]
    \centering
    \begin{subfigure}[h]{0.45\linewidth}
    \centering
    \includegraphics[width=\linewidth]{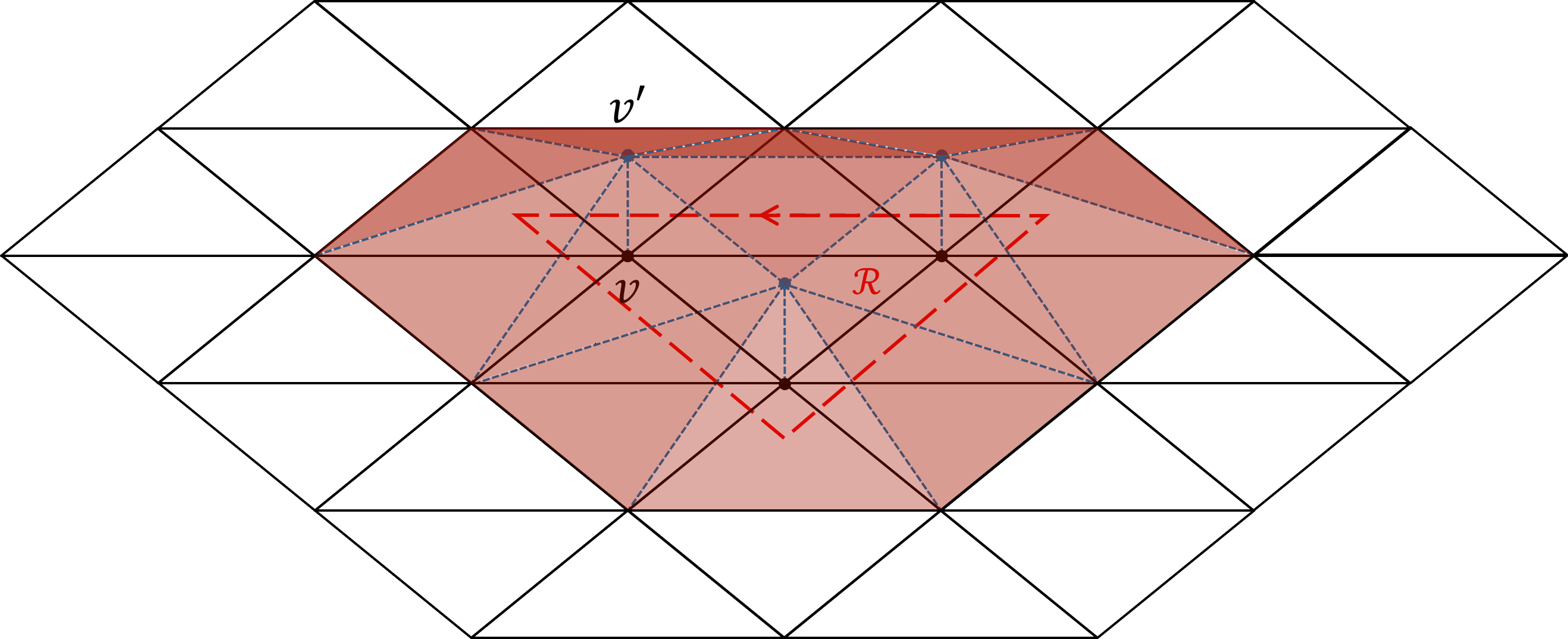}
    \caption{}
    \label{fig:closed symmetry appendix}
    \end{subfigure}
    \hfill
    \begin{subfigure}[h]{0.45\linewidth}
    \centering
    \includegraphics[width=\linewidth]{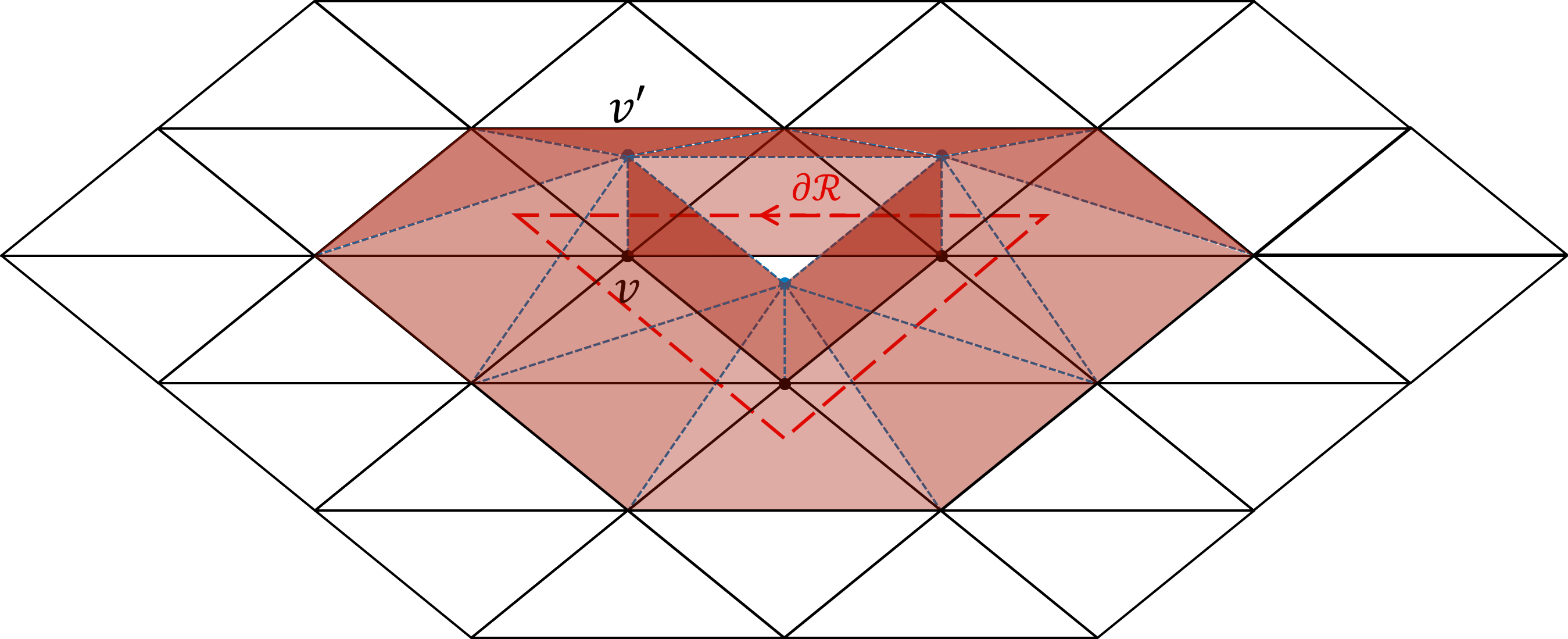}
    \caption{}
    \label{fig:closed defect appendix}
    \end{subfigure}
    \caption{(a) Symmetry action inside of region $\mathcal{R}$ ``lifts'' $\mathcal{R}$ such that all the simplex in the region correspond to $\Tilde{\omega}$. (b) This symmetry action can 
 be equivalently  regarded as the insertion of symmetry branch line on $\partial\mathcal{R}$. }
\end{figure}

As shown above in Fig.~\ref{fig:closed symmetry appendix}, given the configuration $\{g_v\}$, the  product of cocycles $\prod_{\Delta}\omega(\{g_{v_i} x\},\{g_{v_o}\})^{s(\Delta)}$ corresponds to the upper surface, while $\prod_{\Delta}\omega(\{g_{v_i}\},\{g_{v_o}\})^{s(\Delta)}$ corresponds to the lower surface. According to the cocycle conditions, their ratio corresponds 
to the tetrahedrons that they enclose,
\beq  
    Amp^{\mathcal{R}}=\prod_{\text{tetra}\in\mathcal{R}} \omega(\text{tetra})^s.
\eeq

In Sec.~\ref{sec:sbl}, we claim that according to cocycle conditions, this phase factor only depends on the configurations on $\partial\mathcal{R}$. Now we take the simplest case when $\mathcal{R}$ only encloses one plaquette to illustrate. 
\begin{figure}[h!]
    \centering
    \includegraphics[scale=0.3]{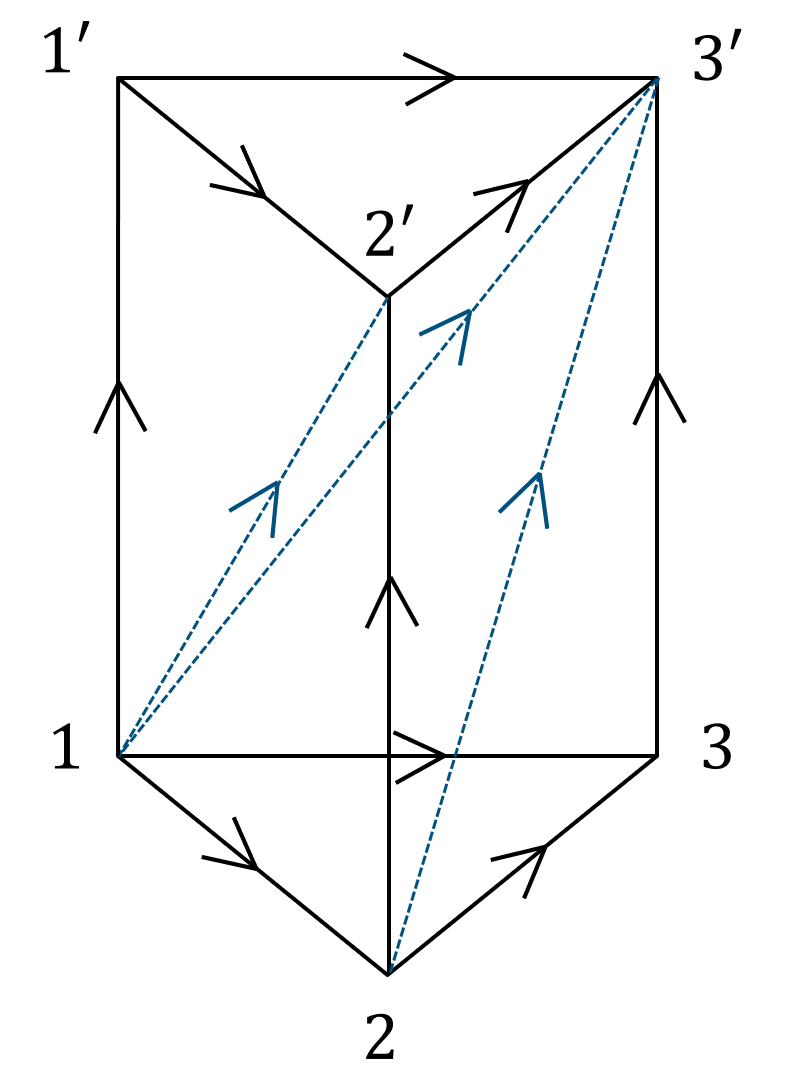}
    \caption{There are 3 tetrahedrons when 1 plaquette is lifted.}
    \label{fig:3tetra}
\end{figure}
As shown in Fig.~\ref{fig:3tetra}, when the plaquette $(123)$ is lifted, there are three tetrahedrons, where the vertices on the upper surface are associated with $g_{i'}=g_i x$. According to the rule introduced in Sec.~\ref{sec:GaugingMap}, the three tetrahedrons correspond to the expression,
\beq
    \frac{\omega(g_3 x g_3^{-1},g_3 g_2^{-1},g_2 g_1^{-1})\omega(g_3g^{-1}_2,g_2g^{-1}_1,g_1 x g_1^{-1})}{\omega(g_3g_2^{-1},g_2xg_2^{-1},g_2g_1^{-1})}=\theta_{g_3xg_3^{-1}}(g_3g_2^{-1},g_2g_1^{-1}),
\eeq
where $\theta_x(g,h)$ is the slant product introduced in Eq.~\eqref{eq:thetadef}. Therefore, in this simplest case, we have
\beq
    Amp^{\mathcal{R}}=\theta_{g_3xg_3^{-1}}(g_3g_2^{-1},g_2g_1^{-1})\prod_{\text{tetra}\in\partial\mathcal{R}} \omega(\text{tetra})^s.
\eeq

For simplicity and without loss of generality, we assume $\partial\mathcal{R}$ with branching structure $ 1\rightarrow 2\rightarrow 3\rightarrow \cdots \rightarrow n \leftarrow 1$, then by using cocycle conditions, we can show that the tetrahedrons inside a union of prisms (see Fig.~\ref{fig:closed defect appendix}), which is formed by lifted plaquettes, will give rise to
\beq
    \Tilde{\Theta}^{g_n x g_n^{-1}}_{\partial\mathcal{R}}=\theta_{g_n x g_n^{-1}}(g_n g_{n-1}^{-1},g_{n-1}g_1^{-1})\cdots\theta_{g_3 x g_3^{-1}}(g_3 g_{2}^{-1},g_{2}g_1^{-1}).
\eeq
Therefore,
\beq
    Amp^{\mathcal{R}}=\Tilde{\Theta}^{g_n x g_n^{-1}}_{\partial\mathcal{R}}\prod_{\text{tetra}\in\partial\mathcal{R}} \omega(\text{tetra})^s.
\eeq

The shift operator $\prod_{{e'\cap\partial\mathcal{R}}}L^x_{\pm e'}$ in Eq.~\eqref{eq:global gauge} is exactly the operator $L^x_{\partial\mathcal{R}}$ defined in Eq.~\eqref{eq:Lx1}. Therefore we arrive at
\beq
    \tilde{\mathcal{B}}^x_{\partial\mathcal{R}}\ket{\Psi_{\text{SPT-pre}}}=G^x_{\mathcal{R}}\ket{\Psi_{\text{SPT-pre}}},
\eeq
when
\begin{align}
    \tilde{\mathcal{B}}^x_{\partial\mathcal{R}}=\sum_{g_v}\mathcal{L}^{x}_{\partial \mathcal{R}} \Tilde{W}^{g_v x g_v^{-1}}_{\partial\mathcal{R}}\Tilde{\Theta}^{g_v x g_v^{-1}}_{\partial\mathcal{R}}\ket{g_v}_v\bra{g_v}.
\end{align}

For a state with nontrivial configuration $\{h_e\}$, there could be fluxes on some plaquette, i.e., for some plaquette $p$, $\prod_{e\in\partial p}h_e\neq1$. For the purpose of illustration, we now focus on states that have fluxes only on 1 plaquette, and only violates terms that are on this plaquette in $H_{\text{SPT-pre}}$ in Eq.~\eqref{eq:hamiltonian SPT-pre}. To write down the branch line operators for these states, we repeat a similar procedure: we impose some gauge transformation on region $\mathcal{R}$, the effect of which would be only on its boundary. It turns out that when all the plaquettes on $\partial\mathcal{R}$ are fluxless (i.e., $\prod_{e\in\partial p}h_e=1$), and the flux on each plaquette $p\in\mathcal{R}$ ($\prod_{e\in \partial p}h_e$) is in the centralizer group $\mathcal{Z}_x$ (since the flux on a plaquette is ambiguous up to a conjugation when choosing a different starting point, we assume that the whole conjugacy class of flux $[\prod_{e\in\partial p}h_e]\subset \mathcal{Z}_x$), we can indeed find such a gauge transformation as we now explain. We first choose a reference vertex $v$ in $\mathcal{R}$, then for any vertex $v'$, we can find a path $l$ that flows from $v$ to $v'$ and define $h_{v'v}=\prod_{e\in l}h_e$. Given that all the fluxes are assumed to be in $\mathcal{Z}_x$, i.e., they commute with $x$, therefore, different choices of path $l$ will give rise to the same gauge transformation (taking $h_{vv}=1$),
\beq
    G^{x}_{v,\mathcal{R}}\equiv\prod_{v'\in\mathcal{R}}G_{v'}^{h_{v'v}xh_{v'v}^{-1}},
\eeq
which will leave all edge DOFs invariant except for the ones on $\partial\mathcal{R}$, and the shift on those edges is exactly the operator $L^x_{\partial\mathcal{R}}$ defined in Eq.~\eqref{eq:Lx}. 
Now since there are plaquettes in $\mathcal{R}$ that have nontrivial fluxes, we cannot write the phase part of the gauge transformation as
\beq  
    Amp^{\mathcal{R}}=\prod_{\text{tetra}\in\mathcal{R}} \omega(\text{tetra})^s.
\eeq
However, since we assume the plaquettes on $\partial\mathcal{R}$ are all fluxless, we still have a well defined factor $\prod_{\text{tetra}\in\partial\mathcal{R}} \omega(\text{tetra})^s$, the holonomy along $\partial\mathcal{R}$ is $h\equiv\prod_{e\in\partial\mathcal{R}}h_e\neq 1$. We conjecture the branch line operator to be
\beq
    \mathcal{B}^{x}_{\partial\mathcal{R}}&=\sum_{g}\mathcal{B}^{x,g}_{\partial\mathcal{R}}\epsilon_{g_v x g_v^{-1}}(g),
\eeq
with  
\beq
\mathcal{B}^{x,g}_{\partial\mathcal{R}}&\equiv\sum_{g_v}\mathcal{L}^{x}_{\partial\mathcal{R}} W^{g_v x g_v^{-1}}_{\partial\mathcal{R}}\Theta^{g_v x g_v^{-1}}_{\partial\mathcal{R}}\delta_{g,g_v (\prod_e h_e )g_v^{-1}}\ket{g_v}_v\bra{g_v}.
\eeq

Again we suppose that $\partial\mathcal{R}$ has the branching structure $ 1\rightarrow 2\rightarrow 3\rightarrow \cdots \rightarrow n \leftarrow 1$. Then, we have
\beq
    \label{eq:Theta}
    \Theta^{g_n x g_n^{-1}}_{\partial\mathcal{R}}=\theta_{g_n x g_n^{-1}}(g_n h_{n,n-1} g_{n-1}^{-1},g_{n-1}h_{n,n-1}^{-1}hg_1^{-1})\cdots\theta_{g_3 x g_3^{-1}}(g_3 h_{3,2} g_{2}^{-1},g_{2}h_{2,1}g_1^{-1})\theta^{-1}_{g_n x g_n^{-1}}(g_n h g_n^{-1},g_n h_{n,1}g_1^{-1}),
\eeq
where $h\equiv h_{n,n-1}\cdots h_{2,1} h_{n,1}^{-1}$ is the holonomy along $\partial\mathcal{R}$. We can calculate the multiplication rule for $\mathcal{B}^x_{\partial\mathcal{R}}$. First off, notice that operator $\mathcal{L}^{x}_{\partial\mathcal{R}} W^{g_n x g_n^{-1}}_{\partial\mathcal{R}}$ and $\Theta^{g_nyg_n^{-1}}_{\partial\mathcal{R}}$ commute for any $x,y\in G$.
Using the cocycle condition, one can show that
\beq
    \mathcal{L}^{x}_{\partial\mathcal{R}} W^{g_n x g_n^{-1}}_{\partial\mathcal{R}}\mathcal{L}^{y}_{\partial\mathcal{R}} W^{g_n y g_n^{-1}}_{\partial\mathcal{R}}=\mathcal{L}^{x}_{\partial\mathcal{R}}W^{g_n xy g_n^{-1}}_{\partial\mathcal{R}}\gamma_{g_n h_{n,n-1}g_{n-1}^{-1}}(g_n x g_n^{-1},g_n y g_n^{-1}) \cdot\text{phase},
    \label{eq:lwlw}
\eeq
where the `phase' in the above equation corresponds to Fig.~\ref{fig:lwlw}. 
Using the correspondence between tetrahedrons and cocycles, one can write

\beq
    \mathcal{L}^{x}_{\partial\mathcal{R}} W^{g_n x g_n^{-1}}_{\partial\mathcal{R}}\mathcal{L}^{y}_{\partial\mathcal{R}} W^{g_n y g_n^{-1}}_{\partial\mathcal{R}}=&\mathcal{L}^{xy}_{\partial\mathcal{R}} W^{g_n xy g_n^{-1}}_{\partial\mathcal{R}}\gamma_{g_n h_{n,n-1}g_{n-1}^{-1}}(g_n x g_n^{-1},g_n y g_n^{-1})\cdots \\
    &\gamma_{g_2 h_{2,1}g_{1}^{-1}}(g_2 x g_2^{-1},g_2 y g_2^{-1})\gamma^{-1}_{g_n h_{n,1}g_{1}^{-1}}(g_n x g_n^{-1},g_n y g_n^{-1}).
\eeq
\begin{figure*}[t]
    \centering
    \includegraphics[scale=0.4]{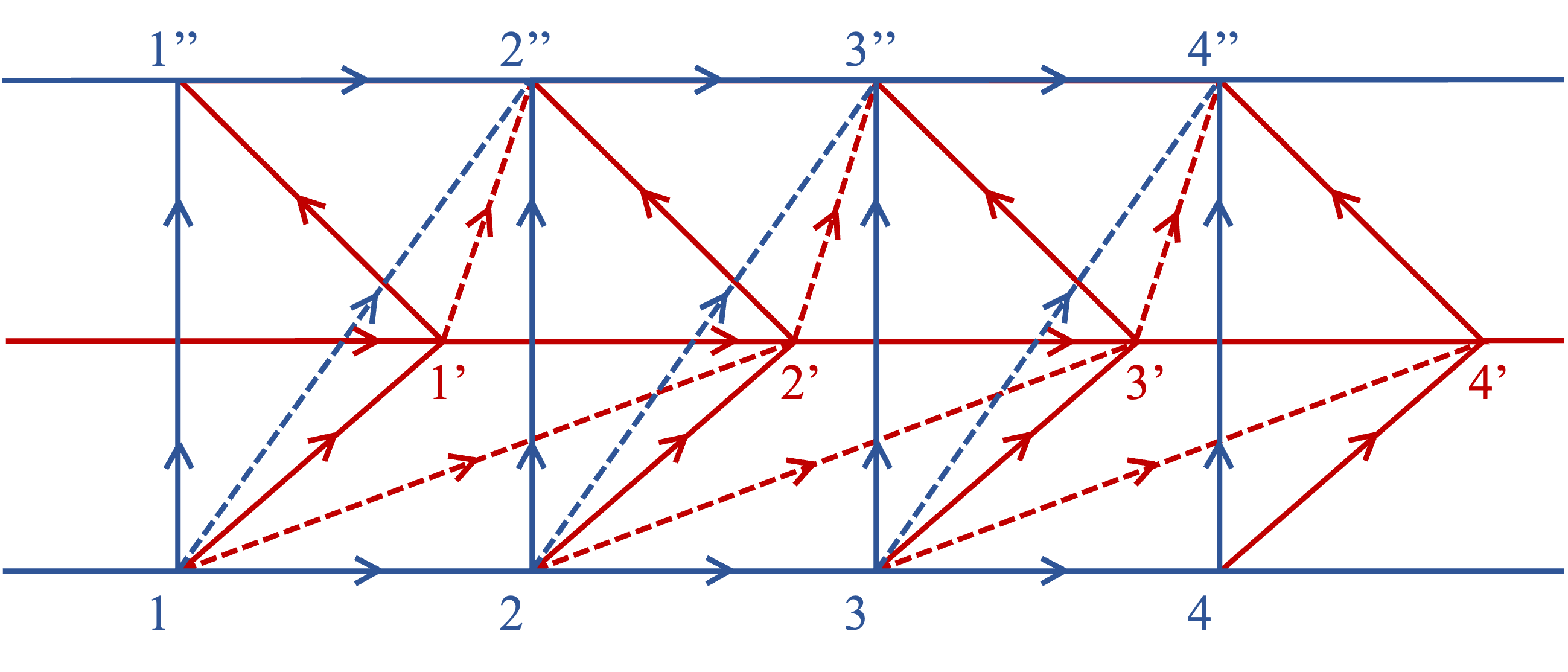}
    \caption{Geometric diagram illustrating the phase combination in Equation~\ref{eq:lwlw}. Here $g_{v'} g^{-1}_{v}= y$ and $g_{v''} g^{-1}_{v'}= x$ for $v \in V$ (which are enumerated by 1,2,3,4).}
    \label{fig:lwlw}
\end{figure*}
By using the definition, we can compute the following product of the two phases,
\beq
    \Theta^{g_n x g_n^{-1}}_{\partial\mathcal{R}}\Theta^{g_n y g_n^{-1}}_{\partial\mathcal{R}}=&\theta_{g_n x g_n^{-1}}(g_n h_{n,n-1} g_{n-1}^{-1},g_{n-1}h_{n,n-1}^{-1}hg_1^{-1})\theta_{g_n y g_n^{-1}}(g_n h_{n,n-1} g_{n-1}^{-1},g_{n-1}h_{n,n-1}^{-1}hg_1^{-1})\cdots\\
    &\theta_{g_3 x g_3^{-1}}(g_3 h_{3,2} g_{2}^{-1},g_{2}h_{2,1}g_1^{-1})\theta_{g_3 y g_3^{-1}}(g_3 h_{3,2} g_{2}^{-1},g_{2}h_{2,1}g_1^{-1})\\
    &\theta^{-1}_{g_n x g_n^{-1}}(g_n h g_n^{-1},g_n h_{n,1}g_1^{-1})\theta^{-1}_{g_n y g_n^{-1}}(g_n h g_n^{-1},g_n h_{n,1}g_1^{-1}).
\eeq
Using the identity,
\beq
    \theta_g(x,y)\theta_h(x,y)\gamma_x(g,h)\gamma_y(x^{-1}gx,x^{-1}hx)=\theta_{gh}(x,y)\gamma_{xy}(g,h),
\eeq
one can derive that
\beq
    \mathcal{L}^{x}_{\partial\mathcal{R}} W^{g_n x g_n^{-1}}_{\partial\mathcal{R}}\Theta^{g_n x g_n^{-1}}_{\partial\mathcal{R}}\mathcal{L}^{y}_{\partial\mathcal{R}} W^{g_n y g_n^{-1}}_{\partial\mathcal{R}}\Theta^{g_n y g_n^{-1}}_{\partial\mathcal{R}}=\mathcal{L}^{xy}_{\partial\mathcal{R}} W^{g_n xy g_n^{-1}}_{\partial\mathcal{R}}\Theta^{g_n xy g_n^{-1}}_{\partial\mathcal{R}}\gamma_{g_n h g_n^{-1}}(g_n x g_n^{-1},g_n y g_n^{-1}).
\eeq
When the pre-gauge structure is trivial (i.e. $h_e\equiv1$), the phases $W^{g_n x g_n^{-1}}_{\partial\mathcal{R}}$ and $\Theta^{g_v x g_v^{-1}}_{\partial\mathcal{R}}$ are reduced to $\Tilde{W}^{g_n x g_n^{-1}}_{\partial\mathcal{R}}$ and $\Tilde{\Theta}^{g_v x g_v^{-1}}_{\partial\mathcal{R}}$, resepctively, and also the holonomy is trivial, $h\equiv1$. 
Therefore, we arrive at the multiplication rule,
\beq
    \tilde{\mathcal{B}}^x_{\partial\mathcal{R}}\tilde{\mathcal{B}}^y_{\partial\mathcal{R}}=\tilde{\mathcal{B}}^{xy}_{\partial\mathcal{R}}.
\eeq
For a more general pre-gauge structure, where the fluxes are in $\mathcal{Z}_x$, one has
\begin{align}
    \mathcal{B}^{x,g}_{\partial\mathcal{R}}\,\mathcal{B}^{y,g'}_{\partial\mathcal{R}}=\mathcal{B}^{xy,g}_{\partial\mathcal{R}}\,\gamma_g(g_v x g_v^{-1},g_v y g_v^{-1})\,\delta_{g,g'}.
\end{align}

 \section{Finite-depth local unitary to map an SET state to a TQD state}
 \label{sec:local unitary}

For an SPT state, we define an operator
\beq
    \Hat{O}=U_{\omega}\Big(\prod_v \mathcal{V}^{q}_{v}\Big)U_{\omega}^{\dagger},
\eeq
where $U_{\omega}$ is the operator that brings a direct product state to a $G$-SPT state,
\beq
    U_{\omega}=\sum_{\{g_v\}}\prod_{\Delta}\omega(g_3g_2^{-1},g_2g_1^{-1},g_1)^{s(\Delta)}\bigotimes_{v} \ket{g_v}\bra{g_v},
\eeq
and $\mathcal{V}^q_v$ is the Fourier transform of the quotient part. Namely, suppose $|Q|=m$, we label the embedding of $Q$ in $G$ as $s(Q)=\{q_0,\cdots,q_{m-1}\}$, where $q_0=1$.
Then we write
\beq
    \mathcal{V}^q_v=\frac{1}{\sqrt{m}}\sum_{k,l=0}^{m-1} e^{\frac{2\pi i kl}{m}}L^{q_k q_l^{-1}}_{+v}\delta(q(g_v),q_l) ,
\eeq
whose hermitian conjugate is calculated to be
\beq 
(\mathcal{V}^q_v)^{\dagger}=\frac{1}{\sqrt{m}}\sum_{k,l=0}^{m-1} e^{- \frac{2\pi i kl}{m}}L^{q_k q_l^{-1}}_{+v}\delta(q(g_v),q_l) .
\eeq
We note that this Fourier transform is indeed unitary: $(\mathcal{V}^q_v)^{\dagger} \mathcal{V}^q_v = 1$. Applying $\Hat{O}$ on the $G$-SPT state brings a $G$-SPT state to an $N$-SPT state:
\beq
    \Hat{O}\ket{\Psi_{G\text{-SPT}}}&=U_{\omega}\Big(\prod_v \mathcal{V}^{q}_{v}\Big)U_{\omega}^{\dagger}\ket{\Psi_{G\text{-SPT}}}\\
    &=U_{\omega}\Big(\prod_v \mathcal{V}^{q}_{v}\Big)\bigotimes_v\Big(\sum_{g\in G}\ket{g}_v\Big)\\
    &=U_{\omega}\bigotimes_v\Big(\sum_{n\in N}\ket{n}_v\Big)\\
    &=\sum_{\{n_v\}}\prod_{\Delta}\omega(n_3n_2^{-1},n_2n_1^{-1},n_1)^{s(\Delta)}\bigotimes_{v} \ket{n_v}\\
    &=\sum_{\{n_v\}}\prod_{\Delta}\nu(n_3n_2^{-1},n_2n_1^{-1},n_1)^{s(\Delta)}\bigotimes_{v} \ket{n_v}\\
    &=\ket{\Psi_{N\text{-SPT}}},
\eeq
where $\nu$ is the restriction of $\omega$ on $N$ and naturally a cocycle in $H^3(N,U(1))$. The operator $\mathcal{V}_v^q$ on different vertices commute,
\beq
    \mathcal{V}_v^q \mathcal{V}_{v'}^q - \mathcal{V}_{v'}^q \mathcal{V}_v^q =0.
\eeq
Therefore we can write an operator on vertex $v$
\beq
    \Hat{O}_v\equiv U_{\omega}\mathcal{V}^{q}_{v} U_{\omega}^{\dagger} =\frac{1}{\sqrt{m}}\sum_{k,l=0}^{m-1} e^{\frac{2\pi i kl}{m}}L^{q_k q_l^{-1}}_{+v}W^{q_k q_l^{-1}}_v\delta(q(g_v),q_l),
\eeq
such that $\Hat{O}_v$ on different vertices commute, and their product over all the vertices is exactly $\Hat{O}$,
\beq
    \Hat{O}\equiv\prod_v\Hat{O}_v.
\eeq 
After gauging the normal subgroup $N$, the operator $\Hat{O}$ is mapped to
\beq
    \Hat{O}_{\text{SET}}=\prod_v \Hat{O}_{v,\text{SET}},
\eeq
where
\beq
    \Hat{O}_{v,\text{SET}}=\frac{1}{\sqrt{m}}\sum_{k,l=0}^{m-1} e^{\frac{2\pi i kl}{m}}W^{q_k q_l^{-1}}_v\ket{q_k}_v\bra{q_l}.
\eeq
The operators $\Hat{O}_{v,\text{SET}}$ on different vertices commute. 
Also note that since $\Hat{O}_{v,\text{SET}}$ is supported on $v$, adjacent vertices, and edges, and thus it can be implemented locally in the state after gauging. 

Note that although all $\Hat{O}_{v,\text{SET}}$ commute, the operators $W^{q_k q_l^{-1}}_{v}$ generally depends on $q_{v'}$ and $n_e$ configuration on adjacent vertices and edges. 
To implement the transformation, we specify an ordering of $\Hat{O}_{v,\text{SET}}$ in $\Hat{O}_{\text{SET}}$ which can be implemented in finite depth as follows. 
We divide the spatial lattice into sublattices such that, the vertices in each sublattice are not adjacent in the spatial lattice. Then the operators $L^{q_{k}q_l^{-1}}_{+v}$ and $W^{q_k q_l^{-1}}_{v'}$ always commute when $v$ and $v'$ are different vertices in one sublattice. As a result, we can implement $\Hat{O}_{v,\text{SET}}$ within one sublattice simultaneously, and implement one sublattice in each step. As long as there is only a finite number of such sublattices in the spatial lattice, the circuit we described above is a finite-depth local unitary,
\beq
    \Hat{O}_{\text{SET}}=\prod_{\text{sublattices}}\prod_{v\in \text{sublattice}}\Hat{O}_{v,\text{SET}}.
\eeq

Applying $\Hat{O}_{\text{SET}}$ on an SET ground state will take all the vertex DOFs to $\ket{q_v}=\ket{1}$, i.e. disentangle all the vertex DOFs. Therefore, under the action of this operator, we get a TQD ground state with gauge group $N$. We emphasize that here we aim to probe the underlying anyons, and the circuit does not respect the global $Q$-symmetry, which is why we can take an SET state to a pure TQD state.

\section{Fusion rule in SET via gauging $Z_4$ group from $D_4$ SPT}
\label{sec:fusionruled4 appendix}
We write the element in $D_4$ as $\Tilde{g}=(G,g)\equiv x^Ga^g$. We construct a representative of 3-cocycle in $H^3(D_4,U(1))$ as follows, 
\beq
    \omega(\Tilde{g},\Tilde{h},\Tilde{l})  =\exp{\frac{2\pi i p_1}{16}g(-1)^{H+L}(h(-1)^L+l-[h(-1)^L+l]_4)+\pi i p_2 GHL+ \pi i p_3 gHL},
\eeq
where $p_1=0,1,2,3$, and $p_2,p_3=0$ or $1$. Gauging the normal subgroup $Z_4$ of a $D_4$-SPT results in a state in an SET that has the same anyon set as $D^{\nu}(Z_4)$, where
\beq
    \nu(g,h,l)=\exp{\frac{2\pi i p_1}{16}g(h+l-[h+l]_4)}
\eeq
is the restriction of $\omega$ on $Z_4$. 
Different values of $p_1$ correspond to different $Z_4$-TQD models.
The symmetry action is nontrivial, and it takes an anyon to its inverse.
In the sector $\mathcal{C}_x$, there are 4 objects of quantum dimension $2$. 
If we pick one of them and name it as $0_x$, by dimension counting, we can write a fusion rule of the form,
\beq
    0_x \times 0_x = a + b + c + d,
\eeq
where $a,b,c,d\in \mathcal{C}$ are abelian anyons. 
Let $b_1=1$ and $b_2=a$.
Then one can write a matrix-valued operator on an open ribbon as
\beq
    (H^x_l)_{ii'}=\sum_{n\in\{1,a^2\}}H_l^{b_i x b_i^{-1},b_i n b_{i'}^{-1}}\epsilon_{b_i x b_i^{-1}}(b_i n b_{i'}^{-1}),
\eeq
where the matrix indices $i,i'=1,2$, and the operator $H_l^{x,g}$ satisfies the same multiplication rule as in Eq.~\eqref{eq:Hoperatormultiply},
\beq
    \label{eq:H-open-multiplication}
H_l^{x,g}H_l^{y,g'}=H_l^{xy,g}\gamma_g(x,y)\delta_{g,g'}.
\eeq

We conjecture that the operator $H^x_l$ creates an object in the sector $\mathcal{C}_x$ on the end point of $l$, and we name it $0_x$. Then the object $0_x\times 0_x$ should be created on the endpoint of $l$ by operator $(H^x_l)^{\otimes2}$. Let $\xi\equiv \exp{\frac{2\pi i }{16}}$. 
We calculate the tensor product of the two open ribbon operators, and by diagonalizing it, we find an expression as follows:
\beq
    (H^x_l)^{\otimes2}=&\begin{pmatrix}
    H_l^{x,1}+\epsilon_{x}(a^2)H_l^{x,a^2} & \epsilon_{x}(a)H_l^{x,a}+\epsilon_{x}(a^3)H_l^{x,a^3}\\
    \epsilon_{xa^2}(a)H_l^{xa^2,a}+\epsilon_{xa^2}(a^3)H_l^{xa^2,a^3} & H_l^{xa^2,1}+\epsilon_{xa^2}(a^2)H_l^{xa^2,a^2}
    \end{pmatrix}^{\otimes2}\\
    =&\begin{pmatrix}
    H_l^{x,1}+\xi^{-4p_1}H_l^{x,a^2} & \xi^{p_1}H_l^{x,a}+\xi^{-3p_1}H_l^{x,a^3}\\
    \xi^{-p_1}H_l^{xa^2,a}+\xi^{-9p_1}H_l^{xa^2,a^3} & H_l^{xa^2,1}+H_l^{xa^2,a^2}
    \end{pmatrix}^{\otimes2}\\
    =&\begin{pmatrix}
        H_l^{1,1}+\xi^{-8p_1}H_l^{1,a^2} & \xi^{2p_1}(-1)^{p_3}H^{1,a}+\xi^{-6p_1}(-1)^{p_3}H^{1,a^3}\\
        \xi^{6p_1}(-1)^{p_3}H^{1,a}+\xi^{-2p_1}(-1)^{p_3}H^{1,a^3} & H_l^{1,1}+\xi^{-8p_1}H_l^{1,a^2}
    \end{pmatrix}\\
    &\oplus \begin{pmatrix}
        H_l^{a^2,1}+\xi^{-4p_1}H_l^{a^2,a^2} & \xi^{-4p_1}(-1)^{p_3}H_l^{a^2,a}+\xi^{8p_1}(-1)^{p_3}H_l^{a^2,a^3} \\
        (-1)^{p_3}H_l^{a^2,a}+\xi^{-4p_1}(-1)^{p_3}H_l^{a^2,a^3} & H_l^{a^2,1}+\xi^{-4p_1}H_l^{a^2,a^2}
    \end{pmatrix} \\
    =&\begin{pmatrix}
        H_l^{1,1}+\xi^{-8p_1}H_l^{1,a^2} & \xi^{2p_1}(-1)^{p_3}H^{1,a}+\xi^{-6p_1}(-1)^{p_3}H^{1,a^3}\\
        \xi^{6p_1}(-1)^{p_3}H^{1,a}+\xi^{-2p_1}(-1)^{p_3}H^{1,a^3} & H_l^{1,1}+\xi^{-8p_1}H_l^{1,a^2}
    \end{pmatrix}\\
    &\oplus \begin{pmatrix}
        H_l^{a^2,1}+\xi^{-8p_1}\epsilon_{a^2}(a^2)H_l^{a^2,a^2} & \xi^{-6p_1}(-1)^{p_3}\epsilon_{a^2}(a)H_l^{a^2,a}+\xi^{2p_1}(-1)^{p_3}\epsilon_{a^2}(a^3)H_l^{a^2,a^3} \\
        \xi^{-2p_1}(-1)^{p_3}\epsilon_{a^2}(a)H_l^{a^2,a}+\xi^{6p_1}(-1)^{p_3}\epsilon_{a^2}(a^3)H_l^{a^2,a^3} & H_l^{a^2,1}+\xi^{-8p_1}\epsilon_{a^2}(a^2)H_l^{a^2,a^2}
    \end{pmatrix} \\
    \simeq& (H_l^{1,1}+\xi^{4p_1}(-1)^{p_3}H_l^{1,a}+\xi^{8p_1}H_l^{1,a^2}+\xi^{12p_1}(-1)^{p_3}H_l^{1,a^3})\\
    &\oplus (H_l^{1,1}-\xi^{4p_1}(-1)^{p_3}H_l^{1,a}+\xi^{8p_1}H_l^{1,a^2}-\xi^{12p_1}(-1)^{p_3}H_l^{1,a^3})\\
    &\oplus (H_l^{a^2,1}+\xi^{4p_1}(-1)^{p_3}\epsilon_{a^2}(a)H_l^{a^2,a}+\xi^{8p_1}\epsilon_{a^2}(a^2)H_l^{a^2,a^2}+\xi^{12p_1}(-1)^{p_3}\epsilon_{a^2}(a^3)H_l^{a^2,a^3})\\
    &\oplus (H_l^{a^2,1}-\xi^{4p_1}(-1)^{p_3}\epsilon_{a^2}(a)H_l^{a^2,a}+\xi^{8p_1}\epsilon_{a^2}(a^2)H_l^{a^2,a^2}-\xi^{12p_1}(-1)^{p_3}\epsilon_{a^2}(a^3)H_l^{a^2,a^3}).
    \label{eq:D4decomposition}
\eeq
In the third equality, we used the multiplication rule Eq.~\eqref{eq:H-open-multiplication}, and the $p_3$ dependence arose from $\gamma_g(x,y)$.

Let us first consider the case with $p_3=0$. When $p_1=0,2$, the four parts in the last line --- each part is a sum of four $H^{\bullet,\bullet}_l$ operators --- are exactly the four ribbon operators in the TQD $D^{\nu}(Z_4)$ that create anyons, $1, e^2, m^2,$ and $e^2m^2$, respectively. 
Therefore, we obtain the fusion rule,
\beq
    0_x\times 0_x=1 + e^2 + m^2 + e^2m^2.
    \label{eq:d4fusion1}
\eeq
When $p_1=1,3$, the four parts in the last line after the decomposition are four ribbon operators in the TQD $D^{\nu}(Z_4)$ that create anyons, $e, e^3, em^2$, and $e^3m^2$. 
Therefore, we obtain the fusion rule,
\beq
    0_x\times 0_x=e + e^3 + em^2 + e^3m^2.
    \label{eq:d4fusion2}
\eeq

According to Eq.~\eqref{eq:D4decomposition}, when $p_3$ takes value $1$ instead of $0$, the ribbon operators creating anyon $e/e^3/em^2/e^3m^2$ after the decomposition become ribbon operators creating respectively $e^3/e/e^3m^2/em^2$ instead. According to Eq.~\eqref{eq:SFCgeneral}, we conclude that different values of $p_3$ indeed give rise to different SFCs such that the fusion rules are shifted by the anyon $[e^2]\in H^2_{\rho}(Z_2,\mathcal{A})$. 
We note that after the shift by $e^2$, the fusion rules in Eq.~\eqref{eq:d4fusion1} and Eq.~\eqref{eq:d4fusion2} are actually invariant. 
However, it does not mean that we are in the same SET order: further analysis, such as the $F$-symbol, the $S$-matrix of the category $\mathcal{C}_{Z_2}^{\times}$~\cite{Barkeshli_2019}, is necessary. 
Indeed in this case, it is ensured that SET orders corresponding to different $p_3$ values are distinct, since further gauging the $Z_2$ symmetry it should result in different TQD orders $D^{\omega}(D_4)$.

\section{Fusion rule in SET via gauging $Z_2\times Z_2$ group from $D_4$ SPT}
\label{sec:fusionruled4 appendix z2z2}
Again, we start from the representative 3-cocycle in $H^3(D_4,U(1))$ as follows: 
\beq
    \omega(\Tilde{g},\Tilde{h},\Tilde{l})  =\exp{\frac{2\pi i p_1}{16}g(-1)^{H+L}(h(-1)^L+l-[h(-1)^L+l]_4)+\pi i p_2 GHL+ \pi i p_3 gHL},
\eeq
where $p_1=0,1,2,3$, and $p_2,p_3=0$ or $1$. 
Gauging the subgroup $Z_2\times Z_2=\{1,x,t,xt\}$ in $D_4$ results in a state within $D^{\nu'}(Z_2\times Z_2)$. 
Let us write $t\equiv a^2$ and $g=x^{g^{(1)}}t^{g^{(2)}}=x^{g^{(1)}}a^{2g^{(2)}}$. The 3-cocycle for this group is
\beq    
    \nu'(g,h,l)=&\exp\Big\{\frac{2\pi i p_1}{4}g^{(2)}(-1)^{h^{(1)}+l^{(1)}}\big(h^{(2)}(-1)^{l^{(1)}}+l^{(2)}\\
    &-[h^{(2)}(-1)^{l^{(1)}}+l^{(2)}]_2\big)+\pi i p_2 g^{(1)}h^{(1)}l^{(1)}\Big\},\\
    =&(-1)^{p_1 (g^{(2)}h^{(2)}l^{(2)}+g^{(2)}h^{(2)}l^{(1)})+p_2 g^{(1)}h^{(1)}l^{(1)}}.
\eeq

We conjecture that the operator $H^a_l$ creates an object in sector $\mathcal{C}_a$ on the end point of $l$ and we name it $0_a$. Then the object $0_a\times 0_a$ should be created on the endpoint of $l$ by operator $(H^a_l)^{\otimes2}$. Let $\xi\equiv \exp{\frac{2\pi i }{16}}$, then a calculation similar to that in Sec.~\ref{sec:fusionruled4 appendix} gives us
\beq
    (H^a_l)^{\otimes2}=&\begin{pmatrix}
    H_l^{a,1}+\epsilon_{a}(a^2)H_l^{a,a^2} & \epsilon_{a}(x)H_l^{a,x}+\epsilon_{a}(xa^2)H_l^{a,xa^2}\\
    \epsilon_{a^3}(x)H_l^{a^3,x}+\epsilon_{a^3}(xa^2)H_l^{a^3,xa^2} & H_l^{a^3,1}+\epsilon_{a^3}(a^2)H_l^{a^3,a^2}
    \end{pmatrix}^{\otimes2}\\
    =&\begin{pmatrix}
    H_l^{a,1}+\xi^{2p_1}H_l^{a,a^2} & \xi^{4p_3}H_l^{a,x}+\xi^{6p_1+4p_3}H_l^{a,xa^2}\\
    \xi^{4p_3}H_l^{a^3,x}+\xi^{2p_1+4p_3}H_l^{a^3,xa^2} & H_l^{a^3,1}+\xi^{6p_1}H_l^{a^3,a^2}
    \end{pmatrix}^{\otimes2}\\
    =&\begin{pmatrix}
        H_l^{a^2,1}+\xi^{4p_1}H_l^{a^2,a^2} & \xi^{4p_1+8p_3}H^{a^2,x}+\xi^{8p_1+8p_3}H^{a^2,xa^2}\\
        \xi^{-4p_1+8p_3}H^{a^2,x}+\xi^{8p_3}H^{a^2,xa^2} & H_l^{a^2,1}+\xi^{4p_1}H_l^{a^2,a^2}
    \end{pmatrix}\\
    &\oplus \begin{pmatrix}
        H_l^{1,1}+H_l^{1,a^2} & \xi^{4p_1+8p_3}H_l^{1,x}+\xi^{4p_1+8p_3}H_l^{1,xa^2} \\
        \xi^{-4p_1+8p_3}H_l^{1,x}+\xi^{-4p_1+8p_3}H_l^{1,xa^2} & H_l^{1,1}+H_l^{1,a^2}
    \end{pmatrix} \\
    =&\begin{pmatrix}
        H_l^{a^2,1}+\epsilon_{a^2}(a^2)H_l^{a^2,a^2} & \xi^{4p_1+8p_3}\epsilon_{a^2}(x)H^{a^2,x}+\xi^{4p_1+8p_3}\epsilon_{a^2}(xa^2)H^{a^2,xa^2}\\
        \xi^{-4p_1+8p_3}\epsilon_{a^2}(x)H^{a^2,x}+\xi^{-4p_1+8p_3}\epsilon_{a^2}(xa^2)H^{a^2,xa^2} & H_l^{a^2,1}+\epsilon_{a^2}(a^2)H_l^{a^2,a^2}
    \end{pmatrix}\\
    &\oplus \begin{pmatrix}
        H_l^{1,1}+H_l^{1,a^2} & \xi^{4p_1+8p_3}H_l^{1,x}+\xi^{4p_1+8p_3}H_l^{1,xa^2} \\
        \xi^{-4p_1+8p_3}H_l^{1,x}+\xi^{-4p_1+8p_3}H_l^{1,xa^2} & H_l^{1,1}+H_l^{1,a^2}
    \end{pmatrix} \\
    \simeq&(H_l^{a^2,1}+(-1)^{p_3}\epsilon_{a^2}(x)H_l^{a^2,x}+\epsilon_{a^2}(a^2)H_l^{a^2,a^2}+(-1)^{p_3}\epsilon_{a^2}(xa^2)H_l^{a^2,xa^2})\\
    &\oplus (H_l^{a^2,1}-(-1)^{p_3}\epsilon_{a^2}(x)H_l^{a^2,x}+\epsilon_{a^2}(a^2)H_l^{a^2,a^2}-(-1)^{p_3}\epsilon_{a^2}(xa^2)H_l^{a^2,xa^2})\\
    &\oplus (H_l^{1,1}+(-1)^{p_3}H_l^{1,x}+H_l^{1,a^2}+(-1)^{p_3}H_l^{1,xa^2})\oplus
    (H_l^{1,1}-(-1)^{p_3}H_l^{1,x}+H_l^{1,a^2}-(-1)^{p_3}H_l^{1,xa^2}).
    \label{eq:D4decomposition_Z2Z2}
\eeq

For different values of $p_1$ and $p_2$, the anyon theory would be different after gauging. In Appendix~\ref{braidingphase}, we give a complete classification of the anyons in all cases. When $p_3=0$ here, one can then find the fusion rule from the above calculation as
\beq
    0_a\times 0_a=1 + e^{(1)} + m^{(2)} + e^{(1)}m^{(2)}.
\eeq

When $p_3$ takes value $1$ instead of $0$, the ribbon operators that create anyon $1/e^{(1)}/ m^{(2)}/e^{(1)}m^{(2)}$ after the decomposition become ribbon operators that create $e^{(1)}/1/e^{(1)}m^{(2)}/ m^{(2)}$, respectively. According to Eq.~\eqref{eq:SFCgeneral}, we can conclude that different values of $p_3$ indeed give rise to different SFCs such that the fusion rules are shifted by the anyon $[e^{(1)}]\in H^2_{\rho}(Z_2,\mathcal{A})$.

\section{Computation of braiding phases of anyons in TQD}\label{braidingphase}

In this appendix, we compute the braiding phases of anyons in $Z_2\times Z_2$ TQD obtained from gauging $Z_2\times Z_2\subset D_4$. First we define the slant product
\begin{align}
    H^n(G,U(1))\xrightarrow[]{i_g^n}H^{n-1}(G,U(1))
\end{align}
as follows:
\begin{align}
    i_g^n\omega(g_1,...,g_{n-1})=\omega(g,g_1,...,g_{n-1})^{(-1)^{n-1}}\prod_{i=1}^{n-1}
    \omega(   g_1,...,g_i,g,g_{i+1},...,g_{n-1} )^{(-1)^{n-1+i}}. 
\end{align}
Now let us consider the cocycle given in Eq.~\eqref{z2z2nu'}:
\begin{align}
    \nu'(g,h,l)=(-1)^{p_1(g^{(2)}h^{(2)}l^{(2)}+g^{(2)}h^{(2)}l^{(1)})+p_2g^{(1)}h^{(1)}l^{(1)}}.
\end{align}
One can calculate the slant product with $g=x,t$ and $xt$:
\begin{align}
    i_{x}\nu'(h,l)&=(-1)^{p_1h^{(2)}l^{(2)}+p_2h^{(1)}l^{(1)}},\\
    i_{t}\nu'(h,l)&=(-1)^{p_1h^{(2)}l^{(2)}},\\
    i_{xt}\nu'(h,l)&=(-1)^{p_2h^{(1)}l^{(1)}}.
    \label{eq:slantz2z2result}
\end{align}
These slant products give the projective phases in the projective representations $\mu_x$, $\mu_t$, and $\mu_{xt}$, respectively, as follows:
\begin{subequations}
    \begin{align}
    \mu_x(h)\mu_x(l)&=i_x\nu'(h,l)\mu_x(hl),\\
    \mu_t(h)\mu_t(l)&=i_t\nu'(h,l)\mu_t(hl),\\
    \mu_{xt}(h)\mu_{xt}(l)&=i_{xt}\nu'(h,l)\mu_{xt}(hl).
\end{align}
\end{subequations}
From Eq.~\eqref{eq:slantz2z2result}, we see that the projective representations are given respectively  by
\begin{subequations}
    \begin{align}
    \mu_x(h)=i^{p_1h^{(2)}+p_2h^{(1)}},\qquad
    \mu_t(h)=i^{p_1h^{(2)}},\qquad
    \mu_{xt}(h)=i^{p_2h^{(1)}}.
\end{align}
\end{subequations}
In addition to these projective representations of $Z_2\times Z_2$, we have the respective ordinary representations 
\begin{align}
    \mu_1(h)=(-1)^{h^{(1)}},\qquad\mu_2(h)=(-1)^{h^{(2)}},\qquad\mu_{12}(h)=(-1)^{h^{(1)}+h^{(2)}}.
\end{align}
If we label the anyons in $Z_2\times Z_2$ TQD by $e^{(1)}$, $e^{(2)}$, $m^{(1)}$, and $m^{(2)}$, where $e^{(i)}$ denote the elementary charges (chargeons) and $m^{(i)}$ denote the elementary fluxes, charges are given by the ordinary representation of $Z_2\times Z_2$ and the fluxes are given by the projective representations of $Z_2\times Z_2$.

Then the general formula for calculating the braiding phase between anyons $a$ and $b$ is 
\begin{align}
    B(a,b)=\mu_a(\text{flux}(b))\mu_b(\text{flux}(a)).
\end{align}
We list the braiding phases between the various elementary charges and fluxes as follows,
\begin{align}
    \begin{split}
        B(e^{(1)},m^{(1)})&=\mu_1(x)\mu_x(1)=-1,\qquad B(e^{(1)},m^{(2)})=\mu_1(t)\mu_t(1)=1,\\
    B(e^{(2)},m^{(1)})&=\mu_2(x)\mu_x(1)=1,\qquad B(e^{(2)},m^{(2)})=\mu_2(t)\mu_{t}(1)=-1,\\
    B(m^{(1)},m^{(1)})&=\mu_x(x)^2=i^{2p_2}=(-1)^{p_2},\qquad B(m^{(1)},m^{(2)})=\mu_x(t)\mu_t(x)=i^{p_1},\qquad \\
    &B(m^{(2)},m^{(2)})=\mu_t(t)^2=i^{2p_1}=(-1)^{p_1}.
    \end{split}
\end{align}

\section{Fusion rule in SET from gauging $S_3$ SPT}
\label{sec:fusionruleS3 appendix}
We write the element in $S_3$ as $\Tilde{g}=(G,g)\equiv x^Ga^g$. We construct a representative of 3-cocycle in $H^3(S_3,U(1))$ as follows:
\beq
    \omega(\Tilde{g},\Tilde{h},\Tilde{l})  =\exp{\frac{2\pi i p_1}{9}g(-1)^{H+L}(h(-1)^L+l-[h(-1)^L+l]_3)+\pi i p_2 GHL},
\eeq
where $p_1=0,1,2$, and $p_2,=0,1$. Gauging the normal subgroup $Z_3$ of a $S_3$-SPT state results in a state in an SET state that has the same anyon set as $D^{\nu}(Z_3)$, where
\beq
    \nu(g,h,l)=\exp{\frac{2\pi i p_1}{9}g(h+l-[h+l]_3)}
\eeq
is the restriction of $\omega$ on $Z_3$. 
Different values of $p_1$ correspond to different $Z_3$-TQD models. 
The symmetry action takes an anyon to its inverse, which is nontrivial. 
In sector $\mathcal{C}_x$, there is only one object of quantum dimension $3$. 
We name it $0_x$. 
By dimension counting, we can write a fusion rule of the form,
\beq
    0_x \times 0_x = \sum_{i=1}^9 a_i,
\eeq
where $a_i\in \mathcal{C}$ are abelian anyons. Let $b_1=1$, $b_2=a$, and $b_3=a^2$. One can write a matrix-valued operator on an open ribbon as
\beq
(H^x_l)_{ii'}=H_l^{b_i x b_i^{-1},b_i b_{i'}^{-1}}\epsilon_{b_i x b_i^{-1}}(b_i b_{i'}^{-1}),
\eeq
where the matrix indices are $i,i'=1,2,3$.
As we conjectured, the object $0_x\times 0_x$ should be created by the operator $(H^x_l)^{\otimes2}$ on the endpoint of $l$. 
Let $\chi\equiv \exp{\frac{2\pi i }{9}}$, then a calculation similar to that in Sec.~\ref{sec:fusionruled4 appendix} gives us
\beq
    (H^x_l)^{\otimes2}=&\begin{pmatrix}
    H_l^{x,1} & \epsilon_{x}(a^2)H_l^{x,a^2} & \epsilon_{x}(a)H_l^{x,a}\\
    \epsilon_{xa}(a)H_l^{xa,a} & H_l^{xa,1} & \epsilon_{xa}(a^2)H_l^{xa,a^2}\\
    \epsilon_{xa^2}(a^2)H_l^{xa^2,a^2} & \epsilon_{xa^2}(a)H_l^{xa^2,a} & H_l^{xa^2,1}
    \end{pmatrix}^{\otimes2}\\
    =&\begin{pmatrix}
    H_l^{x,1} & \chi^{-2p_1}H_l^{x,a^2} & \chi^{p_1}H_l^{x,a}\\
    \chi^{-p_1}H_l^{xa,a} & H_l^{xa,1} & H_l^{xa,a^2}\\
    \chi^{-4p_1}H_l^{xa^2,a^2} & H_l^{xa^2,a} & H_l^{xa^2,1}
    \end{pmatrix}^{\otimes2}\\
    =&\begin{pmatrix}
    H_l^{1,1} & \chi^{-4p_1}H_l^{1,a^2} & \chi^{2p_1}H_l^{1,a}\\
    \chi^{-5p_1}H_l^{1,a} & H_l^{1,1} & \chi^{-3p_1}H_l^{1,a^2}\\
    \chi^{-2p_1}H_l^{1,a^2} & \chi^{-6p_1}H_l^{1,a} & H_l^{1,1}
    \end{pmatrix}
    \oplus\begin{pmatrix}
    H_l^{a,1} & \chi^{-2p_1}H_l^{a,a^2} & \chi^{-3p_1}H_l^{a,a}\\
    \chi^{-4p_1}H_l^{a,a} & H_l^{a,1} & \chi^{2p_1}H_l^{a,a^2}\\
    \chi^{-3p_1}H_l^{a,a^2} & \chi^{p_1}H_l^{a,a} & H_l^{a,1}
    \end{pmatrix}\\
    &\oplus\begin{pmatrix}
    H_l^{a^2,1} & \chi^{-3p_1}H_l^{a^2,a^2} & \chi^{-2p_1}H_l^{a^2,a}\\
    H_l^{a^2,a} & H_l^{a^2,1} & \chi^{-2p_1}H_l^{a^2,a^2}\\
    \chi^{-p_1}H_l^{a^2,a^2} & \chi^{-p_1}H_l^{a^2,a} & H_l^{a^2,1}
    \end{pmatrix}\\
    =&\begin{pmatrix}
    H_l^{1,1} & \chi^{-4p_1}H_l^{1,a^2} & \chi^{2p_1}H_l^{1,a}\\
    \chi^{-5p_1}H_l^{1,a} & H_l^{1,1} & \chi^{-3p_1}H_l^{1,a^2}\\
    \chi^{-2p_1}H_l^{1,a^2} & \chi^{-6p_1}H_l^{1,a} & H_l^{1,1}
    \end{pmatrix}
    \oplus\begin{pmatrix}
    H_l^{a,1} & \chi^{-4p_1}\epsilon_{a}(a^2)H_l^{a,a^2} & \chi^{-4p_1}\epsilon_{a}(a)H_l^{a,a}\\
    \chi^{-5p_1}\epsilon_{a}(a)H_l^{a,a} & H_l^{a,1} & \epsilon_{a}(a^2)H_l^{a,a^2}\\
    \chi^{-5p_1}\epsilon_{a}(a^2)H_l^{a,a^2} & \epsilon_{a}(a)H_l^{a,a} & H_l^{a,1}
    \end{pmatrix}\\
    &\oplus\begin{pmatrix}
    H_l^{a^2,1} & \chi^{-7p_1}\epsilon_{a^2}(a^2)H_l^{a^2,a^2} & \chi^{-4p_1}\epsilon_{a^2}(a)H_l^{a^2,a}\\
    \chi^{-2p_1}\epsilon_{a^2}(a)H_l^{a^2,a} & H_l^{a^2,1} & \chi^{-6p_1}\epsilon_{a^2}(a^2)H_l^{a^2,a^2}\\
    \chi^{-5p_1}\epsilon_{a^2}(a^2)H_l^{a^2,a^2} & \chi^{-3p_1}\epsilon_{a^2}(a)H_l^{a^2,a} & H_l^{a^2,1}
    \end{pmatrix}\\
    \simeq&(H_l^{1,1}+H_l^{1,a}+H_l^{1,a^2})\oplus(H_l^{1,1}+\chi^3 H_l^{1,a}+\chi^{6}H_l^{1,a^2})\oplus(H_l^{1,1}+\chi^{-3}H_l^{1,a}+\chi^{-6}H_l^{1,a^2})\\
    &\oplus(H_l^{a,1}+\epsilon_a(a)H_l^{a,a}+\epsilon_a(a^2)H_l^{a,a^2})\oplus(H_l^{a,1}+\chi^3\epsilon_a(a)H_l^{a,a}+\chi^{6}\epsilon_a(a^2)H_l^{a,a^2})\\
    &\oplus(H_l^{a,1}+\chi^{-3}\epsilon_a(a)H_l^{a,a}+\chi^{-6}\epsilon_a(a^2)H_l^{a,a^2})\oplus(H_l^{a^2,1}+\epsilon_{a^2}(a)H_l^{a^2,a}+\epsilon_{a^2}(a^2)H_l^{a^2,a^2})\\
    &\oplus(H_l^{a^2,1}+\chi^3\epsilon_{a^2}(a)H_l^{a^2,a}+\chi^{6}\epsilon_{a^2}(a^2)H_l^{a^2,a^2})\oplus(H_l^{a^2,1}+\chi^{-3}\epsilon_{a^2}(a)H_l^{a^2,a}+\chi^{-6}\epsilon_{a^2}(a^2)H_l^{a^2,a^2}).
\eeq
The nine parts in the last line after the decomposition are the nine ribbon operators in $D^{\nu}(Z_3)$ that create the anyons $1, e, e^2, m, em, e^2m, m^2, em^2,$ and $e^2m^2$. Therefore, one can conclude the fusion rule from the above calculation:
\beq
    0_x\times 0_x=1 + e + e^2 + m + em + e^2 m + m^2 + em^2 + e^2m^2.
\eeq

\section{An alternative $N$-step gauging via measurement}\label{N-step}

In this section, we give an alternative N-step gauging procedure. Following the 2-step gauging, we generalize the procedure to $N$-step gauging (a similar method was proposed by \cite{Hierarchy} and \cite{Bravyi2022a} for solvable groups) for a group G that satisfies a criterion. The $N$-steps correspond to the $N$ factors of abelian groups of $G$. \par Let us consider a group $G$ with the following property: there exist a sequence of groups $N_0$,$N_1$, ...$N_n$ abelian and another sequence $M_0$,$M_1$, ... $M_n$ such that 
 \begin{align}
 \begin{split}
     N_0&\equiv \{e\},\qquad  M_0=G,\\
     N_1&\triangleleft G\qquad,\hspace{0.2cm} M_1=\frac{G}{N_1},\\
     N_2&\triangleleft M_1, \qquad M_2=\frac{M_1}{N_2},\\
     &\vdots\\
     N_n&\triangleleft M_{n-1},\quad M_n=\frac{M_{n-1}}{N_n}={e}.
 \end{split}
 \end{align}
If the group $G$ satisfies this property we say it admits a sequential normal subgroups. We will prove in appendix \ref{proof:solvale=seqnormal} that admitting sequential normal subgroups is equivalent to the group being solvable. Given $G$ admits a sequential normal subgroup, we get a sequence of short exact sequences
\begin{subequations}
    \begin{align}
    1\longrightarrow N_1\xrightarrow{\mathcal{i}_1} &G\xrightarrow{\pi_1}M_1\longrightarrow 1,\\
    1\longrightarrow N_2\xrightarrow{\mathcal{i}_2}&M_1\xrightarrow{\pi_2} M_2\longrightarrow 1,\\
    \vdots\\
    1\longrightarrow N_{n-1}\xrightarrow{\mathcal{i}_{n-1}} &M_{n-2}\xrightarrow{\pi_n} M_{n-1}\longrightarrow 1,\\
    1\longrightarrow N_{n}&\xrightarrow{\mathcal{i}_{n}} M_{n-1}\longrightarrow 1.
\end{align}
\end{subequations}
where $\mathcal{i}_k$ is an inclusion map and $\pi_k$ is a projection map. Now we choose a sequence of lifts 
\begin{subequations}
    \begin{align}
    M_1&\xrightarrow{s_1}G,\\
    M_2&\xrightarrow{s_2}M_1,\\
    &\vdots\\
    M_{n-1}&\xrightarrow{s_{n-2}}M_{n},
\end{align}
\end{subequations}
which will be fixed throughout this section. With these lifts we can embed each of the normal groups $N_k$ as a set in $G$. If $n\in N_k$, $s_1(s_2(...s_{k-1}(\mathcal{i}_k(n))))\in G$ is an embedding of $n\in N_k$ in $G$. To simplify notation, we denote $s_1\circ s_2\circ ...\circ s_k$ by $\Tilde{s}_k$. Then $\Tilde{s}_k:M_k\longrightarrow G$.
Using this notation, a general $g\in G$ can be written as 
 \begin{align}
 g=\Tilde{s}_{n-1}(\mathcal{i}_n(a_n^{i_n}))...\Tilde{s}_1(\mathcal{i}_2(a_2^{i_2}))\mathcal{i}_1(a_1^{i_1}),
     \label{repg}
 \end{align}
 where $a_j\in N_j$ is a generator for each abelian normal subgroup. The notation $a_j^{i_j}$ is a shorthand for the product of generators for each cyclic subgroup of the abelian group, i.e, $a_j^{i_j}=\prod_{k=1}^l (a^k_j)^{i^k_j}$  where $a^k_j$ for $k=1,...,l$ are the generators for the $l$ cyclic factors. 
 \begin{claim}
     The representation of $g\in G$ given in Eq.~\ref{repg} is unique. If $g=\Tilde{s}_{n-1}(\mathcal{i}_n(a_n^{i_n}))...\mathcal{i}_1(a_1^{i_1})=\Tilde{s}_{n-1}(\mathcal{i}_n(a_n^{i'_n}))...\mathcal{i}_1(a_1^{i'_1})$, then $i_n=i'_n$,..., $i_1=i'_1$.
 \end{claim}
 \begin{proof}
     The proof follows by applying the projections $\Tilde{\pi}_k:=\pi_1\circ \pi_2\circ...\circ\pi_k$ for $k=1,...,n-1$. First apply $\Tilde{\pi}_{n-1}$ to $g$. This gives $i_n=i'_n$. Then apply $\Tilde{\pi}_{n-2}$ which gives $i_{n-1}=i'_{n-1}$. Proceeding similarly, at $k^{th}$ step apply $\Tilde{\pi}_{n-k}$ to get $i_{n-k+1}=i'_{n-k+1}$. At $(n-1)^{th}$ step we get $i_2=i'_2$. This automatically fixes $i_1=i'_1$ proving the unique representation of $g$.
 \end{proof}
 To simplify the notation in the remaining part of this section, we omit writing the lifts explicitly and write $a_k^{i_k}\equiv \Tilde{s}_{k-1}(\mathcal{i}_k(a_k^{i_k}))$. Hence
 \begin{align}
     g=a_n^{i_n}...a_1^{i_1}.
 \end{align}
 Let $h\in G$. Similarly, we can write $h=a_n^{\bar{i}_n}... a_1^{\bar{i}_1}$. Using this notation, we can write down the group multiplication as
 \begin{align}
      gh^{-1}=a_n^{i_n}...a_1^{i_1}a_1^{-\bar{i}_1}...a_n^{-\bar{i}_n}.
      \label{ghinverse}
 \end{align}

 We will use Eq.~(\ref{ghinverse}) to implement the $N$-step gauging procedure.  
 We will gauge the $G$ DOFs on the vertices of the lattice sequentially in $N$-steps. The complete procedure for $N$-step gauging is as follows
 \begin{enumerate}
	\item[(1)] 
	\emph{Include ancillas.} Add ancillas in the state $\ket{e}$, where $e\in G$ is the identity element, on the edges between the vertices. 
	\item[(2)]
	\emph{Entangle gauge and matter DOFs.} Apply the following 2 controlled-shift operators with controls $c_1, c_2$ on neighboring vertices (oriented as $c_2\rightarrow c_1$) and target $t$ on the in-between ancilla,
	\begin{align}
	   \begin{split}
	      U_{N_1}=\sum_{g_1,g_2\in G}\sum_{g_3\in G }&\ket{g_1,g_2}_{c_1,c_2}\bra{g_1,g_2}\otimes \ket{N_1(g_1)g_3 N_1(g_2)^{-1}}_e\bra{g_3}. 
	   \end{split} 
	\end{align}
	Here $N_1(g)$ is the part of the decomposition $g$ which lies in $N_1$; when $g=a_n^{i_n}...a_1^{i_1}$ with $a_j\in N_j$, then $N_1(g)=a_1^{i_1}$.
    	\item[(3)]
	\emph{Measure $X_{N_1}$ on matter DOFs and correct the $z_{N_1}$ factors.} Define: $X_{N_1}=\{ X_1,...,X_{n_{m_1}}\}$ for $N_1=\prod_{k=1}^ {m_1}\mathbb{Z}_{n_k}$ where $X_{j}$ denote $j\times j$ Pauli $X$ matrix. Following the same define: $z_{N_1}=\prod_{k=1}^{m_1} z_{n_k}$ where $z_{n_k}$ is the phase factor coming from measuring $X_{n_k}$ on the vertex. The value of $z_{n_k}$ is same as acting Pauli $Z_{n_k}$ operator on the vertex before measurement. After measurement, with the outcome being $X_{N_1}=\{\omega_1^{-p_1},...,\omega_1^{-p_{m_1}}\}$ ($\omega_k$ being $n_k$th root of unity), there is a corresponding phase factor $\prod_{k=1}^{m_1}  z_{n_k}^{-p_k}$. Using the transmutation rule for each of the phase terms in the product $z_{N_1}$,
 \begin{align}
     Z_{n_k}(N_1(g_2))Z_{n_k}(N_1(g_1)N_1(g_2)^{-1})=Z_{n_k}(N_1(g_1)).
 \end{align}
 one can correct all those factors by moving them to a single vertex, resulting in an $M_1$ SET ground state.

    \item[(4)]
    \emph{Repeat the procedure of entangling gauge and matter DOFs for $M_1$ DOFs on the vertices.} Apply the following unitary as before,
	\begin{align}
	   \begin{split}
	     U_{N_2}=\sum_{g_1,g_2\in M_1}\sum_{g_3\in G}&\ket{g_1,g_2}_{c_1,c_2}\bra{g_1,g_2}\otimes\ket{N_2(g_1)g_3N_2(g_2)^{-1}}_e\bra{g_3}.  
	   \end{split} 
	\end{align}
	\item[(5)]
	\emph{Measure $X_{N_2}$ on matter DOFs and correct the $z_{N_2}$ factors.} Define: $X_{N_2}=\{ X_1,...,X_{n_{m_2}}\}$ for $N_2=\prod_{k=1}^ {m_2}\mathbb{Z}_{n_k}$ where $X_{j}$ denote $j\times j$ Pauli $X$ matrix. Following the same define: $z_{N_2}=\prod_{k=1}^{m_2} z_{n_k}$ where $z_{n_k}$ is the phase factor coming from measuring $X_{n_k}$ on the vertex. The value of $z_{n_k}$ is same as acting Pauli $Z_{n_k}$ operator on the vertex before measurement. After measurement, with the outcome being $X_{N_2}=\{\omega_1^{-p_1},...,\omega_1^{-p_{m_2}}\}$ ($\omega_k$ being $n_k$th root of unity), there is a corresponding phase factor $\prod_{k=1}^{m_2}  z_{n_k}^{-p_k}$. Using the transmutation rule for each of the phase terms in the product  $z_{N_2}$,
 \begin{align}
     Z_{n_k}(N_2(g_2))Z_{n_k}(N_2(g_1)N_2(g_2)^{-1}\sigma^{N_2(g_2)}(g_e))=Z_{n_k}(N_2(g_1)).
 \end{align}
 one can correct all those factors by moving them to a single vertex, resulting in an $M_2$ SET ground state. (Note that $Z_{n_k}(\sigma^{N_2(g_2)}(g_3))=1$. This is because $\sigma^{N_2(g_2)}(g_3)\in N_1$ for $g_3\in N_1$ and hence has no component in $N_2$.)
    \item[(6)]
    \emph{At $l^{th}$ step of gauging process, again repeat the procedure of entangling gauge and matter DOFs for $M_{l-1}$ DOFs on the vertices.}Apply the following unitary
    \begin{align}
	   \begin{split}
	     U_{N_l}=\sum_{g_1,g_2\in M_{l-1}}\sum_{g_3\in G}&\ket{g_1,g_2}_{c_1,c_2}\bra{g_1,g_2}\otimes\ket{N_l(g_1)g_3N_l(g_2)^{-1}}_e\bra{g_3}.  
	   \end{split} 
    \end{align}
    \item[(7)]
    \emph{Measure $X_{N_l}$ on matter DOFs and correct the $z_{N_l}$ factors.}  Define: $X_{N_l}=\{ X_1,...,X_{n_{m_l}}\}$ for $N_l=\prod_{k=1}^ {m_l}\mathbb{Z}_{n_k}$ where $X_{j}$ denote $j\times j$ Pauli $X$ matrix. Following the same define: $z_{N_l}=\prod_{k=1}^{m_l} z_{n_k}$ where $z_{n_k}$ is the phase factor coming from measuring $X_{n_k}$ on the vertex. The value of $z_{n_k}$ is same as acting Pauli $Z_{n_k}$ operator on the vertex before measurement. After measurement, with the outcome being $X_{N_l}=\{\omega_1^{-p_1},...,\omega_1^{-p_{m_l}}\}$ ($\omega_k$ being $n_k$th root of unity), there is a corresponding phase factor $\prod_{k=1}^{m_l}  z_{n_k}^{-p_k}$. Using the transmutation rule for each of the phase terms in the product  $z_{N_l}$,
 \begin{align}
     Z_{n_k}(N_l(g_2))Z_{n_k}(N_l(g_1)N_l(g_2)^{-1}\sigma^{N_l(g_2)}(g_e))=Z_{n_k}(N_l(g_1)),
 \end{align}
 one can correct all those factors by moving them to a single vertex, resulting in an $M_l$ SET ground state. (Note that $Z_{n_k}(\sigma^{N_l(g_2)}(g_3))=1$. This is because $\sigma^{N_l(g_2)}(g_3)\in (...(N_1\otimes N_2)\otimes ....\otimes N_{l-1})$ for $g_3\in (...(N_1\otimes N_2)\otimes...\otimes N_{l-1})$ and hence has no component in $N_l$. The notation $G\otimes H$ is a short hand for group extension of $G$ by $H$.)
    \item[(8)]
    \emph{Repeat this process for all till the last normal subgroup $N_n$.}
 \end{enumerate}
 The groups formed from extensions are
     \begin{align}
     G_l\equiv((N_1\otimes N_2)\otimes N_3)...\otimes N_l&\equiv\{\Tilde{s}_{l-1}(a_l^{i_l})...\Tilde{s}_1(a_2^{i_2})i(a_1^{i_1})|a_l^{i_l}\in N_l, ..., a_1^{i_1}\in N_1\}.
 \end{align}
 \begin{claim}
     $G_l$ \text{ is a subgroup of } $G$.
 \end{claim}
 \begin{proof}
     We prove this by induction. First we prove $G_2$ is subgroup of $G$. $G_2=N_1\otimes N_2$. Suppose $g_1=s(q_1)i(n_1)\in G_2$ and $g_2=s(q_2)i(n_2)\in G_2$, then $g_1g_2=s(q_1)i(n_1)s(q_2)i(n_2)=s(q_1)s(q_2)s(q_2)^{-1}(i(n_1))s(q_2)i(n_2)=s(q_1q_2) \omega(q_1,q_2)\sigma^{q_2^{-1}}(i(n_1))i(n_2)$. Hence $g_1g_2\in G_2$. Assume $G_{l-1}$ is a subgroup of $G$. We prove $G_l$ is a subgroup of $G$. Suppose $g_1=\Tilde{s}_{l-1}(a_l^{\Bar{i}_l})...\Tilde{s}_1(a_2^{\Bar{i}_2})i(a_1^{\Bar{i}_1})\in G_l$, $g_2=\Tilde{s}_{l-1}(a_l^{i_l})...\Tilde{s}_1(a_2^{i_2})i(a_1^{i_1})\in G_l$. 
     Let us define
     \begin{align}
     \begin{split}
     b_1&=\Tilde{s}_{l-1}(a_l^{i_l})...\Tilde{s}_1(a_2^{i_2})\\
     b_2&=\Tilde{s}_{l-1}(a_l^{i_l})...\Tilde{s}_2(a_3^{i_3})\\
     &\vdots\\
     b_{l}&=\Tilde{s}_{l-1}(a_l^{i_l})
     \end{split}
     \end{align}
     Then 
     \begin{align}
         g_1g_2=\Tilde{s}_{l-1}(a_l^{\Bar{i}_l})\Tilde{s}_{l-1}(a_l^{i_l})\sigma^{b_l^{-1}}(\Tilde{s}_{l-2}(a_{l-1}^{\bar{i}_{l-1}}))\Tilde{s}_{l-2}(a_{l-1}^{i_l})...\sigma^{b_1^{-1}}(i(a_1^{\Bar{i}_1}))i(a_1^{i_1}).
         \end{align}
 Only the first two terms in the above equation lies in $N_l$. The remaining terms lie in $G_{l-1}$. Hence they can be expressed as
 \begin{align}
     \sigma^{b_l^{-1}}(\Tilde{s}_{l-2}(a_{l-1}^{\bar{i}_{l-1}}))\Tilde{s}_{l-2}(a_{l-1}^{i_l})...\sigma^{b_1^{-1}}(i(a_1^{\Bar{i}_1}))i(a_1^{i_1})=\Tilde{s}_{l-2}(a_{l-1}^{i'_{l-1}})....i(a_1^{i'_1}).
     \label{termsinG_{l-1}}
 \end{align}
 Now we prove that $\Tilde{s}_{l-1}(a_l^{\Bar{i}_l})\Tilde{s}_{l-1}(a_l^{i_l})=\Tilde{s}_{l-1}(a_l^{\Bar{i}_l+i_l})h$ where $h\in G_{l-1}$. Then by induction hypothesis, 
 \begin{align}
     g_1g_2=\Tilde{s}_{l-1}(a_l^{\Bar{i}_l+i_l})\Tilde{s}_{l-2}(a_{l-1}^{i''_{l-1}})...\Tilde{s}_1(a_2^{i''_2})i(a_1^{i''_1}).
     \label{g_1g_2}
 \end{align}
 which shows that $g_1g_2\in G_l$. First note that applying the relation $s(a)s(b)=s(ab)\omega(a,b)$, we get
 \begin{align}
 \begin{split}
     \Tilde{s}_{l-1}(a_l^{\Bar{i}_l})\Tilde{s}_{l-1}(a_l^{i_l})&=s\left(\Tilde{s}_{l-2}(a_l^{\Bar{i}_l})\Tilde{s}_{l-2}(a_l^{i_l})\right)\omega(\Tilde{s}_{l-2}(a_l^{\Bar{i}_l}),\Tilde{s}_{l-2}(a_l^{i_l})),\\
     &=s\left(s\left(\Tilde{s}_{l-3}(a_l^{\Bar{i}_l})\Tilde{s}_{l-3}(a_l^{i_l})\right)\omega(\Tilde{s}_{l-3}(a_l^{\Bar{i}_l}),\Tilde{s}_{l-3}(a_l^{i_l}))\right)\omega(\Tilde{s}_{l-2}(a_l^{\Bar{i}_l}),\Tilde{s}_{l-2}(a_l^{i_l})),\\
     &\qquad\vdots
     \label{alproduct}
 \end{split}
 \end{align}
 One can write this equation in short hand as
 \begin{align}
 \begin{split}
     \Tilde{s}_{l-1}(a_l^{\Bar{i}_l})\Tilde{s}_{l-1}(a_l^{i_l})&=s(M_1)N_1,\\
     &=s(s(M_2)N_2)N_1,\\
     &=s(s(s(M_3)N_3)N_2)N_1,\\
     &\qquad\vdots\\
     &=s(...(s(M_{l-1})N_{l-1})...)N_1,
 \end{split}
 \end{align}
 where $N_l$ denote the terms coming from $\omega$ factors in Eq.~\ref{alproduct} and $M_r$ denote $\Tilde{s}_{l-1-r}(a_l^{\Bar{i}_l})\Tilde{s}_{l-1-r}(a_l^{i_l})$, $M_{l-1}$ denote $a_l^{\Bar{i}_l+i_l}$. Now applying the relation $s(a)s(b)\omega(a,b)^{-1}=s(ab)$, we get
 \begin{align}
     \Tilde{s}_{l-1}(a_l^{\Bar{i}_l+i_l})\times \text{ (terms in $G_{l-1}$)}.
 \end{align}
 By induction hypothesis, terms in $G_{l-1}$ can be written as $\Tilde{s}_{l-2}(a_{l-1}^{\Tilde{i}_{l-1}})\Tilde{s}_{l-3}(a_{l-2}^{\Tilde{i}_{l-2}})...\Tilde{s}_{1}(a_{1}^{\Tilde{i}_{1}})i(a_1^{\Tilde{i}_1})$. Combining the terms in Eq.~\ref{termsinG_{l-1}}, we get the desired decomposition of $g_1g_2$ as in Eq.~\ref{g_1g_2} which prove $g_1g_2\in G_l$. One can show if $g\in G_l$, $g^{-1}\in G_l$ by applying $\Tilde{\pi}_{l-1}$, $\Tilde{\pi}_{l-2}$,... upto $\pi$ on $g^{-1}$. This will give the explicit decomposition of $g^{-1}$ using each of the normal subgroups $N_i$ for $i=1,...,l$.
\end{proof}
 \begin{claim}
     $G_l$ \text{ is normal in } $G$.
 \end{claim}
\begin{proof}
    Let $g\in G$. $g=\Tilde{s}_{n-1}(a_n^{i_n})... \Tilde{s}_1(a_2^{i_2})i(a_1^{i_1})$. If $k\in G_l$, let $k=\Tilde{s}_{l-1}(a_l^{\Bar{i}_l})...\Tilde{s}_1(a_2^{\Bar{i}_2})i(a_1^{\Bar{i}_1})$. Then $gkg^{-1}=\Tilde{s}_{n-1}(a_n^{i_n})... \Tilde{s}_1(a_2^{i_2})i(a_1^{i_1})\Tilde{s}_{l-1}(a_l^{\Bar{i}_l})...\Tilde{s}_1(a_2^{\Bar{i}_2})i(a_1^{\Bar{i}_1})i(a_1^{i_1})^{-1}\Tilde{s}_1(a_2^{i_2})^{-1}...\Tilde{s}_{n-1}(a_n^{i_n})^{-1}$. Applying $\Tilde{\pi}_i$ for $i=n-1,...,l$, we see that $\Tilde{\pi}_i(gkg^{-1})=1$. This shows that $gkg^{-1}\in G_l$. Hence $G_l$ is normal in $G$.
\end{proof}

 \section{Proof of solvable equivalent to admitting sequential normal abelian subgroups}\label{proof:solvale=seqnormal}
 In this section, we prove that the assumption about the group $G$ we used in section \ref{N-step} is equivalent to the assumption that $G$ is a solvable group.
 \begin{definition}
A derived series of a finite group $G$ is a sequence of normal subgroups normal inside the previous one
\begin{align}
    G\vartriangleright G^{(1)}\vartriangleright G^{(2)}\vartriangleright G^{(3)}....\vartriangleright G^{(n)}\vartriangleright e,
\end{align}
for some $n$ such that the quotient groups $G/G^{(1)}$, $G^{(1)}/G^{(2)}$,...,$G^{(n)}/e$ are all abelian.
\end{definition}
\begin{definition}
A finite group $G$ is solvable if it admits a derived series.
\end{definition}
\begin{proposition}
If $G$ is solvable then it admits the following sequence
\begin{align}
    G\vartriangleright G^1\vartriangleright G^2\vartriangleright...\vartriangleright G^{m}\vartriangleright e,
\end{align}
  for some $m$ which is called the derived length of the group $G$. Here $G^{i+1}=[G^i,G^i]$ is the commutator subgroup of $G^i$.
\label{prop}
\end{proposition}

\begin{proof}
Note that the commutator subgroup of a group $G$ is the smallest normal subgroup in $G$ such that $G/[G,G]$ is abelian. Hence we have $G^1\subset G^{(1)}$. Now we have $G^2=[G^1,G^1]\subset [G^{(1)},G^{(1)}]\subset G^{(2)}$. Inductively we can assume that $G^k\subset G^{(k)}$. Then $G^{k+1}=[G^k,G^k]\subset [G^{(k)},G^{(k)}]\subset G^{(k+1)}$. Hence $G^k\subset G^{(k)}$ $\forall k\in \{1,2,...,m\}$. This clearly says that the sequence of commutator groups doesn't terminate. If it would have terminated at $G^{k+1}$, then $G^{k+2}=[G^{k+1},G^{k+1}]=G^{k+1}$. One can repeat this to argue $G^{k+1}\subset G^{(k+1)}$, $G^{k+1}\subset G^{(k+2)}$, ... $G^{k+1}\subset G^{(n)}$. But $G^{k+1}$ is nonabelian group and $G^{(n)}$ is abelian. Hence we can't have $G^{k+1}\subset G^{(n)}$, contradiction. So the sequence of commutator subgroups doesn't terminate and we have the sequence.
\end{proof}
\begin{definition}
We say a finite group $G$ admits a sequential normal subgroups if it satisfies the following property:
\begin{align}
\begin{split}
    N_1\vartriangleleft G, \text{$N_1$ abelian}&\quad M_1=G/N_1,\\
    N_2\vartriangleleft M_1, \text{$N_2$ abelian}&\quad M_2=M_1/N_2,\\
    &\vdots\\
    N_n\vartriangleleft M_{n-1}, \text{$N_n$ abelian}&\quad M_n=M_{n-1}/N_n=e.
\end{split}
\end{align}
\end{definition}
\begin{claim}
Suppose the finite group $G$ is solvable then it admits sequential normal subgroups.
\end{claim}
\begin{proof}
From proposition \ref{prop} we see that $G$ admits the sequence
\begin{align}
    G\vartriangleright G^1\vartriangleright G^2\vartriangleright...\vartriangleright G^{m}\vartriangleright e.
\end{align}
where $G^{i+1}=[G^i,G^i]$ is the commutator subgroup. First we prove that $G^k\vartriangleleft G$ $\forall k\in \{1,2,...,m\}$. This we prove by induction on $k$. Clearly, $G^1\vartriangleleft G$. Assuming $G^{k-1}\vartriangleleft G$, we need to prove $G^k\vartriangleleft G$. $G^k=[G^{k-1},G^{k-1}]$. Hence, $G^k$ is generated by elements of the form $ghg^{-1}h^{-1}$ where $g,h\in G^{k-1}$. Now one can write $kghg^{-1}h^{-1}k^{-1}=kgk^{-1}khk^{-1}kg^{-1}k^{-1}kh^{-1}k^{-1}$. Since $G^{k-1}\vartriangleleft G$, $kgk^{-1}\in G^k$ $\forall g\in G^k$ and $k\in G$. So $kgk^{-1}khk^{-1}kg^{-1}k^{-1}kh^{-1}k^{-1}\in [G^{k-1},G^{k-1}]=G^k$. So we see that $G^k\vartriangleleft G$.
\par Now we choose
\begin{align}
\begin{split}
     N_1=G^m\vartriangleleft G, \text{$N_1$ abelian}&\quad M_1=G/N_1,\\
    N_2=G^{m-1}/G^m\vartriangleleft G/G^m, \text{$N_2$ abelian}&\quad M_2=M_1/N_2,\\
    &\vdots\\
    N_{m+1}=G/G^1,  \text{$N_{m+1}$ abelian}&\quad M_{m+1}=M_m/N_{m+1}=e.
\end{split}
\end{align}
\end{proof}
\begin{claim}
If the finite group $G$ admits a sequential normal subgroups then $G$ is solvable.
\end{claim}
\begin{proof}
Suppose $G$ is not solvable, assuming it admits sequential normal subgroups. Then the derived series of commutator subgroup terminate
\begin{align}
    G\vartriangleright G^1\vartriangleright G^2\vartriangleright ...\vartriangleright G^k
\end{align}
for some $k$, where $G^{i+1}=[G^i,G^i]$.
Now one could consider the series
\begin{align}
\begin{split}
    M^1\vartriangleright &M^{1(1)}\vartriangleright M^{1(2)}\vartriangleright ... \\
    M^2\vartriangleright &M^{2(1)}\vartriangleright M^{2(2)}\vartriangleright ... \\
    &\vdots\\
    M^l\vartriangleright &M^{l(1)}\vartriangleright M^{l(2)}\vartriangleright ... \\
    M^{n-1}\vartriangleright &M^{(n-1)(1)}\vartriangleright M^{(n-1)(2)}\vartriangleright ... 
\end{split}
\end{align}
where $M^{j(l+1)}=[M^{j(l)},M^{j(l)}]$ is the commutator subgroup and $M^{j(0)}\equiv M^j$. Now let us look at the following proposition.
\begin{proposition}
If the derived series of commutator subgroups of $G$ terminates then so does for $M^1$, $M^2$,..., $M^{n-1}$.
\end{proposition}
\begin{proof}
Consider $[M^1,M^1]=[G/N_1,G/N_1]$. It is generated by $gN_1 g'N_1g^{-1}N_1g'^{-1}N_1=gg'g^{-1}g'^{-1}N_1$. We know that $gg'g^{-1}g'^{-1}\in G^1$. However, $gN_1=N_1$ if and only if $g\in N_1\cap G^1$. Hence, we find $[G/N_1,G/N_1]\cong G^1/(N_1\cap G^1)$. Repeating this we find $\left[G^1/(N_1\cap G^1),G^1/(N_1\cap G^1)\right]=G^2/(N_1\cap G^2)$ and so on. Hence, the derived series for commutator subgroups for $M^1$ is given by
\begin{align}
    G/N_1\vartriangleright G^1/(N_1\cap G^1)\vartriangleright G^2/(N_1\cap G^2)\vartriangleright ... \vartriangleright G^k/(N_1\cap G^k).
\end{align}
This terminates since $[G^k,G^k]=G^k$ and hence $\left[G^k/(N_1\cap G^k),G^k/(N_1\cap G^k)\right]=G^k/(N_1\cap G^k)$. A similar argument shows that all other derived series terminates.
\end{proof}
Now we have the following terminating derived series
\begin{align}
\begin{split}
    G\vartriangleright G^1\vartriangleright ...\vartriangleright G^k,\\
    G/N_1\vartriangleright G^1/(N_1\cap G^1)\vartriangleright &... \vartriangleright G^k/(N_1\cap G^k),\\
    M^1/N_2\vartriangleright M^{1(1)}/(N_2\cap M^{1(1)})\vartriangleright &... \vartriangleright M^{1(k)}/(N_2\cap M^{1(k)}),\\
    &\vdots\\
    M^l/N_{l+1}\vartriangleright M^{l(1)}/(N_{l+1}\cap M^{l(1)})\vartriangleright &... \vartriangleright M^{l(k)}/(N_{l+1}\cap M^{l(k)}),\\
    &\vdots\\
    M^{n-2}/N_{n-1}\vartriangleright M^{(n-2)(1)}/(N_{n-1}\cap M^{(n-2)(1)})\vartriangleright &... \vartriangleright M^{(n-2)(k)}/(N_{n-1}\cap M^{(n-2)(k)}).
\end{split}
\end{align}
Since the last series is $M^{n-1}=N^n\vartriangleright e$, it terminates in length 1. This is a contradiction. Hence $G$ is solvable.

\end{proof}

\end{widetext}
\end{document}